\documentclass[12pt]{article}

\usepackage{amsmath,amsfonts,amsbsy,amssymb,amsthm}

\usepackage[english]{babel}

\usepackage[parfill]{parskip}
\usepackage{fullpage}
\usepackage{dsfont}
\usepackage{verbatim}
\usepackage{mathtools}
\usepackage{xy}
\usepackage{hyperref}
\usepackage{changepage}
\usepackage{enumerate}
\usepackage{enumitem}
\usepackage{cancel}
\usepackage{float}
\usepackage{subcaption}
\usepackage[dvipsnames]{xcolor}
\usepackage[margin=1cm]{caption}

\setlist{leftmargin=*}

\numberwithin{equation}{section}

\makeatletter
\newtheoremstyle{corsivo}
{\medskipamount}{\medskipamount}%
{\itshape}{}%
{\bfseries}{}%
{ }
{\thmname{#1}\thmnumber{\@ifnotempty{#1}{ }\@upn{#2}}%
	\thmnote{ {\bfseries(#3)}}.}%
\makeatother

\theoremstyle{corsivo}
\newtheorem{thm}{Theorem}[section]
\newtheorem{lemma}[thm]{Lemma}
\newtheorem{crl}[thm]{Corollary}
\newtheorem{prop}[thm]{Proposition}

\makeatletter
\newtheoremstyle{dritto}
{\medskipamount}{\medskipamount}%
{\rmfamily}{}%
{\bfseries}{}%
{ }
{\thmname{#1}\thmnumber{\@ifnotempty{#1}{ }\@upn{#2}}%
	\thmnote{ {\bfseries(#3)}}.}%
\makeatother

\theoremstyle{dritto}

\newtheorem{dfn}[thm]{Definition}
\newtheorem{rmk}[thm]{Remark}

\newcommand{\ga}{\gamma}
\newcommand{\Ga}{\Gamma}

\newcommand{\eps}{\varepsilon}

\newcommand{\ph}{\varphi}
\newcommand{\la}{\lambda}  
\newcommand{\La}{\Lambda}

\newcommand{\Id}{\mathds{1}} 

\newcommand{\eu}{\mathrm{e}}
\newcommand{\iu}{\mathrm{i}}  
\newcommand{\di}{\mathrm{d}}

\newcommand{\N}{\mathbb{N}}
\newcommand{\Z}{\mathbb{Z}}
\newcommand{\Q}{\mathbb{Q}}
\newcommand{\R}{\mathbb{R}}
\newcommand{\C}{\mathbb{C}}
\newcommand{\E}{\mathbb{E}}

\newcommand{\D}{\mathfrak{D}}

\newcommand{\Ci}{\mathcal{C}}
\newcommand{\Hi}{\mathcal{H}}
\newcommand{\U}{\mathcal{U}}
\newcommand{\G}{\mathcal{G}}

\newcommand{\norm}[1]{\left\| #1 \right\|}

\newcommand{\bra}[1]{\left\langle #1 \right|}
\newcommand{\ket}[1]{\left| #1 \right\rangle}

\newcommand{\set}[1]{ \left\{  #1 \right\}} 

\DeclareMathOperator{\Tr}{Tr}

\DeclareMathOperator{\Span}{Span}

\DeclareMathOperator{\supp}{supp} 

\DeclareMathOperator{\dist}{dist}

\newcommand{\ie}{{\sl i.\,e.\ }} 
\newcommand{\eg}{{\sl e.\,g.\ }}

\newcommand{\virg}[1]{``#1''}

\begin{document}

\title{Dynamical delocalization in \\
disordered 2D Chern insulators}
\date{}
\author{Gianluca Panati, Constanza Rojas-Molina, Vincenzo Rossi}
\maketitle

\begin{abstract}
We show the existence of energies exhibiting dynamical delocalization in discrete 
2D Chern insulators perturbed by a random potential in a general setting.  Our proof exploits two main features of the model: jumps in the integer value of the Chern character and continuity of averaged spectral projections in both energy and disorder parameters. This allows us to show robustness of the topological index in the presence of disorder, which, combined with existing methods to prove dynamical localization, allows us to provide detailed information on the phase diagram of the model. The novelty of our approach is that we are able to show dynamical delocalization in the disorder parameter, and not only in the energy parameter, which allows to prove Anderson metal-insulator transition even when spectral gaps close due to the strength of disorder.
\end{abstract}

\section*{Introduction}

Topological phases of quantum matter have become a focal point of modern condensed matter physics, driven by both their fundamental conceptual relevance and their potential for  technological applications 
\cite{Hasan Kane 2010, Qi Zhang 2011, BernevigHughes book}. 

In the idealized case of periodic systems, different phases are distinguished by the values of a topological invariant associated with the space of occupied states, 
which - when decomposed with respect to quasi-momentum - forms the Bloch bundle. \\
For 2D systems belonging to class A in the Kitaev (or Altland-Zirnbauer) symmetry classification \cite{Kitaev 2009}, the relevant topological index is provided by the Chern number of the Bloch bundle which equals, up to a universal constant, the transverse conductivity. In view of that, topologically non-trivial systems in that class are usually called \emph{Chern insulators}, or alternatively \emph{Quantum anomalous Hall insulators}, as they exhibit a non-zero transverse conductivity in absence of any external magnetic field. The first model for such systems was introduced by Haldane in  \cite{Haldane 1988},  
a work which pioneered  the investigation of Chern insulators and, more generally, of topological insulators.

Some topological indices - dubbed as \virg{strong} - are expected to be robust w.r.t.\ moderate disorder, which represents impurities, vacancies, or thermal fluctuations. From a mathematical perspective, this robustness is non-trivial since disorder breaks the translation invariance required for the standard Bloch-Floquet construction, thus requiring a completely different approach to the problem. Consequently, the topological nature of the state is no longer captured by the geometry of a bundle over the Brillouin zone, but rather by the algebraic and spectral properties of the Fermi projector itself.

In the mathematical literature, this fact has been clearly understood in the similar setting of Quantum Hall systems \cite{Thouless Kohmoto Nightingale de Nijs 1982}, see the review \cite{Graf 2007}. A variety of results has been obtained by $C^*$-algebraic approaches, in connection with non-commutative geometry \cite{Bellissard 1986, Bellissard van Elst Schulz-Baldes 1994, Richter Schulz-Baldes 2001} or with crossed product algebras, see \cite{Schulz-Baldes Kellendonk Richter 2002, Kellendonk Schulz-Baldes 2004} where the authors prove bulk-edge correspondence in discrete and continuum disordered systems. We refer to \cite{Prodan Schulz-Baldes 2016} for a comprehensive review.  An alternative research approach has been based on analytic methods, \ie random Schr\"{o}dinger operators \cite{Aizenman Graf 1998, Combes Germinet 2004, Elgart Graf Schenker 2005, Bouclet Germinet Klein Schenker 2005, Germinet Klein Schenker 2007, Taarabt 2014, Becker Han 2022, Becker Oltman Vogel 2025}  providing detailed estimates on the (de)localization behavior. 

Our purpose is two-fold: to show the existence of delocalized energies in the spectrum of a 2D disordered Chern insulator in both the energy and the disorder parameter, even when the spectral gap closes, and to give a rigorous description of the metal insulator transition in disordered discrete 2D Chern insulators, in an axiomatic way, that recovers previous results known for disordered operators in the presence of magnetic fields. 

We first detail the regimes of dynamical localization that can be proven with existing techniques (Fractional Moment Method \cite{Aizenman Molchanov 1993} and Multiscale Analysis \cite{Frohlich Spencer 1983}): strong disorder regime, energies away from the unperturbed spectrum, and (external and internal) band edge localization, specifying the minimal assumptions required on the model for the methods to be applicable. While these results are not necessarily new, we believe that a detailed and clear exposition of all techniques available and how they apply to 2D Chern insulators is valuable for further work on these models.

Our main novel results, Theorems \ref{DD disorder} and \ref{curve}, contribute to the understanding of the localization-delocalization transition for 2D Chern insulators, conjectured in \cite{Tan 1994}, based on numerical simulations, and studied in \cite{Bellissard van Elst Schulz-Baldes 1994, Richter Schulz-Baldes 2001} in the similar setting of Quantum Hall insulators. It is expected that there exist continuous lines in the energy-disorder plane where the states are delocalized and which divide the plane in regions of localization in which the Chern character is constant. We are able to prove that any continuous curve in the energy-disorder plane connecting two  regions of dynamical localization in the perturbed spectrum corresponding to different Chern numbers, must contain  at least one point for which the corresponding energy is dynamically delocalized, see Fig. \ref{fig 1}. 

\begin{figure}[ht]
    \begin{subfigure}{\textwidth}
    \centering
    \includegraphics[scale=0.3]{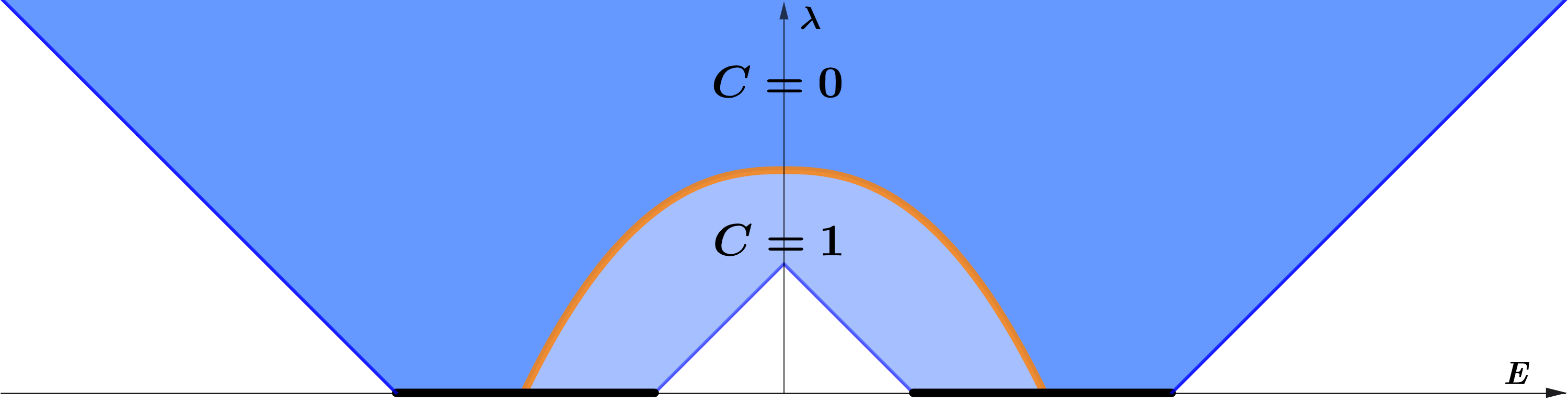}
    \caption{}
    \label{subfig:a}
    \end{subfigure}
    \\[0.2cm]
    \begin{subfigure}{\textwidth}
    \centering
    \includegraphics[scale=0.3]{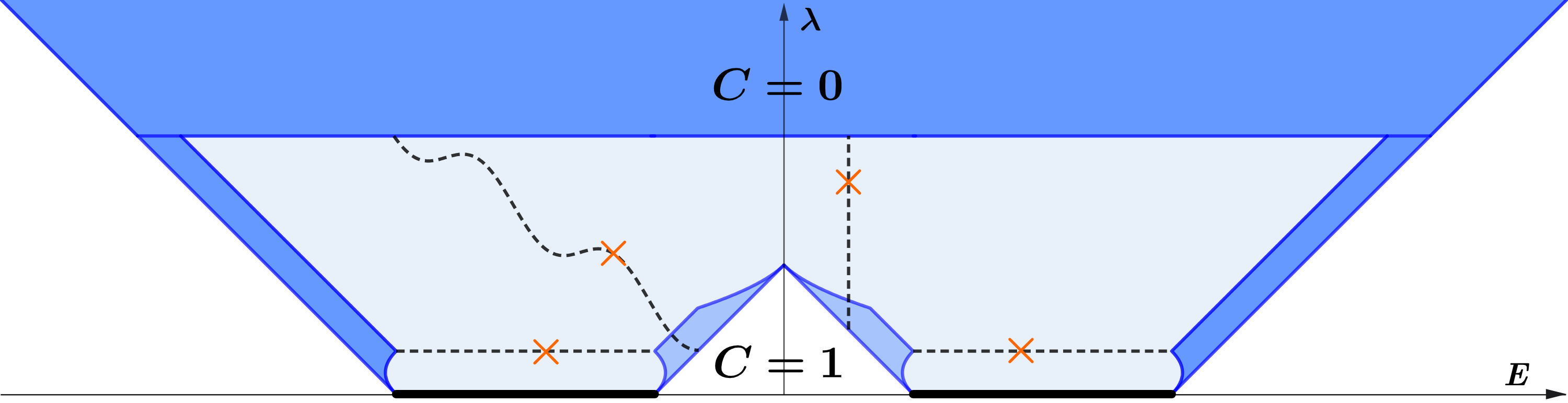}
    \caption{}
    \label{subfig:b}
    \end{subfigure}
    \caption{\small{Qualitative dynamical (de)localization phase diagrams in energy-disorder parameters for the Haldane-Anderson model.\\ Panel (a): conjectured phase diagram. A continuous line (orange), where the metal-insulator transition happens, divides the plane in connected regions of dynamical localization, where the Chern character is constant.\\ Panel (b): partial phase diagram proved in this article. In dark and light blue, the regions of dynamical localization, which enclose the region where the metal-insulator transition may happen. The orange crosses indicate the possible location of the pairs $(E,\lambda)$ for which dynamical delocalization holds, illustrating our main results, Theorems \ref{DD disorder} and \ref{curve}.}}\label{fig 1}
\end{figure}

We illustrate our findings with a thorough analysis of the Haldane-Anderson model, that is, the Haldane model perturbed by a random potential with independent and identically distributed random variables modeling the impurities in the system \cite{Anderson 1958}. 

Comparing with the previous literature, the constancy of the Chern character w.r.t.\ the disorder parameter was first proved in \cite{Richter Schulz-Baldes 2001}. Their proof is based on the homotopic invariance of the Chern character in a $C^*$-algebraic context \cite{Bellissard van Elst Schulz-Baldes 1994}, and rely on the assumption of exponentially decaying Fractional Moments Bounds, that is, a regime of localization. In this article we weaken the latter assumption, as described in Theorem \ref{continuity chern disorder},  by using the continuity in energy and disorder of the averaged spectral projection, which we prove by combining techniques in \cite{Germinet Klein Schenker 2009} and \cite{Krishna 2007}. Differently from \cite{Richter Schulz-Baldes 2001}, this continuity result does not rely on localization assumptions, but only on the regularity of the single-site probability distribution. Our proof exploits analytic methods, in the spirit of \cite{Aizenman Graf 1998}, and we implement techniques coming from the Multiscale Analysis, developed in \cite{Germinet Klein 2001, Germinet Klein 2004}, in order to obtain quantitative estimates on the behavior of delocalized states (see Proposition~\ref{rmk characterization}).

We expect that our results will contribute to a deeper understanding of the robustness of topological indices in disordered topological insulators (either in higher dimension or in other classes of the Kitaev table of topological phases of matter) and 
of the disorder-induced transition from localized to delocalized phases for disordered non-trivial Chern insulators. 

The article is structured as follows: in Section \ref{Setting} we introduce a general framework for 2D Chern insulators and state our main results, before moving to their application to the Haldane-Anderson model. In Section \ref{Chern section} we show the robustness of the Chern character with respect to a random perturbation, by showing that it is well defined in the disordered setting, it is integer valued and is jointly continuous on regions of dynamical localization, making it constant there. In Section \ref{localization section} we provide a detailed description of the different dynamical localization regimes that can be applied to our setting, emphasizing the minimal assumptions under which the results hold. In Section \ref{proof section} everything comes together to prove our main results, which are then applied in Section \ref{section far haldane} to give a thorough description of the Anderson metal-insulator transition for the disordered Haldane model. To make our manuscript as self-contained as possible, we gather in the Appendix known results adapted to our general setting, including an estimate on the region of band-edge localization in terms of the disorder parameter.

{\bf Acknowledgments} C.\,R.-M.\ thanks A. Deleporte for interesting and enriching discussions. V.R.\ is indebted to G.M.\ Graf for valuable discussions. V.R. gratefully acknowledges Laboratoire AGM at CY Cergy Paris University for its hospitality. \\
G.P.\ gratefully acknowledges the financial support of MUR-Italian Ministry of University and Research and Next Generation EU within the PRIN project 2022AKRC5P “Interacting Quantum Systems: Topological Phenomena  and Effective Theories”, and from the National Quantum Science and Technology Institute (PNRR MUR project PE0000023-NQSTI). 
The authors gratefully acknowledge the financial support of GSSI and the European Research Council through the ERC StG MaTCh grant agreement no.101117299. This work was partially supported by the LYSM lab, CNRS, and INdAM.

\newpage

\section{Setting and main results}\label{Setting}

{\bf Notation:} Given $a_1,a_2$ linearly independent vectors in $\R^2$, consider the Bravais lattice
\begin{equation}\label{bravais}
    \Ga=\set{\ga\in\R^2:\ga=\ga_1a_1+\ga_2a_2, \ \ga_1,\ga_2\in \Z}=\Span_{\Z}\set{a_1,a_2}\cong \Z^2. 
\end{equation}
Given $\ga=\ga_1a_1+\ga_2a_2\in\Ga$, $\xi=\xi_1a_1+\xi_2a_2\in\Ga$ we set 
\begin{gather*}
    |\ga|=\sqrt{|\ga_1|^2+|\ga_2|^2},\ |\ga|_\infty=\max\set{|\ga_1|,|\ga_2|}, \ 
    \langle \ga\rangle=\sqrt{1+|\ga|^2}, \ \ga\wedge \xi=\ga_2\xi_1-\ga_1\xi_2.
\end{gather*}
Given $n\in\N^*=\N\setminus\set{0}$, consider the Hilbert space \begin{equation*}
    \ell^2(\Ga;\C^n)=\set{\psi:\Ga\to\C^n:\sum_{\ga\in \Ga}\norm{\psi(\ga)}_{\C^n}^2<\infty}.
\end{equation*}
Denote by $\{\delta_{\ga,i}\}_{\ga\in\Ga,i\in\{1,\ldots,n\}}$ the canonical basis in $\ell^2(\Ga;\C^n)$. 
Given $j=1,2$, define the multiplication operator acting on $\ell^2(\Ga;\C^n)$ by
\begin{equation*}
    (X_j\ph)(\ga)=\ga_j\ph(\ga) \quad \textup{for} \ \ph\in \D(X_j)=\set{\psi\in \ell^2(\Ga;\C^n):\sum_{\ga\in\Ga}|\ga_j|^2\norm{\psi(\ga)}_{\C^n}^2<\infty}.
\end{equation*}
Given $L\in\N^*$, define the box of side $L$ by 
\begin{equation}\label{box}
    \La_L=\set{\ga\in \Ga:\ga=\ga_1 a_1+\ga_2a_2, \ \ga_1,\ga_2\in\set{0,\cdots,L-1}-\lfloor L/2\rfloor}\subset \Ga
\end{equation}
where $\lfloor x\rfloor$ is the integer part of $x$, and denote by $\chi_{\La_L}\in\mathcal{B}(\ell^2(\Ga;\C^n))$ the characteristic function of the set $\La_L$. Denote by $C^\infty_{c,+}(\R)$ the set of non-negative compactly supported, infinitely differentiable functions on $\R$.

\subsection{Ergodic setting}
Given $n\in \N^*$ and a Bravias lattice $\Ga\cong \Z^2$ as defined in \eqref{bravais}, consider the Hamiltonian operator acting on $\ell^2(\Ga;\C^n)$ defined by
\begin{equation}\label{Anderson model}
    H_{\la,\omega}=H_0+ \la V_{\omega}
\end{equation}
where $\la>0$ is called disorder parameter, $H_0 , V_\omega$ are bounded self-adjoint operators. Throughout this article we will assume (some of) the following properties:
\begin{enumerate}[label=$(\mathrm{P}_{\arabic*})$,ref=$(\mathrm{P}_{\arabic*})$]

    \item \label{item: P1} ({\bf magnetic $\Ga$-periodicity}) we have that 
    \begin{equation*}
        (T_\ga^B)^* H_0T_\ga^B=H_0 \quad \textup{for all} \ \ga\in\Ga
    \end{equation*}
    with $(T^B_\ga\psi)(\xi)=\eu^{-\iu B (\ga\wedge \xi)}\psi(\xi-\ga)$, $\xi\in\Ga$, where $B\in\R$.
    \item \label{item: P2} ({\bf finite hopping}) there exists $r\in\N^*$ such that the hopping matrix
    \begin{equation}\label{finite hopping}
        H_0(\ga,\xi)=\Big(\langle\delta_{\ga,i},H_0 \delta_{\xi,j}\,\rangle\Big)_{i,j=1}^n\in M_n(\C)
    \end{equation}
    satisfies $H_0(\ga,\xi)=0$ if $|\ga-\xi|_\infty >r$.
    \item\label{item: P3} ({\bf gapped spectrum}) there exist $N\in\N,\, N>1$, $\set{\alpha_i}_{i=1}^N,\set{\beta_i}_{I=1}^N\subset \R$ with $\alpha_i\leq\beta_i<\alpha_{i+1}\leq\beta_{i+1}$ for all $i\in\set{1,\ldots,N-1}$, such that
    \begin{equation*}
        \sigma(H_0)=\bigsqcup_{i=1}^N[\alpha_i,\beta_i].
    \end{equation*}
    We call the closed intervals $[\alpha_i,\beta_i]$ bands, denoted by $\mathcal{B}_i$ for $i\in\set{1,\ldots,N}$, the open intervals $(\beta_i,\alpha_{i+1})$ internal gaps, denoted by $\G_i$ for $i\in\set{1,\ldots,N-1}$, and we call the open intervals $(-\infty,\alpha_1)$ and $(\beta_N,+\infty)$ external gaps, denoted respectively by $\G_0$ and $\G_N$.
    \item \label{item: P4} ({\bf i.i.d. random variables}) $V_\omega$ is a random multiplication operator given by
    \begin{equation*}
        V_\omega=\sum_{\ga\in\Ga}\sum_{i=1}^n\omega_{\ga,i}\ket{\delta_{\ga,i}}\bra{\delta_{\ga,i}}
    \end{equation*}
    where $\set{\omega_{\ga,i}}_{\ga\in\Ga,i\in\set{1,\ldots n}}$ are independent identically distributed (i.i.d.) random variables, with common probability distribution $\rho$, supported in $[-a,b]$, with $0\leq a,b<\infty$ and $a+b>0$. The probability space is $(\Omega,\mathbb{P})=([-a,b]^{\Ga\times \set{1,\dots,n}},\bigotimes_{(\ga,i)\in\Ga\times\set{1,\ldots,n}}\rho)$ and we denote the expectation w.r.t. $\mathbb{P}$ by $\E[\cdot]$.
    
\end{enumerate}

We will make the following regularity assumptions on the distribution $\rho$:

\begin{enumerate}[label=$(\mathrm{R}_{\arabic*})$,ref=$(\mathrm{R}_{\arabic*})$]
    \item\label{item: R1} ({\bf uniform $\tau$-H\"{o}lder continuity}) there exist $\tau\in (0,1]$, $0<C_{\tau,\rho}<\infty$ such that
    \begin{equation*}
        \sup_{u\in\R}\rho([u,u+t])\leq C_{\tau,\rho} t^\tau \quad \textup{for all} \ t\in (0,\infty).
    \end{equation*}
    \item\label{item: R2} ({\bf decay at the edges of the support}) 
    there exist $\beta>2$ and $0<C<\infty$ such that
    \begin{align*}
        \rho\big([b-\eps,b])\leq C\eps^\beta \quad \textup{and} \quad\rho([-a,-a+\eps])\leq C\eps^\beta, \quad \textup{for small} \ \eps>0.
    \end{align*}
    
\end{enumerate}

Note that if $\rho$ is a uniform distribution, for example $\rho(\di v)=\frac{1}{2}\chi_{[-1,1]}(v)\di v$, then \ref{item: R1} holds with $\tau=1$, however \ref{item: R2} does not hold. If $\rho$ has a bounded density, compactly supported with fast decay at the edge, then \ref{item: R1} holds with $\tau=1$ and \ref{item: R2} is satisfied as well.

As a consequence of \ref{item: P1} and \ref{item: P4}, the operator $H_{\la,\omega}$ is  a bounded self-adjoint operator and ergodic in the standard sense \cite{Pastur 1980, Carmona Lacroix 1990, Pastur Figotin 1992}. Hence standard arguments imply that the spectrum is deterministic and given by
\begin{equation*}
    \sigma(H_{\la,\omega})=\sigma(H_0)+\la [-a,b] \quad \textup{for} \ \mathbb{P}\textup{-almost all} \ \omega\in\Omega
\end{equation*}
where  the sum of sets is defined as $A+B=\set{x+y\in\R: x\in A,y\in B}$ \mbox{(Fig. \ref{fig 2})}. See \eg \cite[Thm. 3.9]{Kirsch 2008}.

\begin{figure}[H]
    \centering
    \includegraphics[width=0.73\textwidth]{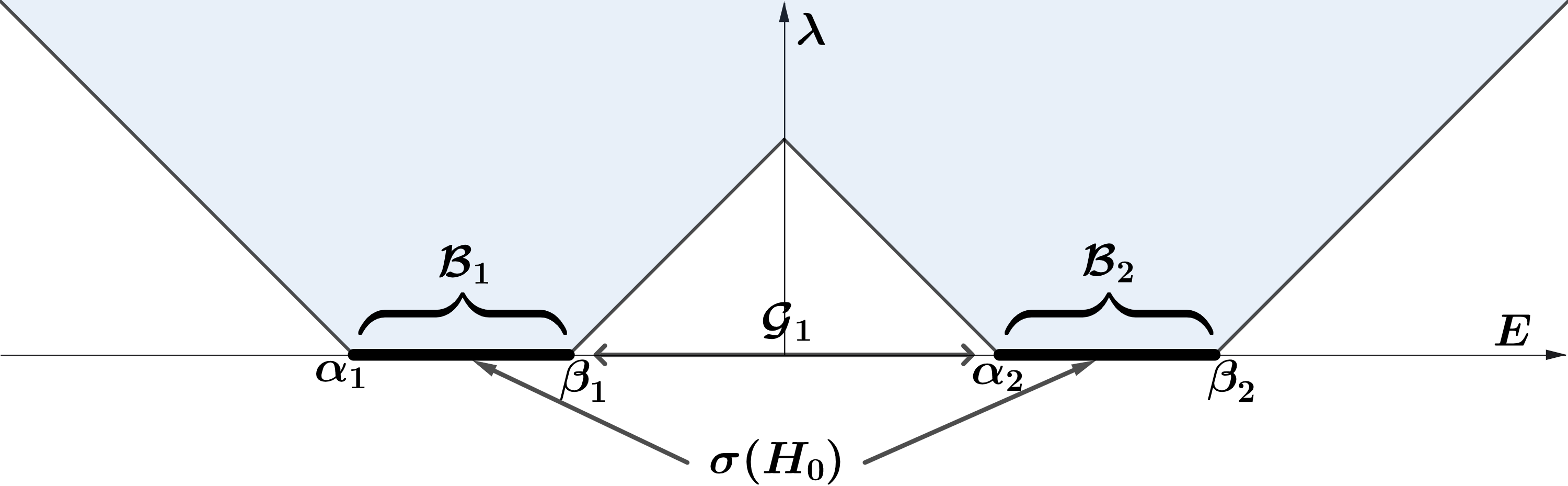}
    \caption{Almost sure spectrum of $H_{\la,\omega}$ plotted in blue in an energy-disorder $(E,\la)$ plane, when $\sigma(H_0)$ is composed of $N=2$ disjoint bands and the distribution $\rho$ is supported in $[-1,1]$.}
    \label{fig 2}
\end{figure}

Under hypotheses \ref{item: P1}, \ref{item: P3}, \ref{item: P4} we have that the spectrum of $H_{\la,\omega}$ can be written, almost surely, as
\begin{equation*}
    \sigma(H_{\la,\omega})=\mathcal{B}_1(\la)\cup\ldots\cup \mathcal{B}_N(\la) \quad \textup{where} \ \mathcal{B}_i(\la)=[\alpha_i-a\la,\beta_i+b\la].
\end{equation*}
If $\la$ is small compared to the size of the unperturbed spectral gaps, that is if 
\begin{equation}\label{gap condition}
        \la<\frac{\min\set{|\G_{i-1},|,|\G_i|}}{a+b} \quad \textup{with} \ \ |\G_i|=
        \begin{cases}
            \alpha_{i+1}-\beta_i &\textup{if}  \ i=1,\ldots ,N-1 \\ \infty &\textup{if} \ i=0,N
        \end{cases}
\end{equation} 
then $\dist (\mathcal{B}_i(\la),\sigma(H_{\la,\omega})\setminus\mathcal{B}_i(\la))>0$ almost surely. Hence the open intervals $\G_{i-1}(\la)$, $\G_i(\la)$, defined by
\begin{gather*}
    \G_i(\la)=
    \begin{cases}
        (-\infty,\alpha_1-a\la) &\textup{if}  \ i=0 \\
        (\beta_i+b\la,\alpha_{i+1}-a\la) &\textup{if} \ i=1,\ldots,N-1 \\
        (\beta_N+b\la,+\infty) &\textup{if} \ i=N
    \end{cases}
\end{gather*}
are nonempty spectral gaps of $H_{\la,\omega}$ almost surely. For the unperturbed operator we also use the notation $\mathcal B_i=\mathcal B_i(0)$, $\mathcal G_i=\mathcal G_i(0)$ for clarity.

\subsection{Main results}

In the following we will denote the spectral projection $\chi_{(-\infty,E]}(H_{\la,\omega})$ by $P_{E,\la,\omega}$ and  $\chi_{(-\infty,E]}(H_0)$ by $P_{E,\la=0}$ for $E\in\R$. Given $L\in\N^*$ consider the box $\La_L$ of side $L$ as defined in \eqref{box}. We recall the concept of Chern character, that generalizes the notion of Chern number, originally defined in the periodic case. This definition does not rely on periodicity and is suitable for our setting, see \cite{Bellissard 1986, Nakamura Bellissard 1990, Bellissard van Elst Schulz-Baldes 1994,Elgart Graf Schenker 2005,Marcelli Moscolari Panati 2023}.

\begin{dfn}[Chern character]
    Let $P\in\mathcal{B}(\ell^2(\Ga;\C^n))$ be an orthogonal projection. Assume that $\chi_{\La_L}P[[X_1,P],[X_2,P]]P\chi_{\La_L}$ is trace class for all $L\in\N^*$. Under these assumptions, the Chern character of $P$ is defined, whenever the limit exists, by
    \begin{equation*}
        C(P)=\lim_{\substack{L\to\infty}}\frac{2\pi\iu}{|\La_L|}\Tr_{\ell^2(\Ga;\C^n)}\Big(\chi_{\La_L}P[[X_1,P],[X_2,P]]P\chi_{\La_L}\Big)
    \end{equation*}
    where $|\La_L|=L^2$ is the number of points of $\Ga$ inside the box $\La_L$.
\end{dfn}

The Chern character $C(P)$ is well defined whenever $P$ is a spectral projection associated to suitable gapped Schr\"{o}dinger-type operators \cite{Elgart Graf Schenker 2005,Marcelli Moscolari Panati 2023}. We say an operator $H_0$ satisfying conditions \ref{item: P1},\ref{item: P2} and \ref{item: P3} is a tight-binding \emph{non-trivial Chern insulator} if there exists (at least) a band $\mathcal B_i =[\alpha_i,\beta_i]$ in its spectrum such that $C(P_{E_1,\la=0})\neq C(P_{E_2,\la=0})$ for $E_1\in \mathcal G_{i-1}$ and $E_2\in \mathcal G_{i}$, that is, the Chern character changes from one spectral gap to the next.

We are interested in the study of non-trivial Chern insulators perturbed by an Anderson-like potential, of the type \eqref{Anderson model}. In order to state our main results, we recall the notion of dynamical (de)localization for random operators:

\begin{dfn}\label{definition dynamical localization}
    We say that $H_{\la,\omega}$ exhibits {\it dynamical localization} (DL) in $I\subset\R$ if there exist constants $C_\la<\infty,\mu_\la>0,\beta_\la\in(0,1]$ such that
    \begin{equation}\label{DL condition}
        \E\bigg[\sup_{E\in I}\Tr_{\C^n}|P_{E,\la,\omega}(\ga,\xi)|\bigg]\leq C_\la \exp(-\mu_\la|\ga-\xi|^{\beta_\la}).
    \end{equation}
    We say that $H_{\la,\omega}$ exhibits dynamical localization in $E\in\R$ if there exists $\delta>0$ such that $H_{\la,\omega}$ exhibits dynamical localization in $(E-\delta,E+\delta)$. Otherwise we say that $H_{\la,\omega}$ exhibits {\it dynamical delocalization} (DD) in $I$ or in $E$.
\end{dfn}

\begin{rmk}
    It can be shown, see \eg \cite[Thm. 4.2]{Germinet Klein 2004} and \cite[Thm. 3]{Germinet Klein 2005}, that, under assumptions \ref{item: P1},\ref{item: P2},\ref{item: P4} and \ref{item: R1}, condition \eqref{DL condition} holding in an open interval $I\subset \R$ is equivalent to absence of transport, in the following sense: for all $g\in C_{c,+}^\infty(I)$
    \begin{equation}
        \E\left[\sup_{t\in\R}\norm{\langle X\rangle^{\frac{p}{2}}\eu^{-\iu tH_{\la,\omega}}g(H_{\la,\omega})\chi_0}\right]<\infty 
    \end{equation}
    for all $p>0$, where $(\langle X\rangle \ph)(\ga)=\langle\ga\rangle \ph(\ga)$, and $\chi_0=\chi_{\La_1}\in \mathcal{B}(\ell^2(\Ga;\C^n))$.
\end{rmk}

In order to measure the property of DL, we follow \cite{Germinet Klein 2004} and introduce the  average moment of order $p\geq 0$, for the time evolution of wave packets initially localized in energy by a functions $g\in C_{c,+}^\infty(\R)$
\begin{equation}
    \mathcal{M}_{\la}(p,g,T)=\frac{2}{T}\int_0^\infty \di t \, \eu^{-\frac{2t}{T}}\E\left[\norm{\langle X\rangle^\frac{p}{2}\eu^{-\iu t H_{\la,\omega}}g(H_{\la,\omega})\chi_0 }^2\right].
\end{equation}
The average moments allow to characterize the region of dynamical localization, as shown in \cite[Thm. 2.11]{Germinet Klein 2004}, which we state in our setting as follows:

\begin{prop}[\protect{\cite[Thm. 2.11]{Germinet Klein 2004}}]\label{rmk characterization}
    Under assumptions \ref{item: P1},\ref{item: P2},\ref{item: P4} and \ref{item: R1} the following two statements are equivalent:
    \begin{enumerate}
        \item $H_{\la,\omega}$ exhibits dynamical localization in $E$;
        \item there exists $g\in C_{c,+}^\infty(\R)$ with $g\equiv 1$ on an open interval containing $E$, $\alpha\geq 0$ and $p>\frac{4\alpha}{\tau}+\frac{24}{\tau}+12$ such that
        \begin{equation*}
            \liminf_{T\to\infty}\frac{1}{T^{\alpha}}\mathcal{M}_\la(p,g,T)<\infty.
        \end{equation*}\end{enumerate}
\end{prop}

Our first result shows that, for disordered non-trivial Chern-insulators with Hamiltonian $H_{\la,\omega}$, there exists (at least) a disordered broadened spectral band $\mathcal B_i(\lambda)$, between spectral gaps, where $H_{\la,\omega}$ exhibits delocalization in (at least) one energy $E^*\in\mathcal B_i(\lambda)$. Moreover, we have quantitative estimates for the growth of the time-averaged moments of the wave packets at these energies.
\begin{thm}\label{DD energy}
    Assume $H_{\la,\omega}=H_0+\la V_\omega$ satisfies hypothesis \ref{item: P1}-\ref{item: P4} and assume that there exists $i\in\set{1,\ldots,N}$ such that 
    \begin{equation*}
        C(P_{E_1,\la=0}) - C(P_{E_2,\la=0}) \neq 0
    \end{equation*}
     for $E_1\in \mathcal G_{i-1}(0)$ and $E_2\in \mathcal G_{i}(0)$. Then: 
    \begin{enumerate}
        \item \label{point 1 dd energy}
        For all $\la<\frac{\min\set{|\G_{i-1}(0)|,|\G_i(0)|}}{a+b}$, there exists $E^*\in\mathcal{B}_i(\la)$ such that $H_{\la,\omega}$ exhibits dynamical delocalization in $E^*$.
        \item\label{point 2 dd energy} If in addition \ref{item: R1} hold, then for all $g\in C_{c,+}^\infty(\R)$, $g\equiv 1$ on an open interval containing $E^*$, for all $p>0$ there exists  $C_{p,g,\la}>0$ such that
        \begin{equation*}
            \mathcal{M}_\la(p,g,T)\geq C_{p,g,\la}T^{\frac{\tau p}{4}-\kappa}
        \end{equation*}
        for all $\kappa>3\tau +6,$ and $T\geq 0$.
        \item\label{point 3 dd energy} If in addition \ref{item: P1} holds with $B\in 2\pi\Q$ and \ref{item: R1},\ref{item: R2} hold, then for all $\la<\frac{\min\set{|\G_{i-1}(0)|,|\G_i(0)|}}{a+b}$, there exists $\delta_{i-1}(\la),\delta_i(\la)>0$ such that $H_{\la,\omega}$ exhibits dynamical localization in $[\alpha_i-a\lambda,\alpha_i-a\la +\delta_{i-1}(\la)]\cup [\beta_i+b\la-\delta_i(\la),\beta_i+b\la ] $ and there exists $E^*\in(\alpha_i-a\la+\delta_{i-1}(\la), \beta_i+b\la-\delta_i(\la))$ such that $H_{\la,\omega}$ exhibits dynamical delocalization in $E^*$. 
    \end{enumerate}
\end{thm}
Theorem \ref{DD energy} is denoted “dynamical delocalization (DD) in energy parameter" throughout this article. Note that Part \eqref{point 1 dd energy} does not require any regularity assumption on the single-site probability distribution. This has been proved by \cite[Thm. 1]{Bellissard van Elst Schulz-Baldes 1994} for ergodic Landau Hamiltonians in a tight binding representation using non-commutative geometry techniques, and by \cite[Thm. 5]{Aizenman Graf 1998} for ergodic Landau Hamiltonians in $\ell^2(\Z^2)$ using functional analytic tools. This strategy was used in \cite{Germinet Klein Schenker 2007} to show delocalization for random Landau Hamiltonians in $L^2(\mathbb R^2)$ and in \cite{Becker Oltman Vogel 2025}, to show delocalization for a random Bistritzer-MacDonald Hamiltonian, a model for disordered twisted bilayer graphene. The case of the magnetic discrete Laplacian perturbed by an Anderson potential in a honeycomb structure was treated in \cite[Thm. 2]{Becker Han 2022}. Assuming the regularity of the single-site probability distribution, \cite{Germinet Klein 2004} gave a quantitative estimate of delocalization using time-averaged moments, which applies to our setting and yields Part \eqref{point 2 dd energy}. This was also exploited in \cite{Becker Oltman Vogel 2025}. Under an additional assumption on the tails of the probability distribution, we show band-edge localization using the Multiscale Analysis as in \cite{Figotin Klein 1994} to obtain Theorem \ref{DD energy}-\eqref{point band edge}. This allows us to give more information on the location of the mobility edge in each band, improving the statement \ref{DD energy}-\eqref{point 1 dd energy}. All these aspects considered, Theorem \ref{DD energy} has the advantage of identifying sufficient requirements on disordered Chern insulators to prove a metal-insulator transition in energy, under the condition of open spectral gaps.

When the disorder parameter $\la$ reaches a certain value (see \eqref{gap condition}), the spectral gaps close and the previous Theorem does not apply anymore. However, in this case it is still possible to show existence of dynamical delocalization in the spectrum, which is the content of our next and main result. More precisely, we prove that for each energy $E$ in the spectral gap of the unperturbed spectrum, whenever the associated Chern character is different from zero, there exists (at least) one value of the disorder parameter $\la^*$ such that $E\in\sigma(H_{\la^*,\omega})$ and $H_{\la^*,\omega}$ exhibits dynamical delocalization in $E$. This result relies on the well-known fact that, if $H_{\la,\omega}$ satisfies \ref{item: P2}, \ref{item: P4} and \ref{item: R1}, there exists $\lambda_{\rho}(H_0)>0$, such that $H_{\la,\omega}$ exhibits dynamical localization throughout its a.s. spectrum for all  $\la>\lambda_{\rho}(H_0)$, see e.g. Theorem \ref{DL regimes}-\eqref{strong disorder regime thm}. This is called the strong disorder regime, and we call $\la_\rho(H_0)$ the \emph{strong disorder threshold}. 

To introduce the next result, we consider an energy $E\in \G_i(0)$ with $i\in\set{1,\ldots,N-1}$, and we denote by $\la_0(E)$ the smallest value for which $E\in \sigma(H_0)+\la[-a,b]$, \ie $E$ is either a left or right internal band edge for $\sigma(H_{\la_0(E),\omega})$ almost surely. Since in hypothesis \ref{item: P4} we cannot have $a=b=0$, we see that $\la_0(E)$ can be expressed as
\begin{equation}\label{lambda zero}
    \la_0(E)=
    \begin{cases}
        \min\set{\frac{E-\beta_i}{b},\frac{\alpha_{i+1}-E}{a}} &\textup{if} \ a\neq0, b\neq 0 \\
        \frac{E-\beta_i}{b} & \textup{if}  \ a=0 \\
        \frac{\alpha_{i+1}-E}{a} &\textup{if} \ b=0
    \end{cases}.
\end{equation}

\begin{figure}[ht]
    \centering
    \includegraphics[width=0.9\textwidth]{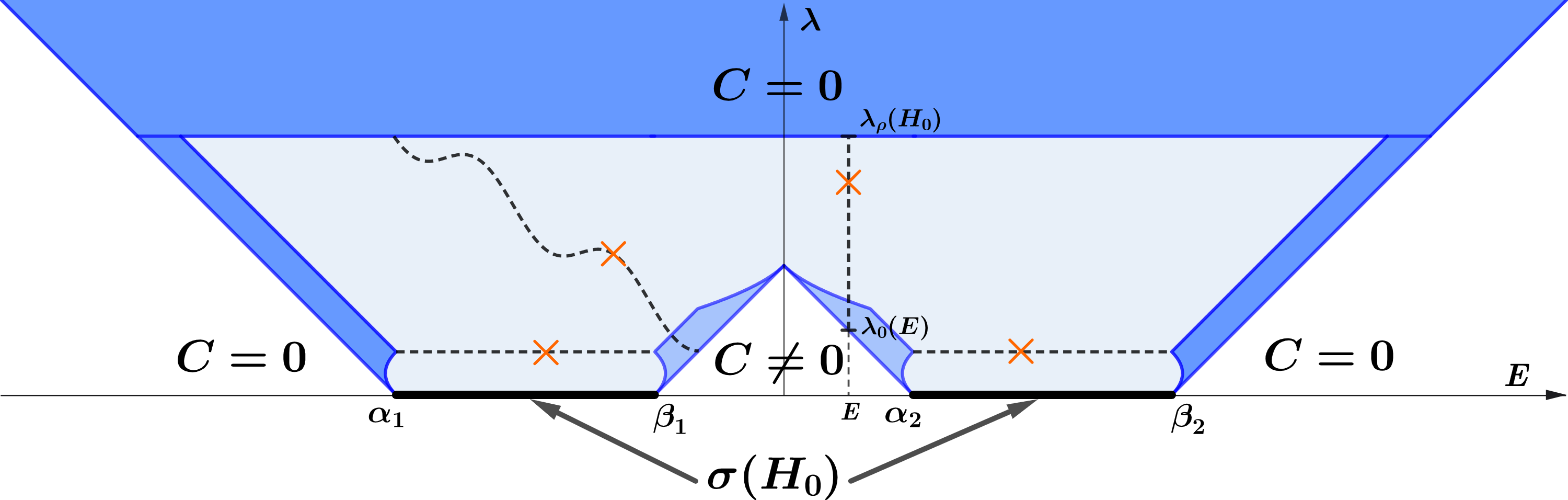}
    \caption{\small{Phase diagram in the case of $\sigma(H_0)$ composed of two bands with non trivial Chern number, with random variables supported in $[-1,1]$. The dark blue regions are the regimes of strong disorder and of external band edge localization, in which the Chern character is zero, while the light blue region is the regime of internal band edge localization, where the Chern character retains the same value of the unperturbed system. The inner region is where it is possible to prove existence of dynamical delocalization (orange crosses).}}
    \label{fig 3}
\end{figure}

\begin{thm}\label{DD disorder}
    Assume $H_{\la,\omega}=H_0+\la V_\omega$ satisfies hypothesis \ref{item: P1}-\ref{item: P4} and \ref{item: R1}. Let $E\in \G_i(0)$ with $i\in\set{1,\ldots,N-1}$ and assume $C(P_{E,\la=0})\neq 0$. Let $\la_0(E)$ be as in \eqref{lambda zero} and let $\lambda_\rho(H_0)$ be the strong disorder threshold.
    \begin{enumerate}
        \item \label{point 1 dd disorder} There exists $\la^*\in[\la_0(E),\la_\rho(H_0)]$ such that $H_{\la^*,\omega}$ exhibits dynamical delocalization in $E$.         Furthermore, for all $g\in C_{c,+}^\infty(\R)$ with $g\equiv 1$ on an open interval containing $E$, for all $p>0$ there exists $C_{p,g,\la^*}>0$ such that
        \begin{equation}\label{transport exponent}
            \mathcal{M}_{\la^*}(p,g,T)\geq C_{p,g,\la^*}T^{\frac{\tau p}{4}-\kappa}
        \end{equation}
        for all $\kappa> 3\tau+6$ and $T\geq 0$. 
        \item\label{point 2 dd disorder} If in addition we assume \ref{item: P1} holds with $B\in 2\pi\Q$ and \ref{item: R2} holds, then for each $E\in \G_i(0)$ such that $E\neq \frac{b\alpha_{i+1}+a\beta_i}{a+b}$, $H_{\la_0(E),\omega}$ exhibits dynamical localization at the band edge $E$, therefore (1) holds with  $\la^*\in (\la_0(E),\la_\rho(H_0)]$.
    \end{enumerate}
\end{thm}

Theorem \ref{DD disorder} is denoted “dynamical delocalization (DD) in disorder parameter" through this article. This shows the existence of a metal-insulator transition  for disordered Chern insulators in the disorder parameter.  \\
The proof of the result is postponed to Section \ref{proof section} and is based on two results: Theorem \ref{continuity chern disorder}, which states that the Chern character is well defined, quantized and remains constant  for $(E,\lambda)$ in a connected region of dynamical localization, as long as the spectral projection is H\"{o}lder continuous in energy and disorder; and Proposition \ref{chern zero strong disorder}, which states that the Chern character is constantly equal to zero for all energies in the spectrum for disorder parameter larger than the strong disorder threshold. In a region where the Chern character associated to a spectral gap of the unperturbed operator is non-zero, the latter implies a jump in the value of the Chern character in the disorder parameter. This is the mechanism for delocalization.

If the disorder parameter is such that the spectral gaps remain open, the existence of delocalized energies can be deduced either by Theorem \ref{DD energy} or by Theorem \ref{DD disorder}. However, when the gaps close, only the latter applies. To illustrate the advantage of Theorem \ref{DD disorder}, in Section \ref{section far haldane} we will consider a disordered Haldane model with a suitable probability distribution, and we prove, by using Theorem \ref{DL regimes}-\eqref{point 2 DL}, the existence of localized energies in the bulk of the spectrum of the gapless Hamiltonian (gapless spectrum). Hence Theorem \ref{DD disorder} yields the existence of delocalized energies well inside the gapless spectrum, where the assumptions of Theorem \ref{DD energy} are not satisfied.

We can extend the previous delocalization results by considering any continuous curve in the energy-disorder plane connecting regions with different Chern characters, without necessarily restricting to horizontal or vertical lines.

\begin{thm}\label{curve}
    Assume $H_{\la,\omega}=H_0+\la V_\omega$ satisfies hypothesis \ref{item: P1}-\ref{item: P4} and \ref{item: R1}. Assume that there exists $i\in \set{1,\ldots,N-1}$ such that $C(P_{E,\la=0})\neq 0$ for $E\in \G_i$. Given $T>0$, let $C(\cdot)=(E(\cdot),\la(\cdot))\in C([0,T];\R\times (0,\infty))$ be a continuous curve with $\min_{t\in [0,T]}|\la(t)|>0$. If $E(0)\in \G_i(\la(0))$ and, either $\la(T)>\la_\rho(H_0)$, or $E(T)\in \G_j(\la(T))$ with $i\neq j$ and $C(P_{E(0),\la=0})\neq C(P_{E(T),\la=0})$, then there exists $t^*\in (0,T)$ such that $H_{\la(t^*),\omega}$ exhibits DD in $E(t^*).$
\end{thm}

The proof of Theorem~\ref{curve} is provided in Section~\ref{proof section}.

\goodbreak

\subsection{A paradigmatic example: the Haldane model}

In this section, we describe the Haldane model in first quantization, following the extended review in \cite{Marcelli Monaco Moscolari Panati 2018}. The Haldane model is a 2D tight-binding model which describes non interacting electrons on a honeycomb structure $\Ci\subset\R^2$ (see Fig. \ref{fig 4}) which is characterized by the displacement vectors
\begin{equation*}
    d_1=\bigg(\frac{1}{2},-\frac{\sqrt{3}}{2}\bigg),\quad d_2=\bigg(\frac{1}{2},\frac{\sqrt{3}}{2}\bigg), \quad d_3=(-1,0).
\end{equation*}
In order to describe the periodicity of the system, one introduces the periodicity vectors
\begin{equation*}
    a_1=d_2-d_3=\bigg(\frac{3}{2},\frac{\sqrt{3}}{2}\bigg), \  a_2=d_3-d_1=\bigg(-\frac{3}{2},\frac{\sqrt{3}}{2}\bigg), \ a_3=d_1-d_2=(0,-\sqrt{3}).
\end{equation*}
Given the Bravais lattice $\Ga=\Span_\Z\set{a_1,a_2,a_3}\cong\Z^2$, the honeycomb structure can be written as $\Ci=\Ga_A\cup \Ga_B, \quad \Ga_A=\Ga, \ \Ga_B=\Ga+\nu$, for $\nu\in\set{d_1,d_2,d_3}$. Notice that it is sufficient to choose only two $a_i$-vectors and one $d_i$-vector in order to reconstruct the whole crystalline structure. Each of these possible choices, which are called dimerizations of $\Ci$, results in an identification $\Ci\cong \Ga\times \set{0,\nu}$, where $\nu$ is a displacement vector, providing an isomorphism $\ell^2(\Ci)\cong \ell^2(\Ga;\C^2)$.
\begin{figure}[ht]
    \centering
    \includegraphics[width=0.3\textwidth]{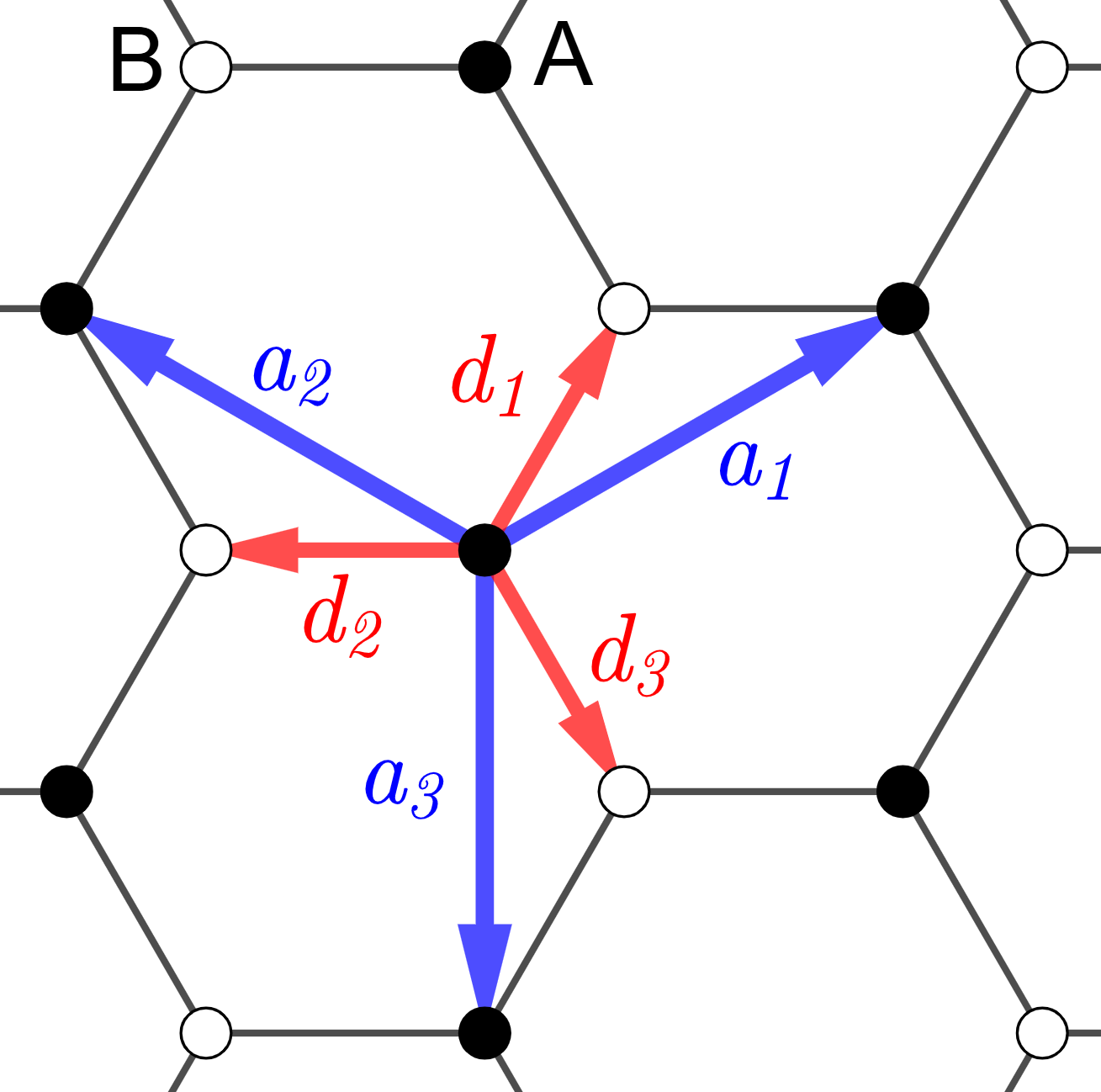}
    \caption{\small{Periodicity vectors (blue) and displacement vectors (red) for the honeycomb structure $\Ci$.}}
    \label{fig 4}
\end{figure}

The dynamics of electrons in the Haldane model is described by a Hamiltonian that acts on $\ell^2(\Ci)$ and depends on several parameters $(t_1,t_2,\phi,M)$, where $t_1>0,t_2>0$ are hopping energies, $\phi\in (-\pi,\pi]$ is a local magnetic flux and $M\in\R$ is an on-site energy which distinguishes among $\Ga_A$ and $\Ga_B$. 
    
In order to introduce the Hamiltonian, we define the translation operator by $u\in\R^2$ as
\begin{equation}
        (\tilde{T}_u\psi)(x)=
        \begin{cases}
            \psi(x-u) &\textup{if} \ x-u\in\Ci \\
            0 &\textup{otherwise}
        \end{cases} 
        \quad \textup{for} \ \psi\in\ell^2(\Ci).
    \end{equation}
We denote by $\chi_A$, $\chi_B$ the characteristic functions over the sublattices $\Ga_A$, $\Ga_B$, respectively. Then we define the Haldane operator acting on $\ell^2(\Ci)$ as $H_{\textup{Hal}}=H_{t_1,t_2,\phi,M}=H_{NN}+H_{NNN}+V$, where
\begin{gather*}
H_{NN}=t_1\sum_{i=1}^3(\tilde{T}_{d_i}+\tilde{T}_{d_i}^*),
  \quad \quad V=M(\chi_A-\chi_B),\\
H_{NNN}=t_2(\cos\phi)\sum_{i=1}^3(\tilde{T}_{a_i}+\tilde{T}_{a_i}^*)+ t_2(\iu\sin\phi)(\chi_{A}-\chi_B)\sum_{i=1}^3(\tilde{T}_{a_i}-\tilde{T}_{a_i}^*). 
\end{gather*}
Given the dimerization isomorphism $\Ci\cong \Ga\times \set{0,\nu}$, with $\nu\in \set{d_1,d_2,d_3}$, consider the associated unitary operator $U:\ell^2(\Ci)\xrightarrow[]{\cong}\ell^2(\Ga;\C^2)$. Define the operator $\Hi_{\textup{Hal}}=UH_{\textup{Hal}}U^*\in\mathcal{B}(\ell^2(\Ga;\C^2))$ and consider, for $\ga,\xi\in\Ga$, the hopping matrix
\begin{equation*}
    \Hi_{\textup{Hal}}(\ga,\xi)=
    \begin{bmatrix}
        \langle \delta_{\ga,1},\Hi_{\textup{Hal}}\delta_{\xi,1}\rangle & \langle\delta_{\ga,1},\Hi_{\textup{Hal}}\delta_{\xi,2}\rangle \\ \langle\delta_{\ga,2},\Hi_{\textup{Hal}}\delta_{\xi,1}\rangle &\langle \delta_{\ga,2},\Hi_{\textup{Hal}}\delta_{\xi,2}\rangle
    \end{bmatrix}=
    \begin{bmatrix}
        \langle \tilde{\delta}_\ga,H_{\textup{Hal}}\tilde{\delta}_\xi\rangle & \langle \tilde{\delta}_\ga,H_{\textup{Hal}}\tilde{\delta}_{\xi+\nu}\rangle \\ \langle \tilde{\delta}_{\ga+\nu},H_{\textup{Hal}}\tilde{\delta}_\xi\rangle & \langle \tilde{\delta}_{\ga+\nu},H_{\textup{Hal}}\tilde{\delta}_{\xi+\nu}\rangle
    \end{bmatrix}
\end{equation*}
where $\set{\tilde{\delta}_x}_{x\in\Ci}$ is the canonical basis in $\ell^2(\Ci)$. Since $H_{\textup{Hal}}$ commutes with $\tilde{T}_\ga$ for all $\ga\in \Ga$, we have that $\Hi_{\textup{Hal}}$ satisfies \ref{item: P1} with $B=0$. Moreover, since $H_{NN}$ and $H_{NNN}$ are hopping operators between nearest and next-to-nearest neighbor, we get that $\Hi_{\textup{Hal}}(\ga,\xi)=0$ if $|\ga-\xi|_\infty>1$, so \ref{item: P2} holds with $r=1$. 

As far as the spectrum is concerned, the Hamiltonian has absolutely continuous spectrum, as a consequence of periodicity, composed of two intervals which touch each other whenever
\begin{equation}\label{parameter haldane}
    M=\pm 3\sqrt 3\,t_2\sin\phi.
\end{equation}
Therefore, for parameters for which \eqref{parameter haldane} is not satisfied, we have that $H_{\textup{Hal}}$ satisfies \ref{item: P3} with $N=2$ bands. Each open region in the parameter space described by condition \eqref{parameter haldane} is labeled by an integer which corresponds to the Chern number of the Bloch bundle (\cite[Thm. 1]{Panati 2007}) corresponding to the spectral projection $P_F$ associated to the energy band below the non-empty gap (see Fig. \ref{fig 5}). To compute the Chern number, we use the so called Bloch-Floquet-Zak transform
\begin{equation*}
    \U_Z:\ell^2(\Ci)\to L^2(\mathbb{T}^2_*;\C^2)
\end{equation*}
where $\mathbb{T}^2_*=\R^2/\Ga^*$ and  $\Ga^*$ is the dual lattice of $\Ga$ which is defined as
\begin{equation*}
    \Ga^*=\Span_\Z\set{b_1,b_2}, \quad \textup{with} \  a_i\cdot b_j=2\pi\delta_{i,j}, \ i,j\in\set{1,2}.
\end{equation*} In the Bloch-Floquet-Zak representation, we can write the eigenprojection $P_F$ as
\begin{equation*}
    (\U_Z P_F \U_Z^{-1}\ph)(k)=P_F(k)\ph(k)\in\C^2, \quad \textup{for} \  k\in\mathbb{T}^2_*
\end{equation*}
for any $\ph\in L^2(\mathbb{T}_*^2;\C^2).$ Through the operator $P_F(k)$, we have that the (first) Chern number associated to $P_F$ can be expressed as
\begin{equation*}
    c_1(P_F)=\frac{1}{2\pi\iu}\int_{\mathbb{T}^2_*} \Tr_{\C^2}\Big(P_F(k) \,[\partial_{k_1}P_F(k),\partial_{k_2}P_F(k)]\Big) \di k_1\wedge \di k_2\in\Z.
\end{equation*}
As depicted in Fig. \ref{fig 5}, the Chern number associated to $P_F$, for the Haldane model, takes value +1 in the orange region, -1 in the blue region and 0 in the white one.

\begin{figure}[ht]
    \centering
    \includegraphics[width=0.55\textwidth]{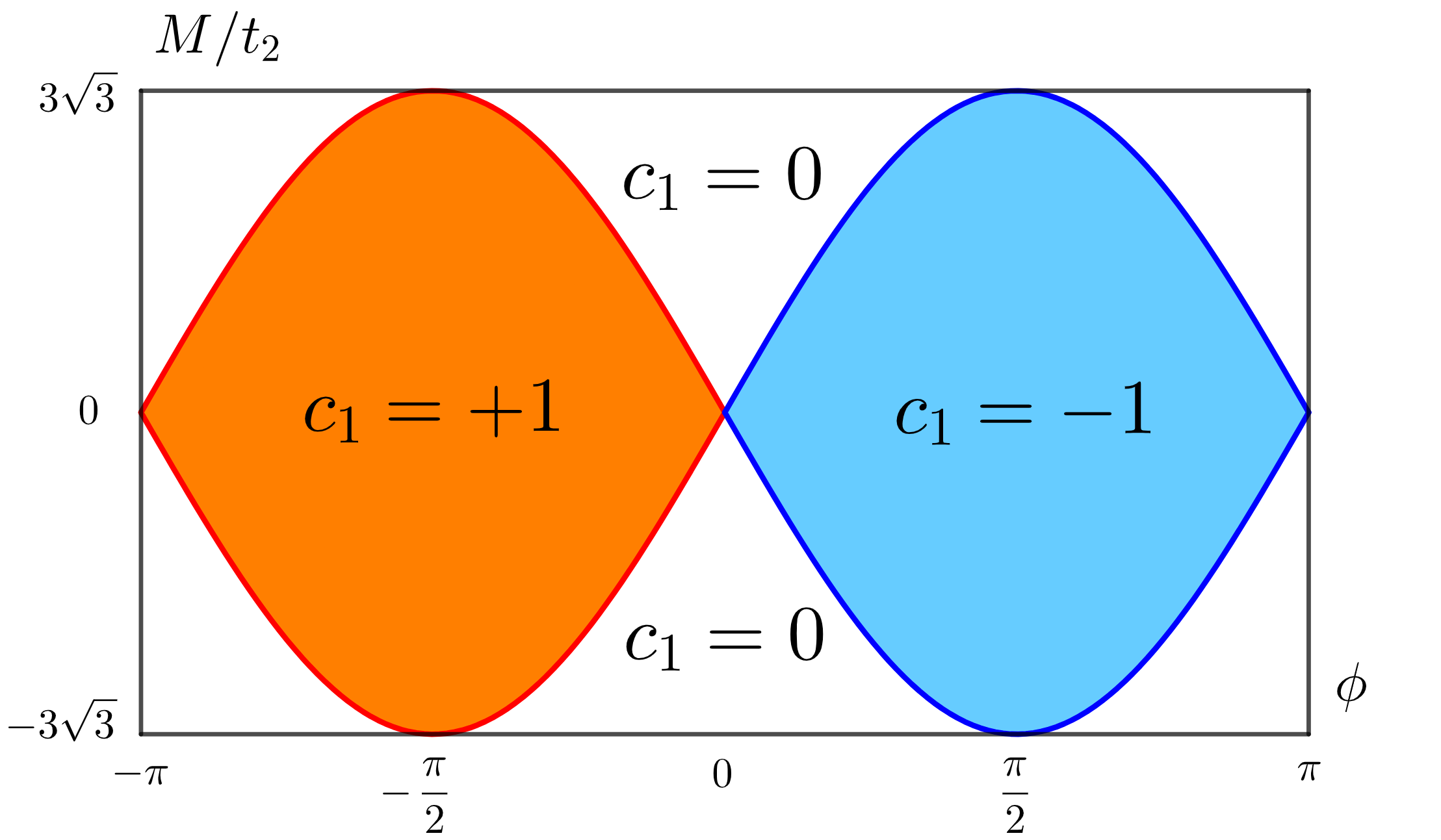}
    \caption{\small{The bold continuous line corresponds to $M=\pm 3\sqrt 3\,t_2\sin\phi$ and represents the set of parameters $(\phi,M/t_2)$ for which the Hamiltonian is gapless. The line divides the phase diagram into regions where the (gapped) Hamiltonian exhibits constant Chern numbers given by $0,-1,+1$.}}
    \label{fig 5}
\end{figure}

Finally, we introduce the Haldane model with an Anderson-type potential.

\begin{dfn}[Haldane-Anderson model]\label{defn:haldane-anderson}
    Consider a dimerization of the honeycomb structure $\Ci\cong \Ga\times \set{0,\nu}$, with $\nu\in\set{d_1,d_2,d_3}$, which yields a unitary isomorphism $U:\ell^2(\Ci)\xrightarrow[]{\cong} \ell^2(\Ga;\C^2)$.  We define the Haldane-Anderson Hamiltonian, acting on $\ell^2(\Ga;\C^2)$, as
    \begin{equation*}
        \Hi_{\textup{Hal},\la,\omega}=\mathcal{H}_{\textup{Hal}}+\la V_\omega,\quad \la>0
    \end{equation*}
    where $\Hi_{\textup{Hal}}=U H_{\textup{Hal}} U^*\in\mathcal{B}(\ell^2(\Ga;\C^2))$ and $V_\omega$ is an Anderson-type potential acting on $\ell^2(\Ga;\C^2)$ satisfying \ref{item: P4}.
\end{dfn}

Theorems \ref{DD energy} and \ref{DD disorder}  apply to the Haldane-Anderson model, as detailed in the following corollary. For the sake of a simpler statement, we specialize to the case $\phi=\pm \frac{\pi}{2}$, so that the spectrum is symmetric w.r.t.\ zero.  Anyhow, our general results apply to the Haldane-Anderson model for any parameters $(\phi, M/t_2)$ corresponding to a \emph{non-trivial} topological phase. 

\begin{crl}\label{corollary haldane}
    Consider the Haldane Hamiltonian $\Hi_{\textup{Hal}}$  with parameters $(\phi,M/t_2)$ satisfying $-3\sqrt 3t_2<M<3\sqrt{3}t_2$ and $\phi=\pm \frac{\pi}{2}$.
    Then, besides satisfying \ref{item: P1} and \ref{item: P2}, $\Hi_{\textup{Hal}}$ satisfies \ref{item: P3}, with $\sigma(\Hi_{\textup{Hal}})=\mathcal{B}_1\sqcup\mathcal{B}_2$, with
    \begin{equation*}
        \mathcal{B}_1=\bigg [-\norm{\Hi_{\textup{Hal}}},-\frac{|\G_1|}{2}\bigg], \quad \mathcal{B}_2=\bigg [\frac{|\G_1|}{2},\norm{\Hi_{\textup{Hal}}}\bigg ],
    \end{equation*}
    where $|\G_1|>0$ is the size of the internal spectral gap and $C(P_{E,\la=0})=\mp1$ for $E\in\G_1=\left(-\frac{|\G_1|}{2},\frac{|\G_1|}{2}\right)$. Let $\Hi_{\textup{Hal},\la,\omega}$ be the Haldane-Anderson Hamiltonian satisfying assumptions \ref{item: R1} and \ref{item: R2}. Then:
    \begin{enumerate}
        \item(DL at band edges) There exists $\delta(\la)>0$ such that $H_{\textup{Hal},\la,\omega}$ exhibits dynamical localization in $\Big[-\norm{\Hi_{\textup{Hal}}}-a\lambda,-\norm{\Hi_{\textup{Hal}}}-a\la +\delta(\la)\Big]\cup\Big[\norm{\Hi_{\textup{Hal}}}+b\lambda-\delta(\la),\norm{\Hi_{\textup{Hal}}}+b\la\Big]$. For all $\la<\frac{|\G_1|}{a+b}$, there exists $\delta_1(\la)>0$ such that $H_{\textup{Hal},\la,\omega}$ exhibits dynamical localization in $\Big[-\frac{|\G_1|}{2}+b\la-\delta_1(\la),-\frac{|\G_1|}{2}+b\la\Big]\cup\Big[\frac{|\G_1|}{2}-a\la,\frac{|\G_1|}{2}-a\la+\delta_1(\la)\Big]$.
        \item (DL at strong disorder) There exists $\la_\rho(\Hi_{\textup{Hal}})>0$ such that, for all $\la>\la_\rho(\Hi_{\textup{Hal}})$, we have that $\Hi_{\textup{Hal},\la,\omega}$ exhibits dynamical localization throughout its spectrum.
        \item (DD in energy) For all $\la<\frac{|\G_1|}{a+b}$, there exist $\delta(\la),\delta_1(\la)>0$ and (at least) two energies $E_1^*\in \Big(-\norm{\Hi_{\textup{Hal}}}-a\la+\delta(\la), -\frac{|\G_1|}{2}+b\la-\delta_1(\la)\Big)$, $E_2^*\in \Big(\frac{|\G_1|}{2}-a\la+\delta_1(\la),\norm{\Hi_{\textup{Hal}}}+b\la-\delta(\la)\Big)$ such that $\Hi_{\textup{Hal},\la,\omega}$ exhibits dynamical delocalization in $E_1^*$ and $E_2^*$. 
        \item (DD in disorder) For all $E\in \G_1$ there exists (at least) one disorder value $\la^*\in \displaystyle[\la_0(E),\la_\rho(\Hi_{\textup{Hal}})]$, where $\la_0(E)$ is defined in Theorem \ref{DD disorder}, such that $\Hi_{\textup{Hal},\la^*,\omega}$ exhibits dynamical delocalization in $E$. Moreover $\la^*>\la_0(E)$ if $E\neq \frac{|\G_1|(b-a)}{2(a+b)}$.
    \end{enumerate}
\end{crl}

Last but not least, we spend a few words on the so-called \emph{unit cell consistency}, that is the requisite of independence of the results from the choice of a unit cell \cite{Yang He Zheng 2020} and from the choice of a position operator. As recalled above, the Haldane model is originally formulated on the honeycomb structure $\Ci$, invariant under the action of a Bravais lattice $\Ga\cong \Z^2$. The choice of an origin and of a linear basis $\set{a_1,a_2}$ for $\Ga$ yields a unit cell, and hence a decomposition $\Ci\cong \Z^2\times \set{0,\nu}$, where $\nu$ is a displacement vector. Moreover, one can define the position operator w.r.t.\ the canonical basis $\set{e_1,e_2}\subset\R^2$, or w.r.t.\ the lattice basis $\set{a_1,a_2}\subset\R^2$. The fact that the whole theory is independent of these choices is, in our opinion, a physically relevant aspect of the problem. However, to keep the paper to a reasonable length, we decided to postpone the discussion of the consistency issues to a forthcoming short note \cite{PR next}. 

\goodbreak

\section{Continuity \& quantization of Chern character}\label{Chern section}

\subsection{Well-posedness and quantization of Chern character}

In this section we recall some preliminary results about the Chern character, concerning its well-posedness and its quantization property, that is, the fact that it takes only integer values. In the following we will use the shortened notation 
\begin{equation*}
    \mathfrak{C}_P=P[[X_1,P],[X_2,P]]P,
\end{equation*}
so that the Chern character reads
\begin{equation}\label{Chern marker}
    C(P)=\lim_{\substack{L\to\infty}}\frac{1}{|\La_L|}\Tr_{\ell^2(\Ga;\C^n)}(\chi_{\La_L}\mathfrak{C}_P\chi_{\La_L}).
\end{equation}
Since the multiplication operators appearing \eqref{Chern marker} are unbounded, we need to require some decay properties from the integral kernel of the projection, at least in expectation when dealing with random operators, in order for $\chi_{\La_L}\mathfrak{C}_P\chi_{\La_L}$ to be trace class for all $L\in\N^*$. Under this assumption, we have that the Chern character associated to $P_{E,\la,\omega}$ is almost surely constant, by a standard application of the Birkhoff's ergodic theorem. 

We summarize these results in the following Proposition, which shows that the Chern character is well defined and deterministic. For the reader's convenience, we give a proof in Appendix \ref{appendix deterministic chern}.  

\begin{prop}\label{deterministic chern}
    Let $H_{\la,\omega}$ be the Anderson model given in \eqref{Anderson model}, satisfying hypothesis \ref{item: P1} and \ref {item: P4}. Let $E\in\R$ and assume that
    \begin{equation}\label{sobolev condition}
        \sum_{\ga\in\Ga} |\ga|\, \E \Big[\Tr_{\C^n}|P_{E,\la,\omega}(0,\ga)|^2\Big]^{\frac{1}{2}}<\infty.
    \end{equation}
    Then there exists $\Omega_0\subset \Omega$ with $\mathbb{P}(\Omega_0)=1$ such that:
    \begin{enumerate}
        \item \label{point 1 chern} $\chi_{\La_L}\mathfrak{C}_{P_{E,\la,\omega}}\chi_{\La_L}$ is trace class for all $\omega\in \Omega_0$ and $L\in\N^*;$
        \item\label{chern point 2} we have that
        \begin{align*}
            &C(P_{E,\la,\omega})=\E[C(P_{E,\la,\omega})]=2\pi\iu\,\E\Big[\Tr_{\C^n}\mathfrak{C}_{P_{E,\la,\omega}}(0,0)\Big]\\
            &=2\pi\iu\sum_{\ga,\xi\in\Ga}\E\Big[\Tr_{\C^n}\big(P_{E,\la,\omega}(0,\ga)P_{E,\la,\omega}(\ga,\xi)P_{E,\la,\omega}(\xi,0)\big)\Big](\ga\wedge \xi) \quad \textup{for all} \ \omega\in\Omega_0.
        \end{align*}
    \end{enumerate}
\end{prop}

Condition \eqref{sobolev condition} is closely related to the so-called \emph{Sobolev condition}, which appears in \cite{Bellissard van Elst Schulz-Baldes 1994} in the following form
\begin{equation*}
    \sum_{j=1}^2\E\Big[\Tr_{\C^n}\big|\big[X_j,P_{E,\la,\omega}]\big|^2(0,0)\Big]=\sum_{j=1}^2\sum_{\ga\in\Ga}|\ga_j|^2 \, \E\Big[\Tr_{\C^n}|P_{E,\la,\omega}(0,\ga)|^2\Big]<\infty.
\end{equation*}
Under the latter condition, the authors were able to show that the Chern character is well-defined and quantized, using a $C^*$-algebraic approach.
    
Another way of proving quantization, based on functional analytic methods, was shown in \cite{Aizenman Graf 1998} under the stronger decay assumption \eqref{sobolev 3 condition} below. The idea of the proof goes through the concept of index of pair of orthogonal projections given in \cite{Bellissard 1986, Nakamura Bellissard 1990, Avron Seiler Simon 1994}, whose definition we recall next.

\begin{dfn}
    Let $P,Q$ be orthogonal projections on a Hilbert space $\Hi$, whose difference $P-Q$ is a compact operator. The index of pair of projections $P,Q$ is defined by
    \begin{equation*}
        \textup{Index}(P,Q)=\dim\ker(P-Q-1)-\dim\ker(P-Q+1)\in\Z.
    \end{equation*}
\end{dfn}

Given $p=p_1a_1+p_2a_2\in\R^2\setminus\Ga$, consider the unitary operator $U_p\in \mathcal{B}(\ell^2(\Ga;\C^n))$ defined by
\begin{equation}\label{unitary}
    (U_p\psi)(\ga)=\eu^{-\iu\theta_p(\ga)}\psi(\ga) \quad \textup{for all} \ \psi\in\ell^2(\Ga;\C^n)
\end{equation}
where $\theta_p(\ga)=\textup{Arg}((\ga_1,\ga_2)-(p_1,p_2))$. Assuming fast enough decay in expectation of the spectral projection $P_{E,\la,\omega}$, the Chern character can be represented by the index of the pair of projections $P_{E,\la,\omega}$ and $U_p P_{E,\la,\omega}U_p^*$, as shown in \cite{ Aizenman Graf 1998} in the case of $\ell^2(\mathbb Z^2)$. This result readily generalizes to $\ell^2(\Ga;\C^n)$ with minimal changes, taking the following form:

\begin{prop}[\protect\cite{Aizenman Graf 1998}]
    Assume that $H_{\la,\omega}$ satisfies hypothesis \ref{item: P1} and \ref {item: P4}. Let $E\in\R$ and assume that
    \begin{equation}\label{sobolev 3 condition}
        \sum_{\ga\in\Ga} |\ga| \E \Big[\Tr_{\C^n}|P_{E,\la,\omega}(0,\ga)|^3\Big]^{\frac{1}{3}}<\infty.
    \end{equation}
    Given $p\in \R^2\setminus\Ga$ consider the unitary operator defined in \eqref{unitary}.
    Then there exists $\Omega_0\subset \Omega$ with $\mathbb{P}(\Omega_0)=1$ such that:
    \begin{enumerate}
        \item $\textup{Index}(U_pP_{E,\omega}U_p^*,P_{E,\omega})$ is well defined for all $\omega\in\Omega_0$ and 
        \begin{equation*}
            \textup{Index}(U_pP_{E,\la,\omega}U_p^*,P_{E,\la,\omega})=\E\Big[\textup{Index}(U_pP_{E,\la,\omega}U_p^*,P_{E,\la,\omega})\Big] \quad \textup{for all} \ \omega\in\Omega_0;
        \end{equation*}
        \item we have that
        \begin{equation*}
            C(P_{E,\la,\omega})=\E[C(P_{E,\la,\omega})]=\E\Big[\textup{Index}(U_pP_{E,\la,\omega}U_p^*,P_{E,\la,\omega})\Big]\in\Z \quad \textup{for all} \ \omega\in\Omega_0.
        \end{equation*}
    \end{enumerate}
\end{prop}

Notice that, if $E$ is in a gap of the spectrum of $H_{\la,\omega}$ for all $\omega\in\Omega$, then by the Combes-Thomas estimate, the Fermi projection $P_{E,\la,\omega}$ decays exponentially uniformly in $\omega\in\Omega$. Thus, in this case, the quantization result follows directly from \cite[Prop. 3 \& Rmk. 3]{Elgart Graf Schenker 2005} in the discrete setting (see also \cite[Prop. 2.13]{Marcelli Moscolari Panati 2023} in the continuum setting) \emph{without any ergodicity assumption}.

\subsection{Continuity of the Chern character in energy \& \mbox{disorder}}

In this section we show that the Chern character is continuous w.r.t. the energy parameter $E$ and the disorder parameter $\la$. The continuity in energy was first proved in \cite[Thm. 1 (iv)]{Bellissard van Elst Schulz-Baldes 1994} using a $C^*$-algebraic approach, and in \cite[Thm. 5]{Aizenman Graf 1998} using a more analytic approach. An important ingredient in the proof of \cite{Aizenman Graf 1998} is given by the continuity in energy of the integrated density of states (IDS), which is defined under assumptions \ref{item: P1},\ref{item: P4} as
\begin{equation*}
    N(E)=\E\Big[\Tr_{\C^n}P_{E,\la,\omega}(0,0)\Big].
\end{equation*}
It was shown in \cite{Deylon Souillard 1984}, under the additional assumption \ref{item: P2}, that the IDS has no atoms, \ie
\begin{equation*}
    \lim_{E'\to E^+}[N(E')-N(E)]=\E\Big[\Tr_{\C^n}P_{\set{E},\la,\omega}(0,0)\Big]=0
\end{equation*}
where $P_{\set{E},\la,\omega}=\chi_{\set{E}}(H_{\la,\omega}).$ Using this result the authors of \cite{Aizenman Graf 1998} prove that the Chern character is continuous in energy, assuming fast enough decay of the spectral projection, without any assumption on the regularity on the probability distribution. Their approach generalizes to our setting and can be stated as follows:

\begin{thm}[\protect{\cite[Thm. 5]{Aizenman Graf 1998}}]\label{chern enegy continuity}\label{continuity chern energy}
    Assume that $H_{\la,\omega}$ satisfies hypothesis \ref{item: P1},\ref{item: P2} and \ref {item: P4}. Let $I\subset\R$ be an interval and assume that
    \begin{equation*}
        \sup_{E\in I}\sum_{\ga\in\Ga}|\ga|\E\Big[\Tr_{\C^n}|P_{E,\la,\omega}(0,\ga)|^3\Big]^{\frac{1}{3}}<\infty.
    \end{equation*}
    Then the map
    \begin{equation*}
        I\ni E\mapsto \E[C(P_{E,\la,\omega})]\in \Z
    \end{equation*}
    is continuous. In particular the Chern character is constant for all energies in $I$.
\end{thm}

Next, we turn to the continuity of the Chern character in the disorder parameter. This was first proved in \cite[Thm. 1]{Richter Schulz-Baldes 2001} under the assumption of exponentially decaying Fractional Moment Bounds, that is, a regime of localization. Our next result, which is the similar to \cite[Thm.1 (ii)]{Richter Schulz-Baldes 2001}, shows that the assumption of exponential decay of Fractional Moment Bounds can be weakened. Namely,  following the approach in \cite{Aizenman Graf 1998}, we show that it is enough to have uniform decay of the projection in the form of \eqref{sobolev 3 condition} in a connected region of the energy-disorder plane, and the continuity of the projection w.r.t. $(E,\la)$ in the form of \eqref{ids disorder}. We will show, in Remark \ref{off diagonal} and in Proposition \ref{ids continuity disorder}, that the latter is a consequence of the regularity of the single-site probability distribution, without any assumption on the localization regime for a disorder parameter $\lambda>0$.

\begin{thm}\label{continuity chern disorder}
     Assume that $H_{\la,\omega}$ satisfies hypothesis \ref{item: P1} and \ref {item: P4}. Let $\mathcal{R}\subset \R\times [0,\infty)$ be a connected set. Assume that
    \begin{equation*}
        \sup_{(E,\la)\in \mathcal{R}}\sum_{\ga\in\Ga}|\ga|\E\Big[\Tr_{\C^n}|P_{E,\la,\omega}(0,\ga)|^3\Big]^{\frac{1}{3}}<\infty.
    \end{equation*}
    Moreover assume that there exist $C<\infty$, $\alpha_1,\alpha_2>0$ such that
    \begin{equation}\label{ids disorder}
        \sup_{\ga\in\Ga}\E\Big[\Tr_{\C^n}|(P_{E_1,\la_1,\omega}-P_{E_2,\la_2,\omega})(0,\ga)|\Big]\leq C( |E_1-E_2|^{\alpha_1}+|\la_1-\la_2|^{\alpha_2})
    \end{equation}
    for all $(E_1,\la_1),(E_2,\la_2)\in \mathcal{R}$. Then we have that the map
    \begin{equation*}
        \mathcal{R}\ni (E,\la)\mapsto \E[C(P_{E,\la,\omega})]\in \Z
    \end{equation*}
    is continuous. In particular the Chern character is constant for all $(E,\la)\in \mathcal{R}$.
\end{thm}

\begin{proof}
    By Proposition \ref{deterministic chern}-\eqref{chern point 2} we have that the expectation of the Chern character can be written, for $(E,\la)\in \mathcal{R}$, as
    \begin{equation*}
        \E[C(P_{E,\la,\omega})]=2\pi\iu\sum_{\ga,\xi\in\Ga}\E\Big[\Tr_{\C^n}\big(P_{E,\la,\omega}(0,\ga)P_{E,\la,\omega}(\ga,\xi)P_{E,\la,\omega}(\xi,0)\big)\Big](\ga\wedge \xi).
    \end{equation*}
    Let $(E_1,\la_1),(E_2,\la_2)\in \mathcal{R}$ then, omitting the dependence on $\omega$ of the spectral projection, we have
    \begin{align*}
        &(2\pi)^{-1}\big|\,\E[C(P_{E_1,\la_1})-C(P_{E_2,\la_2})]\,\big| \\
        &\leq \sum_{\ga,\xi\in\Ga}\E\Big[\Tr_{\C^n}|(P_{E_1,\la_1}-P_{E_2,\la_2})(0,\ga)P_{E_1,\la_1}(\ga,\xi)P_{E_1,\la_1}(\xi,0)|\Big]|\ga\wedge \xi| \\
        &+\sum_{\ga,\xi\in\Ga}\E\Big[\Tr_{\C^n}|P_{E_2,\la_2}(0,\ga)(P_{E_1,\la_1}-P_{E_2,\la_2})(\ga,\xi)P_{E_1,\la_1}(\xi,0)|\Big]|\ga\wedge \xi| \\
        &+ \sum_{\ga,\xi\in\Ga}\E\Big[\Tr_{\C^n}|P_{E_2,\la_2}(0,\ga)P_{E_2,\la_2}(\ga,\xi)(P_{E_1,\la_1}-P_{E_2,\la_2})(\xi,0)|\Big]|\ga\wedge \xi|.
    \end{align*}

    We  will focus only on the first addendum in the sum, since the other two terms can be treated similarly. Then, setting $P_{E_1,\la_1}-P_{E_2,\la_2}=\Delta P$, using the H\"{o}lder inequality first for the Schatten norm and then for the expectation, we get 
    \begin{align*}
        &\sum_{\ga,\xi\in\Ga}\E\Big[\Tr_{\C^n}|(P_{E_1,\la_1}-P_{E_2,\la_2})(0,\ga)P_{E_1,\la_1}(\ga,\xi)P_{E_1,\la_1}(\xi,0)|\Big]|\ga\wedge \xi| \\
         &\leq\sum_{\ga,\xi\in\Ga}\E\Big[\Tr_{\C^n}|\Delta P(0,\ga)|^3\Big]^{\frac{1}{3}}|\ga-\xi|\E\Big[\Tr_{\C^n}|P_{E_1,\la_1}(\ga,\xi)|^3\Big]^{\frac{1}{3}}|\xi|\E\Big[\Tr_{\C^n}|P_{E_1,\la_1}(\xi,0)|^3\Big]^{\frac{1}{3}}
    \end{align*}
    where we used $|\ga\wedge \xi|=|(\ga-\xi)\wedge \xi|\leq |\ga-\xi||\xi|$. By exploiting the monotonicity of the Shatten norm we get
    \begin{align*}
            \E\Big[\Tr_{\C^n}|\Delta P(0,\ga)|^3\Big]^{\frac{1}{3}}&\leq \E\Big[\big(\Tr_{\C^n}|\Delta P(0,\ga)|\big)^3\Big]^{\frac{1}{3}}=n\,\E\bigg[\bigg(\frac{1}{n}\Tr_{\C^n}|\Delta P(0,\ga)|\bigg)^3\bigg]^{\frac{1}{3}} \\
            &\hspace{-01cm}\leq n^{\frac{2}{3}}\E\Big[\Tr_{\C^n}|\Delta P(0,\ga)|\Big]^{\frac{1}{3}}\leq C^{\frac{1}{3}}n^{\frac{2}{3}}\left(|E_1-E_2|^{\frac{\alpha_1}{3}}+|\la_1-\la_2|^{\frac{\alpha_2}{3}}\right)
    \end{align*}
    where in the second inequality we used that $\Tr_{\C^n}|\Delta P(0,\ga)|\leq n\norm{\Delta P(0,\ga)}\leq n\norm{\Delta P}\leq n$. Finally, we obtain
    \begin{align*}
        &\sum_{\ga,\xi\in\Ga}\E\Big[\Tr_{\C^n}|(P_{E_1,\la_1}-P_{E_2,\la_2})(0,\ga)P_{E_1,\la_1}(\ga,\xi)P_{E_1,\la_1}(\xi,0)|\Big]|\ga\wedge \xi| \\
        &\leq C^{\frac{1}{3}}n^{\frac{2}{3}}\left(|E_1-E_2|^{\frac{\alpha_1}{3}}+|\la_1-\la_2|^{\frac{\alpha_2}{3}}\right)\bigg(\sum_{\ga\in\Ga}|\ga|\E\Big[\Tr_{\C^n}|P_{E_1,\la_1}(0,\ga)|^3\Big]^{\frac{1}{3}}\bigg)^2
    \end{align*}
    which shows the desired continuity result.
\end{proof}

The continuity in energy appearing in \eqref{ids disorder} can be shown, under a regularity assumption of the single-site probability distribution, as a consequence of the Wegner estimate, from which one derives the H\"{o}lder continuity in energy of the IDS. The latter can be expressed, under hypothesis \ref{item: P1},\ref{item: P4} and \ref{item: R1}, in the following way: 
\begin{equation}\label{ids continuity energy}
     \E\Big[\Tr_{\C^n}(P_{E_2,\la,\omega}-P_{E_1,\la,\omega})(0,0)\Big]\leq 2^{2-\tau}n\pi C_\tau (\rho)\bigg|\frac{E_1-E_2}{\la}\bigg|^\tau
\end{equation}
for all $E_1,E_2\in \R$ and $\la>0$. See \cite[Thm. 1.3 \& Rmk. 1.5]{Krishna 2007}. 

\begin{rmk}\label{off diagonal}
    The continuity condition \eqref{ids continuity energy} can be extended to the case with off-diagonal terms. Indeed let $\ga\in\Ga$, $\la>0$ and $E_1>E_2$, so that $P_{E_1,\la,\omega}-P_{E_2,\la,\omega}\geq 0$. By the inequality $\Tr_{\C^n}|A|\leq n(\Tr_{\C^n}|A|^2)^{\frac{1}{2}}$, which holds for all $A\in M_n(\C)$, we get
    \begin{align*}
        &\E\Big[\Tr_{\C^n}|(P_{E_1,\la,\omega}-P_{E_2,\la,\omega})(0,\ga)|\Big]\leq n  \E\Big[(\Tr_{\C^n}|(P_{E_1,\la,\omega}-P_{E_2,\la,\omega})(0,\ga)|^2)^{\frac{1}{2}}\Big]\nonumber\\ 
        &=n\E\Bigg[\bigg(\sum_{i=1}^n\sum_{j=1}^n|\langle\delta_{0,i},(P_{E_1,\la,\omega}-P_{E_2,\la,\omega})\delta_{\ga,j}\rangle|^2\bigg)^{\frac{1}{2}}\Bigg] \nonumber \\
        & \leq n\E\Bigg[\bigg(\sum_{i=1}^n\sum_{j=1}^n\langle \delta_{0,i},(P_{E_1,\la,\omega}-P_{E_2,\la,\omega})\delta_{0,i}\rangle\langle \delta_{\ga,j},(P_{E_1,\la,\omega}-P_{E_2,\la,\omega})\delta_{\ga,j}\rangle\bigg)^{\frac{1}{2}}\Bigg] \nonumber \\
        &= n \E\left[\left(\Tr_{\C^n}(P_{E_1,\la,\omega}-P_{E_2,\la,\omega})(0,0)\right)^{\frac{1}{2}}\left(\Tr_{\C^n}(P_{E_1,\la,\omega}-P_{E_2,\la,\omega})(\ga,\ga)\right)^{\frac{1}{2}}\right]\nonumber \\
        & \leq n\E\Big[\Tr_{\C^n}(P_{E_1,\la,\omega}-P_{E_2,\la,\omega})(0,0)\Big]^{\frac{1}{2}}\E\Big[\Tr_{\C^n}(P_{E_1,\la,\omega}-P_{E_2,\la,\omega})(\ga,\ga)\Big]^{\frac{1}{2}} \nonumber\\
        &= n \E\Big[ \Tr_{\C^n}(P_{E_1,\la,\omega}-P_{E_2,\la,\omega})(0,0)\Big]\leq 2^{2-\tau} n^2\pi C_\tau(\rho)\left(\frac{E_1-E_2}{\la}\right)^\tau
    \end{align*}
    where in the second inequality we used Cauchy-Schwarz inequality in the Hilbert space $\ell^2(\Ga;\C^n)$, in the third inequality we used Cauchy-Schwarz inequality for the expectation, in the last equality we used the ergodicity property, which holds under assumptions \ref{item: P1},\ref{item: P4}, while in the last inequality we used \eqref{ids continuity energy}.
\end{rmk}

The continuity in disorder in condition \eqref{ids disorder} holds in our setting, either  for energies in a spectral gap, or for disorder parameters away from zero, under a regularity condition on the single-site probability distribution, following the arguments in \cite[Thm. 1.2]{Hislop Klopp Schenker 2005} and \cite[Lemma 5.5]{Germinet Klein Schenker 2009}. We summarize these results below, adapting them to our setting. We recall that, given $E\in \G_i(0)$ with $i\in \set{1,\ldots,N-1}$, $\la_0(E)$ denotes the smallest value $\la$ such that $E\in \sigma(H_0)+\la[-a,b]$.

\begin{prop}\label{ids continuity disorder}
    Let $H_{\la,\omega}$ be the operator given in \eqref{Anderson model} satisfying \ref {item: P4}. 
    \begin{enumerate}
        \item\label{point 1 continuity disorder} Assume \ref{item: P2},\ref {item: P3} hold. Let $K=[E_1,E_2]\times [0,\la_0]\subset \G_i\times [0,+\infty)$ for $i\in \set{1,\ldots,N-1}$ where $\la_0<\min\set{\la_0(E_1),\la_0(E_2)}$ so that $[E_1,E_2]\subset \G_i(\la)$ for all $\la\in [0,\la_0]$. Then there exists $C_K=C_{H_0,a,b,n,K}>0, \mu_K=\mu_{H_0,a,b,K}<\infty$ such that for all $E\in[E_1,E_2]$ and $\omega\in\Omega$ we have
        \begin{gather*}
            \Tr_{\C^n}|P_{E,\la,\omega}(\ga,\xi)|\leq C_K\eu^{-\mu_K|\ga-\xi|} \\
            \Tr_{\C^n}|(P_{E,\la',\omega}-P_{E,\la'',\omega})(\ga,\xi)|\leq C_K|\la'-\la''|\eu^{-\mu_K|\ga-\xi|}
        \end{gather*}
        for all $\la,\la',\la''\in[0,\la_0]$, $\ga,\xi\in\Ga$.
        \item\label{point 2 continuity disorder}
        Assume  \ref{item: P1},\ref{item: R1} hold. Let $K=[E_1,E_2]\times [\la_1,\la_2]\subset \R\times (0,\infty)$. Then there exists $C_K=C_{H_0,a,b,n,\rho,\tau,K}>0$ such that for all $E\in [E_1,E_2]$ we have
        \begin{equation*}
            \sup_{\ga\in\Ga}\E\Big[\Tr_{\C^n}|(P_{E,\la',\omega}-P_{E,\la''.\omega})(0,\ga)|\Big]\leq C_K|\la'-\la''|^{\frac{\tau}{\tau+2}}
        \end{equation*}
        for all $\la',\la''\in [\la_1,\la_2]$.
    \end{enumerate}
\end{prop}
    
\begin{rmk}
    The restriction $\la>0$ in Proposition \ref{ids continuity disorder}-\eqref{point 2 continuity disorder} comes from the use, in the proof, of the continuity in energy of the IDS \eqref{ids continuity energy}, in which the prefactor appearing in front of $|E_1-E_2|^\tau$ diverges as $\la\to0$. This restriction can be removed when dealing with periodic systems, if we assume that the IDS of the unperturbed Hamiltonian $H_0$ is $\alpha$-H\"{o}lder continuous in energy, which holds for example in the case of the Haldane model with $\alpha=1$, and the single-site probability distribution has bounded density. Indeed, it was shown in \cite{Combes Hislop Klopp 2003} and \cite{Schenker 2004} that, under the aforementioned assumptions, the IDS of $H_{\la,\omega}$ is H\"{o}lder continuous in energy, with constants which remain bounded as $\la\to 0$. Hence Proposition \ref{ids continuity disorder}-(2) holds for all $\la\geq 0$ for all energies, even outside the internal spectral gaps.
\end{rmk}

As a consequence of Theorem \ref{continuity chern disorder} and Proposition \ref{ids continuity disorder}-\eqref{point 1 continuity disorder}, we have that the Chern number is a topological invariant which is robust w.r.t. disorder in 2D, in the sense that it is still well defined in the random setting and retains the same value of the unperturbed Hamiltonian, for small values of disorder.  Indeed, we have:
    
\begin{crl}\label{robusteness}
    Assume that $H_{\la,\omega}$ satisfies hypothesis \ref{item: P1}-\ref {item: P4}. Then, given $\la<\frac{|\G_i|}{a+b}$, with $i\in\set{1,\ldots,N-1}$, it holds that
    \begin{equation*}
        \E[C(P_{E,\la,\omega})]=\E[C(P_{E,\la=0})]=C(P_{E,\la=0}) \quad \textup{for all} \ E\in\G_i(\la).
    \end{equation*}
\end{crl} 

\begin{proof}[Proof of Corollary \ref{robusteness}]
 Given $\la<\frac{|\G_i|}{a+b}$ and $E\in \G_i(\la)$, consider $\mathcal{R}=\set{E}\times [0,\la]$. By Proposition \ref{ids continuity disorder}-\eqref{point 1 continuity disorder}, we have that conditions \eqref{sobolev 3 condition} and \eqref{ids disorder} are satisfied with $\alpha_2=1$, so we are in position to apply Theorem \ref{continuity chern disorder}, from which the thesis follows.
\end{proof}
    
\begin{proof}[Proof of Proposition \ref{ids continuity disorder}]
    For the clarity of the exposition, and to see how all the parameters of the model appear in the estimates, we detail all the steps.

    1. Let $E_1,E_2\in \G_i$, with $i\in \set{1,\ldots,N-1}$, let $\la_0<\min\set{\la_0(E_1),\la_0(E_2)}$, then $[E_1,E_2]\subset \G_i(\la)$ for all $\la\in [0,\la_0]$. Consider $S_\alpha$ as defined in Proposition \ref{combes thomas}, and let $\alpha>0$ be such that $2 S_\alpha=\min\set{\dist(E_1,\sigma(H_0)+\la_0[-a,b]),\dist (E_2,\sigma(H_0)+\la_0[-a,b])}$. Given $E\in [E_1,E_2]$ consider the contour in the complex plane $\Ci(E)=\Ci_1\cup \Ci_2^+\cup \Ci_2^-\cup \Ci_3$ where
    \begin{gather*}
        \Ci_1=\set{E+\iu\eta: \eta\in (-2S_\alpha,2S_\alpha)} \\
        \Ci_2^{\pm}=\set{u\pm\iu2S_\alpha: u\in[\inf\sigma(H_{\la_0})-2 S_\alpha,E]} \\
        \Ci_3=\set{\inf\sigma(H_{\la_0})-2S_\alpha+\iu\eta:\eta\in (-2S_\alpha,2S_\alpha)}.
    \end{gather*}
    By \cite[Thm. V.4.10]{Kato 1995} we have $\sigma(H_{\la,\omega})\subset \sigma (H_0)+\la [-a,b]$ for all $\omega\in\Omega$, and so it holds that $\dist(z,\sigma(H_{\la,\omega}))\geq 2 S_\alpha$ for all $z\in \Ci(E)$, $E\in [E_1,E_2]$, $\la\in [0,\la_0]$ and $\omega\in\Omega$. Hence by the Combes-Thomas estimate \eqref{ct1} we get
    \begin{equation*}
        \norm{(H_{\la,\omega}-z)^{-1}(\ga,\xi)}\leq \frac{1}{S_\alpha}\eu^{-\alpha|\ga-\xi|} \quad \textup{for all} \ z\in \Ci(E), E\in [E_1,E_2], \la\in [0,\la_0],\omega\in\Omega.
    \end{equation*}
    Using the Riesz integral formula we have
    \begin{equation*}
        \norm{P_{E,\la,\omega}(\ga,\xi)}\leq \frac{1}{2\pi}\int_{\Ci(E)}\di z \, \norm{(H_{\la,\omega}-z)^{-1}(\ga,\xi)}\leq \frac{|\Ci(E_2)|}{2\pi S_\alpha}\eu^{-\alpha|\ga-\xi|}
    \end{equation*}
    for all $E\in [E_1,E_2]$, $\la\in [0,\la_0]$, $\omega\in \Omega$. Moreover, using the geometric resolvent identity together with the Riesz formula, it holds that
    \begin{align*}
        &\norm{(P_{E,\la',\omega}-P_{E,\la'',\omega})(\ga,\xi)}\\
        &\leq \frac{\norm{V_\omega}_\infty}{2\pi}|\la'-\la''|\int _{\Ci(E)}\di z \, \sum_{\zeta\in \Ga}\norm{(H_{\la',\omega}-z)(\ga,\zeta)}\norm{(H_{\la'',\omega}-z)(\zeta,\xi)} \\
        &\leq \frac{\norm{V_\omega}_\infty |\Ci(E_2)|}{2\pi (S_\alpha)^2}|\la'-\la''|\eu^{-\frac{\alpha}{2}|\ga-\xi|}\sum_{\zeta\in \Ga}\eu^{-\alpha|\zeta|}=C|\la'-\la''|\eu^{-\frac{\alpha}{2}|\ga-\xi|}
    \end{align*}
    for all $E\in [E_1,E_2]$, $\la',\la''\in [0,\la_0]$, $\omega\in\Omega$. The thesis follows from the inequality $\Tr_{\C^n}|A|\leq n\norm{A}$ which holds for all matrices $A\in M_n(\C)$.
    
    2. Let $E_1<E_2$ and $0<\la_1<\la_2$. In order to prove the statement it suffices to consider the case $\la_2-\la_1<1$. Let $E\in [E_1,E_2]$, let $\la',\la''\in [\la_1,\la_2]$ and define $\delta=|\la'-\la''|^{\alpha}$ where $\alpha\in (0,1)$ will be chosen later. In the following, we will omit the dependence on $\omega$.
    Notice that if $E-\delta<\inf\sigma(H_{\la_2})=\alpha_0-\la_2a$ we have that
    \begin{equation*}
        P_{E,\la'}-P_{E,\la''}=(P_{E,\la'}-P_{E-\delta,\la'})+(P_{E-\delta,\la''}-P_{E,\la''}).
    \end{equation*}
    By the same argument of Remark \ref{off diagonal}, we get for $\la^\#\in\set{\la',\la''}$ that
    \begin{align*}
        &\E\Big[\Tr_{\C^n}|(P_{E,\la^\#}-P_{E-\delta,\la^\#})(0,\ga)|\Big]\leq \frac{2^{2-\tau} n^2\pi C_\tau(\rho)\delta^\tau}{(\la^{\#})^\tau}\leq \frac{2^{2-\tau} n^2\pi C_\tau(\rho)}{\la_1^\tau}|\la'-\la''|^{\alpha\tau}.
    \end{align*}
    Now assume that $E-\delta>\inf\sigma(H_{\la_2})$. Consider $f\in C^3(\R)$ defined by
    \begin{equation*}
        f(x)=
        \begin{cases}
            0 &\textup{if} \ x\leq 0 \\
            1 & \textup{if} \ x\geq1
        \end{cases}.
    \end{equation*}
    Define $\displaystyle g(x)=f\Big(\frac{x-\inf{\sigma(H_{\la_2})+\delta}}{\delta}\Big)f\Big(\frac{E-x}{\delta}\Big)$. Note that
    \begin{equation*}
        g(x)=
        \begin{cases}
            1 & \textup{if} \ \inf\sigma(H_{\la_2})\leq x\leq E-\delta \\
            0 & \textup{if} \ x\geq E \ \lor \ x\leq \inf\sigma(H_{\la_2})-\delta
        \end{cases}
    \end{equation*}
    and $\norm{g^{(k)}}_\infty\leq c_k\delta^{-k}$ for all $k=1,2,3$. Proceeding as in \cite[Eq. 5.11]{Hislop Klopp Schenker 2005} we write
    \begin{align}\label{equation}
        P_{E,\la'}-P_{E,\la''}=\Big(P_{E,\la'}-g(H_{\la'})\Big)+\Big(g(H_{\la'})-g(H_{\la''})\Big)+\Big(g(H_{\la''})-P_{E,\la''}\Big).
    \end{align}
    By construction we have, for $\la^\#\in\set{\la',\la''}$, that
    \begin{equation}\label{positivity inequality}
        0\leq P_{E,\la^{\#}}-g(H_{\la^\#})\leq P_{E,\la^\#}-P_{E-\delta,\la^\#}.
    \end{equation}
    As before we can use the same argument in Remark \ref{off diagonal}, along with \eqref{positivity inequality}, to estimate the first and the third term in the right hand side of \eqref{equation}. In this way we obtain for $\la^\#\in \set{\la',\la''}$ that
    \begin{align*}
        &\E\Big[\Tr_{\C^n}|(P_{E,\la^\#}-g(H_{\la^\#}))(0,\ga)|\Big]
        \leq \frac{2^{2-\tau} n^2\pi C_\tau(\rho)}{\la_1^\tau}|\la'-\la''|^{\alpha\tau}.
    \end{align*}
    In order to estimate the second term in \eqref{equation} we use the  Helffer–Sj\"{o}strand formula (see \eg \cite{Davies 1995})
    \begin{equation*}
        g(H_{\la^\#})=\frac{1}{\pi}\int_{\C}\di z  \, \partial_{\overline{z}}\tilde{g}(z)(H_{\la^\#}-z)^{-1}
    \end{equation*}
    for $\la^\#\in\set{\la',\la''}$, where $\tilde{g}:\C\to\C$ is a quasi analytic extension of $g$ of order 2, \ie
    \begin{equation*}
        \tilde{g}(z)\equiv \tilde{g}(x+\iu y)=\sum_{k=0}^2 g^{(k)}(x)\frac{(\iu y)^k}{k!}\sigma(x,y)
    \end{equation*}
    with $\sigma(x,y)=\tau(y/\langle x\rangle)$, where $\tau:\R\to\R$ is a smooth function such that
    \begin{equation*}
        \tau(s)=
        \begin{cases}
            1 & \textup{if} \ |s|<1 \\
            0 &\textup{if} \ |s|>2
        \end{cases}
    \end{equation*}
    and $\langle x \rangle =(1+|x^2|)^{\frac{1}{2}}$. Notice that
    \begin{align*}
        &\partial_{\overline{z}}\tilde{g}(z)=\frac{1}{2}\bigg[ \frac{\partial}{\partial x}\tilde{g}(z)+\iu \frac{\partial}{\partial y}\tilde{g}(z)\bigg] \\
        &=\frac{1}{2}\sum_{k=0}^2g^{(k)}(x)\frac{(\iu y)^k}{k!}\bigg(\frac{\partial}{\partial x}\sigma(x,y)+\iu\frac{\partial}{\partial y}\sigma(x,y)\bigg)+\frac{1}{2}g^{(3)}(x)\frac{(\iu y)^2}{2!}\sigma(x,y).
    \end{align*}
    The modulus of the complex derivative can be estimated in the following way
    \begin{align*}
        |\partial_{\overline{z}}\tilde{g}(z)|\leq C\sum_{k=0}^2|g^{(k)}(x)|\frac{|y|^k}{k!}\chi_A(x,y)+\frac{1}{2} |g^{(3)}(x)|\frac{|y|^2}{2!}\chi_B(x,y)
    \end{align*}
    where $A=\set{(x,y):\langle x \rangle< |y|<2\langle x \rangle}$ and $B=\set{(x,y): |y|<2\langle x\rangle}$.

    Using the geometric resolvent identity we get
    \begin{equation*}
        g(H_{\la'})-g(H_{\la''})=\frac{\la''-\la'}{\pi}\int_\C\di z \, \partial_{\overline{z}}\tilde{g}(z)(H_{\la'}-z)^{-1}V_\omega(H_{\la''}-z)^{-1}.
    \end{equation*}
    Using the inequality $\Tr_{\C^n}|A|\leq n \norm{A}_{M_n(\C)}$, which holds for all $A\in M_n(\C)$, we get
    \begin{equation}\label{integral}
        \Tr_{\C^n}|(g(H_{\la'})-g(H_{\la''}))(0,\ga)|\leq n\frac{|\la'-\la''|}{\pi}\norm{V_\omega}_\infty\int_{\C}\di x\, \di y \, |\partial_{\overline{z}}\tilde{g}(z)| \, |y|^{-2} 
    \end{equation}
    where we used the estimate $\norm{(H_{\la^\#}-z)^{-1}}\leq |\textup{Im}z|^{-1}$ for $\la^\#\in\set{\la',\la''}$.

    In order to estimate the integral in \eqref{integral} we study separately the contributions in $\partial_{\overline{z}}\tilde{g}(z)$ coming from the derivatives $g^{(k)}(x)$ for $k=0,1,2, 3$.
    
    Starting with $k=3$ we have that
    \begin{align*}
        \int_\C \di x \, \di y |g^{(3)}(x)|\chi_B(x,y) &\leq \frac{2c_3}{\delta^3}\bigg(\int_{E-\delta}^E\di x \, \int_{0}^{2\langle x\rangle}\di y +\int_{\inf\sigma(H_{\la_2})-\delta}^{\inf\sigma(H_{\la_2})} \di x\,\int_0^{2\langle x\rangle}\di y\bigg)\\
        &\leq \frac{4 c_3}{\delta^2}\Big(4+\max\set{|E_1|,|E_2|}+|\inf\sigma(H_{\la_2})|\Big)
    \end{align*}
    where we used $\delta<1$. For $k=2$ we have
    \begin{align*}
        \int_\C \di x \, \di y |g^{(2)}(x)|\chi_A(x,y)&\leq \frac{2c_2}{\delta^2}\bigg(\int_{E-\delta}^E\di x \, \int_{\langle x\rangle}^{2\langle x\rangle}\di y +\int_{\inf\sigma(H_{\la_2})-\delta}^{\inf\sigma(H_{\la_2})} \di x\,\int_{\langle x\rangle}^{2\langle x\rangle}\di y\bigg)\\
        &\leq \frac{2 c_2}{\delta}\Big(4+ \max\set{|E_1|,|E_2|}+|\inf\sigma(H_{\la_2})|\Big)
    \end{align*}
    where we used $\delta<1$. For $k=1$ we have
    \begin{align*}
        &\int_\C \di x \, \di y |g^{(1)}(x)||y|^{-1}\chi_A(x,y) \\
        &\leq \frac{2c_1}{\delta}\bigg(\int_{E-\delta}^E\di x \, \int_{\langle x\rangle}^{2\langle x\rangle}\di y \, y^{-1} +\int_{\inf\sigma(H_{\la_2})-\delta}^{\inf\sigma(H_{\la_2})} \di x\,\int_{\langle x\rangle}^{2\langle x\rangle}\di y\, y^{-1}\bigg)= 4c_1\ln2.
    \end{align*}
    For $k=0$ we have
    \begin{align*}
        &\int_\C \di x \, \di y |g(x)||y|^{-2}\chi_A(x,y) \leq 2\int_{\inf\sigma(H_{\la_2})-\delta}^E\di x \, \int_{\langle x \rangle}^{2\langle x\rangle}\di y\, y^{-2}\\
        &=\int_{\inf\sigma(H_{\la_2})-\delta}^E\di x\,\langle x\rangle^{-1}\le (E-\inf\sigma(H_{\la_2})+\delta)\leq (E_2-\inf\sigma(H_{\la_2})+1)
    \end{align*}
    where in the last inequality we used $\delta< 1$.
    
    Therefore there exist a constant $C=C_{H_0,a,E_1,E_2,\la_2}>0$ such that we get
    \begin{equation*}
        \Tr_{\C^n}|(g(H_{\la'})-g(H_{\la''}))(0,\ga)|\leq \frac{n}{\pi}\norm{V_\omega}_\infty C|\la'-\la''|^{1-2\alpha}. 
    \end{equation*}
    Combining everything together we get
    \begin{align*}
        &\E\Big[\Tr_{\C^n}|(P_{\omega,E,\la'}-P_{\omega,E,\la''})(0,\ga)|\Big]\\
        &\leq  \frac{2^{2-\tau}n^2\pi C_\tau(\rho)}{\la_1^\tau}|\la'-\la''|^{\alpha\tau}+\frac{n}{\pi}\norm{V_\omega}_\infty C|\la'-\la''|^{1-2\alpha} \\
        &=\bigg( \frac{2^{2-\tau} n^2\pi C_\tau(\rho)}{\la_1^\tau}+ \frac{n}{\pi}\norm{V_\omega}_\infty C\bigg)|\la'-\la''|^{\frac{\tau}{\tau+2}}
    \end{align*}
    where we have chosen $\alpha=\frac{1}{\tau+2}$. Thus the proof is concluded.
\end{proof}

\goodbreak

\section{Dynamical localization regimes}\label{localization section}

In this section we state the dynamical localization results that apply to our setting in three regimes: for all energies  at strong disorder, shown in \cite{Aizenman Molchanov 1993} and \cite{Aizenman Graf 1998}, for energies away from the unperturbed spectrum, shown in \cite{Aizenman 1994}, and for internal and external band edges, shown in \cite{Figotin Klein 1994} using the technique of Multiscale Analysis (MSA). These results are well-known. In the context of disordered Chern insulators, the first two regimes haven been well studied, while the third regime, which requires MSA techniques, is less so (see e.g. \cite{De Nittis Drabkin Schulz-Baldes 2015, Becker Han 2022} for the discrete setting, and \cite{Germinet Klein Schenker 2007} for continuum models). In the following Theorem, we state under which requirements these localization results hold for disordered Chern insulators.

\begin{thm}\label{DL regimes}
    Assume that $H_{\la,\omega}$ satisfies \ref{item: P2},\ref{item: P4} and \ref{item: R1}. Then:
    \begin{enumerate}
        \item\label{strong disorder regime thm} (Strong disorder regime) 
        For $s\in (0,\tau), \mu>0$ define
        \begin{equation}\label{strong disorder threshold}
            \la_{s,\mu,\rho}(H_0)=\bigg[C_{s,\tau}(\rho)\sup_{(\ga,i)\in\Ga\times\set{1,\ldots,n}}\sum_{(\xi,j)\neq (\ga,i)} |\langle\delta_{\ga,i},H_0\,\delta_{\xi,j}\rangle|^s \eu^{\mu|\ga-\xi|}\bigg]^{\frac{1}{s}}
        \end{equation}
        where $C_{s,\tau}(\rho)=\frac{\tau C_\tau(\rho)^{\frac{s}{\tau}}}{\tau-s}$ and define $\la_\rho(H_0)=\inf_{s\in (0,\tau)}\inf_{\mu>0}\la_{s,\mu,\rho}(H_0)$. If $\la>\la_\rho(H_0)$ there exist $s\in (0,\tau), \mu>0$ such that
        \begin{equation}\label{FMM}
            \sup_{\eta\neq 0}\E\Big[|\langle\delta_{\ga,i},(H_{\la,\omega}-(E+\iu\eta))^{-1} \delta_{\xi,j}\rangle|^s \Big]\leq \frac{C_{s,\tau}(\rho)}{\la^s-\la_{s,\mu,\rho}(H_0)^s}\eu^{-\mu|\ga-\xi|}
        \end{equation}
        for all $\ga,\xi\in\Ga$, $i,j\in\set{1,\ldots,n}$ and $E\in\R$.
        
        \item\label{point 2 DL} (Localization away from the unperturbed spectrum) Assume, in addition, that $H_{\la,\omega}$ satisfies \ref{item: P1} and \ref{item: P3}. Let $E\in \G_i$ with $i\in\set{1,\ldots,N}$ and denote $\Delta(E)=\dist(E,\sigma(H_0))>0$. Define for $s\in(0,\tau)$
        \begin{equation}\label{sigma-moment regularity}
            D_{s,1}(\rho)=\displaystyle\sup_{z\in\C}\bigg(\int\frac{|v|^s}{|v-z|^s}\rho(\di v)\bigg)\bigg/\bigg(\int\frac{1}{|v-z|^s}\rho(\di v)\bigg).
        \end{equation}
        Consider $S_\alpha$ as defined in Proposition \ref{combes thomas}, with $\alpha>0$ such that $2 S_\alpha\leq \frac{|\G_i|}{2}$. Then given $s\in (0,\tau)$ there exists $C_{s,\alpha}=\displaystyle \frac{n|\G_i|^2}{2}\bigg(1+\frac{32}{s^2\alpha^2}\bigg)$ such that, if $(E,\la)\in \G_i\times \R_+$ satisfies
        \begin{equation}\label{weak far}
            \la_0(E)\leq \la< \frac{\Delta(E)^{1+\frac{2}{s}}}{(C_{s,\alpha}D_{s,1}(\rho))^{\frac{1}{s}}}
        \end{equation} 
        then $E\in\sigma(H_{\la,\omega})$ and there exist $C(E,\la)=C_{H_0,\rho,s}(E,\la)$ and $\mu=\mu_{s,\alpha}$ such that
        \begin{equation}\label{FMM weak1}
        \sup_{\eta\neq 0}\E\Big[|\langle\delta_{\ga,i},(H_{\la,\omega}-(E+\iu\eta))^{-1} \delta_{\xi,j}\rangle|^s \Big]\leq C(E,\la)\eu^{-\mu\Delta(E)|\ga-\xi|}
        \end{equation}
        for all $\ga,\xi\in\Ga$, $i,j\in\set{1,\ldots,n}$.

        \item\label{point band edge} (Band edge regime) Assume, in addition, that $H_{\la,\omega}$ satisfies \ref{item: P1},\ref{item: P3}  and \ref{item: R2}. Then there exists $\delta(\la)=\delta_{H_0,a,b,\beta,\tau,\rho,n}(\la)>0$ such that $H_{\la,\omega}$ exhibits DL in the following closed neighborhoods of the external band edges $\alpha_1-a\la,\beta_N+b\la$
        \begin{align*}
            &[\alpha_1-a\la,\alpha_1-a\la+\delta(\la)], &\quad [\beta_N+b\la-\delta(\la),\beta_N+b\la].
        \end{align*}
         Moreover, if \ref{item: P1} holds with $B\in 2\pi \Q$, then for all $\la<\frac{|\G_i|}{a+b}$ with $i\in\set{1,\ldots, N-1}$, we have that the gap $\G_i(\la)$ remains open and there exist $\delta_i(\la)=\delta_{|\G_i|,H_0,a,b,\beta,\tau,\rho,n}(\la)>0$ such that $H_{\la,\omega}$ exhibits DL in the following closed neighborhoods of the internal band edges $\beta_{i}+b\la,\alpha_{i+1}-a\la$
        \begin{align*}
            &[\beta_i+b\la-\delta_i(\la),\beta_i+b\la], &\quad [\alpha_{i+1}-a\la,\alpha_{i+1}-a\la+\delta_i(\la)].
        \end{align*}
    \end{enumerate}
\end{thm}

\begin{rmk}
    In Theorem \ref{DL regimes}-\eqref{point band edge}, we have an explicit dependence of $\delta$ on $\la$. More precisely, given $\eps\in(0,1)$ there exists $C_\eps=C_{H_0,a,b,\beta,\tau,\rho,n,\eps}>0$ such that
    \begin{gather*}
        \delta(\la)=C_\eps\min\set{1,\la^{\frac{\beta}{(\beta-2)(1-\eps)}}},\quad \delta_i(\la)=C_\eps\min\set{1,\la^{\frac{\beta}{(\beta-2)(1-\eps)}},\left(\frac{|\G_i|}{a+b}-\la\right)^{\frac{1}{1-\eps}}}.
    \end{gather*}
    for $i\in\set{1,\ldots,N-1}$. Notice that, if $\rho$ admits exponentially decaying tails, then we can take $\beta\to \infty$, yielding a linear dependence of $\delta$ on $\la$ if $\eps\to 0$, for sufficiently small values of $\la$. See Appendix \ref{band edge region} for a detailed proof.
    The problem of understanding the shape of the curve that separates the regions of DL and DD has been investigated in several papers, see \eg \cite{Aizenman 1994, Wang 2001, Germinet Klein Schenker 2007,Elgart 2008}.
\end{rmk}

\begin{rmk}
    The condition appearing in \eqref{FMM},\eqref{FMM weak1} is the outcome of the Fractional moment method (FMM), and is a sufficient condition for DL, under the assumption of bounded density, as shown in \cite[Appendix A]{Aizenman Schenker Friedrich Hundertmark 2001} and \cite[Thm. 2.2]{Aizenman Elgart Naboko Schenker Stolz 2006}. If we only assume $\tau$-H\"{o}lder regularity of the distribution \ref{item: R1} then we obtain a Wegner estimate of the form \eqref{wegner estimate}, and we can apply \cite[Prop. 5.2]{Elgart Tautenhahn Veselic 2011} to deduce that the outcome of FMM implies the outcome of MSA, which in turn is equivalent to the condition of DL.
\end{rmk}
    
\subsection{Vanishing of the Chern character in the strong disorder regime}\label{strong disorder regime}

Condition $\la>\la_\rho(H_0)$ describes the so-called \emph{strong disorder regime}, in which it is possible to prove the outcome of FMM \eqref{FMM}. The latter condition implies exponential decay of the expectation of the spectral projection, as proved in \cite[Thm. 3]{Aizenman Graf 1998}. More precisely, for all $\la>\la_\rho(H_0)$ there exists $C=C_{H_0,a,b,\rho,\tau,\la}<\infty$, $\mu>0$ such that
    \begin{equation}\label{aizenman graf exp decay}
        \sup_{E\in\sigma(H_{\la,\omega})}\E\Big[\Tr_{\C^n}|P_{E,\la,\omega}(\ga,\xi)|\Big]\leq C\eu^{-\mu|\ga-\xi|}.
    \end{equation}
Consequently the Chern character associated to $P_{E,\la,\omega}$ vanishes for all energies in the spectrum in the strong disorder regime. We have:

\begin{prop}\label{chern zero strong disorder}
        Assume that $H_{\la,\omega}$ satisfies hypothesis \ref{item: P1},\ref{item: P2},\ref {item: P4} and \ref{item: R1}. Then, for all $\la>\la_\rho(H_0)$, it holds that
        \begin{equation*}
            \E[C(P_{E,\la,\omega})]=0 \quad \textup{for all} \ E\in\sigma(H_{\la,\omega}).
        \end{equation*}
\end{prop}

\begin{proof}
    Given $\la>\la_\rho(H_0)$, consider $E_-\in \G_0(\la)$ and $E_+\in \G_N(\la)$. By \eqref{aizenman graf exp decay} we have that
    \begin{equation*}
        \sup_{E\in[E_-,E_+]}\sum_{\ga\in\Ga}|\ga|\E\Big[\Tr_{\C^n}|P_{E,\la,\omega}(0,\ga)|^3\Big]^{\frac{1}{3}}<\infty.
    \end{equation*}
    Therefore, by Theorem \ref{chern enegy continuity}, it holds that $\E[C(P_{E,\la,\omega})]=\E[C(P_{E_\pm,\la,\omega})]=0$ for all  $E\in\sigma(H_{\la,\omega})$, because $P_{E_-,\la,\omega}=0$, $P_{E_+,\la,\omega}=\Id$ and $C(0)=0=C(\Id)$.
\end{proof}

\goodbreak

\subsection{Localization away from the unperturbed spectrum}\label{section localization away}

The localization regime away from the unperturbed spectrum was investigated in \cite{Aizenman 1994}. To study it, one defines 
for $s\in (0,\tau)$, $\mu>0$
\begin{gather}\label{sigma moment regular}
    \widehat{\la}_{s,\mu,\rho}(E)=\left[D_{s,1}(\rho)\sup_{\eta\neq 0}\sup_{\substack{\ga\in\Ga\\i\in\set{1,\ldots,n}}}\sum_{\substack{\xi\in\Ga \\ j\in\set{1,\ldots,n}}}|\langle\delta_{\ga,i},(H_0+(E+\iu\eta))^{-1}\delta_{\xi,j}\rangle|^s\eu^{\mu|\ga-\xi|}\right]^{-\frac{1}{s}}
\end{gather}
and $\widehat{\la}_{\rho}(E)=\sup_{s\in (0,\tau)}\sup_{\mu>0}\widehat{\la}_{s,\mu,\rho}(E).$ Under assumptions \ref{item: P2} (or exponentially decaying off-diagonal matrix elements), \ref{item: P4} and \ref{item: R1} the author showed that, if $(E,\la)\in \R\times \R_+$ satisfies $\la<\widehat{\la}_\rho(E)$, then there exist $s\in (0,\tau)$ and $\mu>0$ such that
\begin{equation}\label{FMM weak2}
    \sup_{\eta\neq 0}\E\Big[|\langle\delta_{\ga,i},(H_{\la,\omega}-(E+\iu\eta))^{-1} \delta_{\xi,j}\rangle|^s \Big]\leq \frac{1}{D_{s,1}(\rho)(\widehat{\la}_{s,\mu,\rho}(E)^s-\la^s)}\eu^{-\mu|\ga-\xi|}
\end{equation}
for all $\ga,\xi\in\Ga$, $i,j\in\set{1,\ldots,n}$ (see \cite[Thm. 10.4]{Aizenman Warzel 2016}). For an application of this localization result in the context of Bogoliubov-de Gennes Hamiltonians, we refer to \cite[Thm. 1]{De Nittis Drabkin Schulz-Baldes 2015}.

Condition $\la<\widehat{\la}_\rho(E)$ is known as the \emph{weak disorder regime} in the literature, in which it is possible to show DL for sufficiently small values of the disorder parameter, see \cite[Theorem 1.1]{Aizenman 1994}. Notice, however, that the result might be void when applied to internal spectral gaps, if the support of the distribution is bounded, since the latter inequality might be satisfied only by values $E$ in the spectral gap of $H_{\la,\omega}$. Therefore, in order to have a meaningful result, one has to check that there exist energies $E\in\sigma(H_{\la,\omega})$ such that $\la<\widehat{\la}_\rho(E)$. To do that we need to require the disorder to be not too small and the energy to be far enough from the unperturbed spectrum, which leads to condition \eqref{weak far} (see Fig. \ref{fig 6}). 

\begin{figure}[ht]
    \centering
    \includegraphics[width=0.55\textwidth]{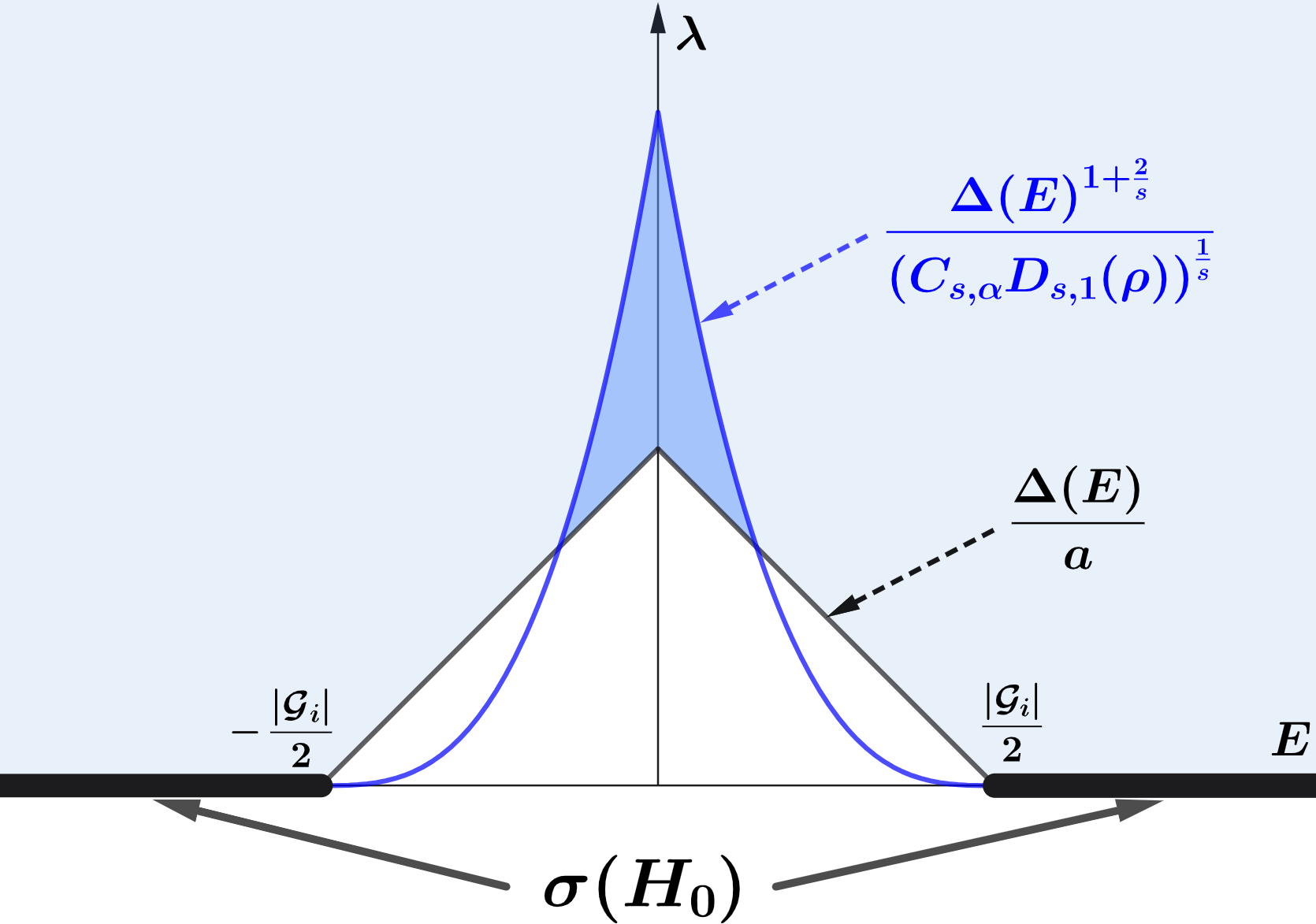}
    \caption{\small{(Refer to Theorem \ref{DL regimes}-(2)) Plot in light blue of the region of localization away from the unperturbed spectrum, described by condition \eqref{weak far}, for Hamiltonians $H_0$ with a spectral gap $\G_i$ centered in $0$ and distribution supported in $[-a,a]$.}}
    \label{fig 6}
\end{figure}

In order to clarify the expression of the constant $C_{s,\alpha}$ in \eqref{weak far}, 
we give a proof of Proposition \ref{DL regimes}-\eqref{point 2 DL} for the convenience of the reader.

\begin{proof}[Proof of Proposition \ref{DL regimes}-\eqref{point 2 DL}]
    By Proposition \ref{combes thomas}, since $2S_\alpha\leq \frac{|\G_i|}{2}$, we have that
        \begin{align*}
            \sup_{\eta\neq 0}|\langle \delta_{\ga,i},(H_0-(E+\iu\eta)^{-1}\delta_{\xi,j}\rangle|&\leq \frac{2}{\Delta(E)}\exp\bigg(-\alpha\min\set{\frac{\Delta(E)}{2 S_\alpha},1}|\ga-\xi|\bigg) \\
            &\leq  \frac{2}{\Delta(E)}\exp\bigg(-\alpha\min\set{\frac{2\Delta(E)}{|\G_i|},1}|\ga-\xi|\bigg).
        \end{align*}
    Since $E\in\G_i$ we have that $\Delta(E)\leq \frac{|\G_i|}{2}$ and thus 
    \begin{equation*}
         \sup_{\eta\neq 0}|\langle \delta_{\ga,i},(H_0-(E+\iu\eta)^{-1}\delta_{\xi,j}\rangle|\leq \frac{2}{\Delta(E)}\exp\bigg(-\frac{2\alpha \Delta(E)}{|\G_i|}|\ga-\xi|\bigg).
    \end{equation*}
    Given $s\in (0,\tau)$ consider $\mu=\frac{s\alpha}{|\G_i|}$ then, using $|\ga|\geq |\ga|_\infty$, we have
    \begin{align*}
        &\sum_{(\xi,j)\in \Ga\times\set{1,\ldots,n}} |\langle \delta _{\ga,i}, (H_0-(E+\iu\eta))^{-1}\delta_{\xi,j}\rangle|^s\eu^{\mu \Delta(E)|\ga-\xi|}\\
        &\leq \frac{2^sn}{\Delta(E)^s} \sum_{\ga\in\Ga}\exp\bigg(-\frac{s\alpha \Delta(E)}{|\G_i|}|\ga|_\infty\bigg)= \frac{2^sn}{\Delta(E)^s}\bigg(1+ 8\sum_{k=1}^\infty k \exp\bigg(-\frac{s\alpha \Delta(E)}{|\G_i|}k\bigg)\bigg)
    \end{align*}
    where in the last equality we used $\#\set{n\in \Z^2: \max\set{|n_1|,|n_2|}=k}=8k$. Notice that, if $y>0$, we have
    \begin{equation*}
        \sum_{k=1}^\infty k(\eu^{-y})^k=\eu^{-y}\frac{\di}{\di x}\bigg(\sum_{k=0}^\infty \eu^{-ky}\bigg)=\frac{\eu^{-y}}{(1-\eu^{-y})^2}=\frac{1}{y^2}\frac{y^2 \eu^{-y}}{(1-\eu^{-y})^2}\leq \frac{1}{y^2}.
    \end{equation*}
    Hence we get
    \begin{align*}
         &\sum_{(\xi,j)\in \Ga\times\set{1,\ldots,n}} |\langle \delta _{\ga,i}, (H_0-(E+\iu\eta))^{-1}\delta_{\xi,j}\rangle|^s\eu^{\mu \Delta(E)|\ga-\xi|} \\
         &\leq \frac{2^sn}{\Delta(E)^s}\bigg(1+8\bigg(\frac{|\G_i|}{s\alpha \Delta(E)}\bigg)^2\bigg)
         \leq \frac{n|\G_i|^2}{2\Delta(E)^{s+2}}\bigg(1+\frac{32}{s^2\alpha^2}\bigg)=\frac{C_{s,\alpha}}{\Delta(E)^{s+2}}
    \end{align*}
    where in the second inequality we used $\Delta(E)\leq \frac{|\G_i|}{2}.$ Therefore, if the upper bound in \eqref{weak far} is satisfied, we obtain
    \begin{align*}
        \la^s D_{s,1}(\rho)\sup_{\eta\neq 0}\sup_{(\ga,i)\in\Ga\times \set{1,\ldots,n}}\sum_{(\xi,j)\in \Ga\times \set{1,\ldots,n}}|\langle \delta_{\ga,i},(H_0-(E+\iu\eta))^{-1}\delta_{\xi,j}\rangle|^s \eu^{\mu \Delta(E)|\ga-\xi|}<1.
    \end{align*}
    As a consequence we have that $\la<\widehat{\la}_\rho(E)$, defined in \eqref{sigma moment regular}, and so the Fractional Moment bound \eqref{FMM weak2} holds.
\end{proof}
If we assume that $\supp(\rho)=[-a,a]$, then a sufficient condition for the regime described by \eqref{weak far} to be non empty, is to require $D_{s,1}(\rho)$ to be bounded uniformly in $a$. Indeed in this way we would have, for values of $a$ large enough, that the light blue region in Fig. \ref{fig 6} would enter inside the gapless spectrum. However notice that the uniformity condition of $D_{s,1}(\rho)$ on $a$ is not always guaranteed for every distribution. In fact the latter is not valid in the case of the uniform distribution $\rho(\di v)=\frac{1}{2a}\chi_{[-a,a]}(v)\di v$, since
\begin{equation*}
    a^s\geq D_{s,1}(\rho)\geq 2a\bigg(\int_{-a}^a   \frac{\di v}{|v|^s}\bigg)^{-1}=a^s(1-s).
\end{equation*}
A sufficient condition for the uniformity of $D_{s,1}(\rho)$ on $a$ was given in \cite[Lemma A1]{Aizenman 1994}, in which the author shows, under the assumption of bounded density $\rho(\di v)=\rho(v)\di v$, that $D_{s,1}(\rho)$ can be bounded in terms of the $L^{1+q}$-norm of the density and of the corresponding $t$-moment, for $t\in(0,1]$ and $q>0$.
More precisely, assume that there exist positive constants $B,C$ such that
\begin{align*}
    &\int|v|^t \rho(v)\di v\leq B<\infty &
     \int \rho(v)^{1+q}\ \di v\leq C<\infty
\end{align*}
for some $t\in (0,1]$ and $q>0$. Then, for all $s<\Big[1+\frac{2}{t}+\frac{1}{q}\Big]^{-1}$ there exists $K=K(B,C,s,t,q)$ such that $D_{s,1}(\rho)\leq K$. Moreover, keeping track of the constants in the proof of \cite[Lemma A1]{Aizenman 1994}, we obtain an explicit expression for $K$ given by
\begin{equation}\label{explicit constant}
    K=\max\set{5(2B)^{\frac{s}{t}},2^{2s+1}B^{\frac{s}{t}}(1+B^{\frac{s}{t}}C_{p,q})} 
\end{equation}
where $\displaystyle C_{p,q}=1+\frac{p(2^qC)^{\frac{1}{1+q}}}{\frac{q}{1+q}-p}$ and $ p\displaystyle=\frac{s}{1-\frac{2s}{t}}$.

In Section \ref{section far haldane} we will use the previous estimate in the case of the Haldane-Anderson model in order to prove DL inside the bulk of the gapless spectrum. As a consequence we show the existence of a metal-insulator transition in disorder and we give an estimate on the location of the delocalization region in energy-disorder plane. 

\subsection{Band edge localization}\label{band edge}

Localization at the band edges of the spectrum is proved using the techniques of MSA. In order to introduce them, we define the restriction of $H_{\la,\omega}$ to the box $\La_L\subset \Ga$ with simple boundary conditions, and the corresponding resolvent operator, by
\begin{equation}\label{defn:simple-bc}
    H_{\la,\omega,\La_L}=\chi_{\La_L}H_{\la,\omega}\chi_{\La_L}, \quad G_{\la,\omega,\La_L}(z)=(H_{\la,\omega,\La_L}-z)^{-1}
\end{equation}
for $z\notin\sigma(H_{\la,\omega,\La_L})$, where $\La_L$ is defined in \eqref{box}.
    
Next, we recall the notion of “good boxes" for random operators satisfying \ref{item: P2}. This ensures a good decay of the finite-volume resolvent from the center of the box to its (inner) boundary $\partial\La_L^{(r)}$ given by
\begin{equation}\label{boundary-box}
    \partial\La_L^{(r)}=\{\xi\in \La_L:\exists \,\eta\in\Ga\setminus\La_L, |\xi-\eta|_\infty\leq r\}\subset \Gamma
\end{equation}
where $r\in \N^*$ is the range of the hopping terms in \ref{item: P2}.

\begin{dfn}\label{definition good box}
    Given $\theta>0, E\in\R$ and $L\in 3\N^*+4r$, we say that the box $\La_L$ is $(\omega,\theta,E)$-suitable with range $r$, if $E\notin\sigma(H_{\la,\omega,\La_L})$ and
    \begin{equation*}
        \norm{G_{\la,\omega,\La_L}(\ga,\xi;E)}\leq \frac{1}{L^{\theta}},
    \end{equation*}
    for all $\ga\in\La_{\frac{L+2r}{3}}, \xi\in\partial\La_L^{(r)}$, where $\norm{\cdot}$ is the operator norm on $M_n(\C)$.
\end{dfn}

\begin{rmk}
    The notion of suitable box given in definition \ref{definition good box} is slightly different from the one introduced in \cite[Def. 3.1]{Germinet Klein 2001} because we are dealing with random operators of finite hopping $r>1$. The choice of the core $\La_{\frac{L+2r}{3}}$ of the box $\La_L$ comes from the following observation: for each $\ga$ in the outer boundary $\La_{L+2r}\setminus\La_L$ there exists a unique $\xi\in \Ga$ with $|\xi|_\infty=\frac{L+2r}{3}$ such that $\ga\in \La_{\frac{L+2r}{3}}+\xi$. Thus we are able, using the geometric resolvent identity, to move from the core of a box to its outer boundary, which allows to initialize the bootstrap Multiscale Analysis discussed in \cite{Germinet Klein 2001}.
\end{rmk}

The band edge localization result stated in Theorem \ref{DL regimes}-\eqref{point band edge} is a consequence of the following result from \cite{Figotin Klein 1994}, the so-called initial length scale estimate for MSA at external band edges, and at internal band edges for small disorder. We give a proof adapted to our setting in Appendix \ref{Appendix band edge} for the convenience of the reader. 

\begin{thm}[\protect{\cite[Lemma 2.19]{Figotin Klein 1994}}]\label{inital length estimate}
    Assume that $H_{\la,\omega}$ satisfies hypothesis \ref{item: P1}-\ref {item: P4}, \ref{item: R1} and \ref{item: R2}. Then for all $\theta>\frac{2}{\tau}$ it holds that
    \begin{align}
        &\limsup_{L\to\infty}\mathbb{P}(\La_L \ \textup{is} \ (\omega,\theta,\alpha_1-a\la)\textup{-suitable wih range} \ r)=1 \label{input 0 or} \\
        &\limsup_{L\to\infty}\mathbb{P}(\La_L \ \textup{is} \ (\omega,\theta,\beta_N+b\la)\textup{-suitable with range} \ r)=1. \label{input N or}
    \end{align}
    Moreover, if \ref{item: P1} holds with $B\in 2\pi \frac{p}{q}$, with $p\in \N,q\in \N^*$, then for all $\la<\frac{|\G_i|}{a+b}$ with $i\in\set{1,\ldots, N-1}$, for all $\theta>\frac{2}{\tau}$ it holds that
    \begin{align}
        &\limsup_{L\to\infty}\mathbb{P}(\La_{qL} \ \textup{is} \ (\omega,\theta,\beta_i+b\la)\textup{-suitable with range} \ qr)=1 \label{input 1 or} \\
        &\limsup_{L\to\infty}\mathbb{P}(\La_{qL} \ \textup{is} \ (\omega,\theta,\alpha_{i+1}-a\la)\textup{-suitable with range} \ qr)=1. \label{input 2 or}
    \end{align}
\end{thm}
The probabilistic condition appearing in \eqref{input 0 or}-\eqref{input 2 or} is the input of the MSA. It was shown in \cite[Thm. 4.2]{Germinet Klein 2004} that, under assumptions \ref{item: P1},\ref{item: P2},\ref{item: P4} and \ref{item: R1}, the verification of the initial length scale estimate at an energy level $E$ is equivalent to $H_{\la,\omega}$ exhibiting DL in $E$. Therefore, under the additional assumptions \ref{item: P3} and \ref{item: R2}, with \ref{item: P1} holding with $B\in 2\pi \Q$, the previous result yields dynamical localization at external and internal band edges energies of the spectrum $\sigma(H_{\la,\omega})$.

\goodbreak
    
\section{Proofs of Theorems \ref{DD energy}, \ref{DD disorder} \& \ref{curve}}\label{proof section}

In this section we prove Theorems \ref{DD energy}, \ref{DD disorder} and \ref{curve}, which show dynamical delocalization, respectively, in the energy parameter at fixed disorder, in the disorder parameter at fixed energy, and on continuous curves in the energy-disorder plane, whenever there is a jump of the Chern character.

\begin{proof}[Proof of Theorem \ref{DD energy}]
    In order to prove Part \eqref{point 1 dd energy} we argue by contradiction. Given $\la<\frac{\min\set{|\G_{i-1}(0)|,|\G_i(0)|}}{a+b}$ assume that $H_{\la,\omega}$ exhibits DL for all $E\in \mathcal{B}_i(\la)$. In particular $H_{\la,\omega}$ exhibits DL in $\G_{i-1}(\la)\sqcup\mathcal{B}_i(\la)\sqcup \G_{i}(\la)$ and thus there exists $C<\infty,\mu>0,\zeta\in (0,1]$ such that
    \begin{equation*}
        \E\bigg[\sup_{E\in [E_-,E_+]}\Tr_{\C^n}|P_{E,\la,\omega}(0,\ga)|\bigg]\leq C\exp{(-\mu|\ga|^\zeta)}
    \end{equation*}
    for some $E_-\in \G_{i-1}(\la)$ and $E_+\in \G_i(\la)$. Therefore, by Theorem \ref{continuity chern energy}, we have
    \begin{equation*}
        \E[C(P_{E_-,\la,\omega})]=\E[C(P_{E_+,\la,\omega})].
    \end{equation*}
    Hence, by Corollary \ref{robusteness}, we get
    \begin{equation*}
        C(P_{E_-,\la=0})=\E[C(P_{E_-,\la,\omega})]=\E[C(P_{E_+,\la,\omega})]=C(P_{E_+,\la=0})
    \end{equation*}
    which is in contradiction with the hypothesis.
    
    Part \eqref{point 2 dd energy} is a consequence of Proposition \ref{rmk characterization}. Indeed for all $g\in C_{c,+}^\infty(\R)$ with $g\equiv 1$ on an open interval containing $E$, for all $\alpha\geq 0$ and $p>\frac{4\alpha}{\tau}+\frac{24}{\tau}+12$ it holds
    \begin{equation*}
        \liminf_{T\to\infty}\frac{1}{T^{\alpha}}\mathcal{M}_\la(p,g,T)=\infty.
    \end{equation*}
    Hence $\alpha <\frac{p\tau}{4}-(3\tau+6)$ and so for all $p>0$ there exists $C_{p,g,\la}>0$ such that
    \begin{equation*}
         \mathcal{M}_\la(p,g,T)\geq C_{p,g,\la}T^{\frac{p\tau}{4}-\kappa}
    \end{equation*}
    for all $\kappa>3\tau +6$ and $T\geq 0$. Lastly, Part \eqref{point 3 dd energy} follows from Theorem \ref{DL regimes}-\eqref{point band edge} as in Part \eqref{point 1 dd energy}.
\end{proof}

Next, we will prove the general delocalization result in Theorem \ref{curve}, from which Theorem \ref{DD disorder} will follow as a corollary.

\begin{proof}[Proof of Theorem \ref{curve}] 
    Throughout the proof we will omit the dependence on $\omega$. Assume by contradiction that for all $t\in [0,T]$ there exists $C_t<\infty, \mu_t>0$ and $\zeta_t\in (0,1]$ such that:
    \begin{equation*}
        \E\bigg[\Tr_{\C^n}|P_{E(t),\la(t)}(0,\ga)|\bigg]\leq C_t\exp{(-\mu_t|\ga|^{\zeta_t})}.
    \end{equation*}
    Let $t_1,t_2\in [0,T]$, then by Remark \ref{off diagonal} and Proposition \ref{ids continuity disorder}-(2) we have that
    \begin{align*}
        &\E\bigg[\Tr_{\C^n}|(P_{E(t_1),\la(t_1)}-P_{E(t_2),\la(t_2)})(0,\ga)|\bigg] \\
        &\leq \E\bigg[\Tr_{\C^n}|(P_{E(t_1),\la(t_1)}-P_{E(t_2),\la(t_1)})(0,\ga)|+\Tr_{\C^n}|(P_{E(t_2),\la(t_1)}-P_{E(t_2),\la(t_2)})(0,\ga)|\bigg] \\
        &\leq \frac{2^{2-\tau}n^2\pi C_\tau(\rho)}{\min_{t\in[0,T]}|\la(t)|}|E(t_1)-E(t_2)|^\tau+C_K|\la(t_1)-\la(t_2)|^{\frac{\tau}{\tau+2}}
    \end{align*}
    where $\min_{t\in[0,T]}|\la(t)|>0$ by assumption, and $C_K$ is the constant appearing in Proposition \ref{ids continuity disorder}-(2) with $K=E([0,T])\times \la([0,T])$. In particular the map
    \begin{equation*}
        [0,T]\ni t\mapsto\E\bigg[\Tr_{\C^n}|P_{E(t),\la(t)}(0,\ga)|\bigg]
    \end{equation*}
    is continuous, and thus there exists $t'\in [0,T]$ such that
    \begin{equation*}
        \sup_{t\in [0,T]}\E\bigg[\Tr_{\C^n}|P_{E(t),\la(t)}(0,\ga)|\bigg]=\E\bigg[\Tr_{\C^n}|P_{E(t'),\la(t')}(0,\ga)|\bigg]\leq C_{t'}\exp{(-\mu_{t'}|\ga|^{\zeta_{t'}})}.
    \end{equation*}
    Hence it holds that
    \begin{equation*}
        \sup_{t\in [0,T]}\sum_{\ga\in\Ga}|\ga|\E\bigg[\Tr_{\C^n}|P_{E(t),\la(t)}(0,\ga)|^3\bigg]^{\frac{1}{3}}<\infty.
    \end{equation*}
    Since condition \eqref{ids disorder} is satisfied with $\alpha_1=\tau$ and $\alpha_2=\frac{\tau}{\tau+2}$, we can apply Theorem \ref{continuity chern disorder} with $\mathcal{R}=\set{(E(t),\la(t)),\, t\in[0,T]}$, obtaining
    \begin{equation*}
        \E[C(P_{E(0),\la(0)})]=\E[C(P_{E(T),\la(T)})]
    \end{equation*}
    which is a contradiction. Indeed by Corollary \ref{robusteness} we have that $\E[C(P_{E(0),\la(0)})]=C(P_{E(0),\la=0})$ and $C(P_{E(0),\la=0})\neq 0$ by assumption. On the other hand, if $\la(T)>\la_\rho(H_0)$ then $\E[C(P_{E(T),\la(T)})]=0$ by Corollary \ref{chern zero strong disorder}, while if $E(T)\in \G_j(\la(T))$, with $i\neq j$, then  $\E[C(P_{E(T),\la(T)})]=C(P_{E(T),\la=0})$ by Corollary \ref{robusteness}, and $C(P_{E(T),\la=0})\neq C(P_{E(0),\la=0})$ by hypothesis. In both cases we have arrive at a contradiction, thus concluding the proof.
\end{proof}

\begin{proof}[Proof of Theorem \ref{DD disorder}] 
    Given $E\in \G_i(0)$ let $\la_-\in (0,\la_0(E))$ and $\la_+\in (\la_\rho(H_0),+\infty)$. Consider the continuous curve
    \begin{equation*}
        [\la_-,\la_+]\ni \la\mapsto C_E(\la)=(E,\la)\in \R\times (0,\infty).
    \end{equation*}
    By assumption we have that $C(P_{E,\la=0})\neq 0$, and since $\la_+>\la_\rho(H_0)$, by Theorem \ref{curve}, there exists $\la^*\in [\la_-,\la_+]$ such that $H_{\la^*,\omega}$ exhibits DD in $E$. Since $H_{\la,\omega}$ exhibits DL in $E$ for all $\la\in [\la_-,\la_0(E))$ and $(\la_\rho(H_0),\la_+]$ it holds that $\la^*\in [\la_0(E),\la_\rho(H_0)]$.
    
    Similarly to the proof of Theorem \ref{DD energy}-(2), \eqref{transport exponent} is a consequence of Proposition \ref{rmk characterization}. The proof of Part \eqref{point 2 dd disorder} follows from Theorem \ref{DL regimes}-\eqref{point band edge} as in Part \eqref{point 1 dd disorder}.
\end{proof} 

\goodbreak

\section{The Anderson metal-insulator transition for the Haldane-Anderson model}\label{section far haldane}

In this section we consider the particular case of the Haldane-Anderson model, as in Definition \ref{defn:haldane-anderson}, with a single-site probability distribution $\rho_a$, for which we prove existence of localized and delocalized energies in the bulk of the perturbed spectrum. For a specific choice of the parameters, we provide estimates on the location of the delocalized energies, showing that the interval for values of $\la$ in Theorem \ref{DD disorder}-(1) and in Corollary \ref{corollary haldane} can be made more precise, and it is at a positive distance from $\la_0(E)$. 
The key idea is to show that there is a region of localization inside the bulk of the gapless spectrum by using Theorem \ref{DL regimes}-\eqref{point 2 DL}. Then the existence of delocalized states follows from Theorem \ref{DD disorder}.

Consider parameters $-3\sqrt 3 t_2< M< 3\sqrt{3}t_2$, $\phi=\pm \frac{\pi}{2}$ and assume \ref{item: P4} holds with an absolutely continuous distribution $\rho_a$ having support in $[-a,a]$ \mbox{with $a\geq 1$. Then}
\begin{gather*}
    \sigma(\Hi_\text{Hal})=\left[-\norm{\Hi_{\textup{Hal}}},-\frac{|\G_1|}{2}\right]\sqcup  \left[\frac{|\G_1|}{2},\norm{\Hi_{\textup{Hal}}}\right] \\
    \sigma(\Hi_{\textup{Hal},\la,\omega})=\left[-\norm{\Hi_{\textup{Hal}}}-a\la,-\frac{|\G_1|}{2}+a\la\right]\cup \left[\frac{|\G_1|}{2}-a\la,\norm{\Hi_{\textup{Hal}}}+a\la\right] \quad \mathbb{P}\textup{-a.s.}
\end{gather*}
where $|\G_1|>0$ is the size of the internal gap. In particular $\sigma(\Hi_{\textup{Hal},\la,\omega})$ exhibits, almost surely, a gap $\G_1(\la)=\left(-\frac{|\G_1|}{2}+a\la,\frac{|\G_1|}{2}-a\la\right)$ which remains open for $\la< \frac{|\G_1|}{2a}$, and we have that $E\in \left(-\frac{|\G_1|}{2},\frac{|\G_1|}{2}\right)$ is in $\G_1(\la)$ if and only if $\la<\frac{\Delta(E)}{a}$, where $\Delta(E)=\dist(E,\sigma(\Hi_{\textup{Hal}}))$. For simplicity, we will focus on the case $E=0$.

Note that with our choice of  $\rho_a$, \ref{item: R1} holds with $\tau=1$. In order to apply Theorem \ref{DL regimes}-\eqref{point 2 DL} to show DL, we need to prove that for $E=0$, there exist values of $\la$ which satisfy condition \eqref{weak far}, in which the left-hand side, corresponding to $\la_0(0)$, takes the value $\frac{\Delta(0)}{a}=\frac{|\G_1|}{2a}$. Note that the values $ |\mathcal G_1|$ and $C_{s,\alpha}$ appearing on the RHS of \eqref{weak far} depend on the unperturbed operator $\Hi_{\textup{Hal}}$ and on some parameters $\alpha,s$, but do not depend on the radius $a$ of the support of the distribution. As discussed in Section \ref{section localization away}, it suffices then to show that the constant $D_{s,1}(\rho_a)$ appearing in \eqref{sigma-moment regularity} can be bounded uniformly in $a$. In this way, we can choose  $a$  large enough such that
\begin{equation}\label{eq:lhs}
    \frac{\Delta(0)^{1+\frac{2}{s}}}{(C_{s,\alpha}D_{s,1}(\rho_a))^{\frac{1}{s}}}=\frac{|\G_1|^{1+\frac{2}{s}}}{2(4C_{s,\alpha}D_{s,1}(\rho_a))^{\frac{1}{s}}}> \frac{|\G_1|}{2a},
\end{equation}
that is,
\begin{equation}\label{eq:lhs-a}
  \left( \frac{|\G_1|^2}{4 C_{s,\alpha}D_{s,1}(\rho_a)}\right)^{\frac{1}{s}}> \frac{1}{a}.
\end{equation}
The latter guarantees that the highest point of the localization curve appearing in Fig. \ref{fig 6} is well inside the spectrum.

In order to have a uniform bound of $D_{s,1}(\rho_a)$ on $a$, we consider the distribution $\rho_a(\di v)=\rho_a(v)\di v$, where $\rho_a(v)$ is a truncated Gaussian function having support in $[-a,a]$, that is
\begin{equation}\label{truncated gaussian}
    \rho_a(v)=\frac{f(v)\chi_{[-a,a]}(v)}{\norm{f}_{L^1[-a,a]}} \quad \textup{where} \quad f(v)=\frac{1}{\sqrt{2\pi}}\eu^{-\frac{v^2}{2}}.
\end{equation}
Recalling that $a\geq 1$, we find a constant $B$ independent of $a$ and, for any $q>0$, we find $C=C_q$ independent of $a$, such that 
\begin{align*}
    &\int|v|\rho_a(v)\di v=\frac{1}{\norm{f}}_{L^1[-a,a]}\frac{1}{\sqrt{2\pi}}\int _{-a}^a|v|\, \eu^{\frac{-v^2}{2}}\di v\leq \sqrt{\frac{2}{\pi}}\frac{1}{\norm{f}_{L^1[-1,1]}}\leq \sqrt \eu= B \\
    & \int \rho_a(v)^{1+q}\di v=\frac{1}{\norm{f}_{L^1[-a,a]}^{q+1}}\frac{1}{(2\pi)^{\frac{1+q}{2}}}\int_{-a}^a \eu^{\frac{-(q+1)v^2}{2}}\di v\\
    &\hspace{2.4cm}\leq \sqrt{\frac{2}{q+1}}\frac{\sqrt{\pi}}{\norm{f}^{q+1}_{L^1[-1,1]}(2\pi)^{\frac{q+1}{2}}}
    \leq \sqrt{\frac{2}{q+1}}\frac{\sqrt{\pi}\eu^{(q+1)/2}}{2^{q+1}}=C.
\end{align*}
Taking $t=1$ in \eqref{explicit constant}, with $s<\bigg[1+\frac{2}{t}+\frac{1}{q}\bigg]^{-1}=\frac{1}{3+\frac{1}{q}}$, we have that 
\begin{equation*}
\sup_{a\geq 1}D_{s,1}(\rho_a)\leq K=K(B,C,s,q)
\end{equation*}
where the constant $K$ is given in  \eqref{explicit constant}, depends only on $B,C,s,q$ and not on $a$, and $K>1$.
Therefore the LHS of \eqref{eq:lhs-a} can be estimated by
\begin{equation}\label{eq:lhs2}
   \left( \frac{|\G_1|^2}{4 C_{s,\alpha}D_{s,1}(\rho_a)}\right)^{\frac{1}{s}}\geq\left( \frac{|\G_1|^2}{4 C_{s,\alpha}K}\right)^{\frac{1}{s}}.
\end{equation}
Next, define
\begin{equation}\label{eq:a_0}
 a_0=  a_0(|\G_1|, s,\alpha, K)=\bigg(\frac{4K C_{s,\alpha}}{|\G_1|^2}\bigg)^{\frac{1}{s}}.
\end{equation}
Recalling the expression for $C_{s,\alpha}$ given in  Theorem \ref{DL regimes}-\eqref{point 2 DL}, we have that $C_{s,\alpha}>|\G_1|^2$, and so $a_0>1$. Thus, for any $a>a_0$, \eqref{eq:lhs-a} holds, so the interval $\left[\la_0(0), \frac{\Delta(0)^{1+\frac{2}{s}}}{(C_{s,\alpha}D_{s,1}(\rho_a))^{\frac{1}{s}}}\right] $ is not empty. By Theorem \ref{DD disorder} and Theorem \ref{DL regimes} Parts \eqref{strong disorder regime thm} and \eqref{point 2 DL}, for all $a>a_0$ the following holds
\begin{enumerate}
    \item for all $\la\in \left[\frac{|\G_1|}{2a},\frac{|\G_1|}{2a_0}\right)\cup(\la_{\rho_a}(\Hi_{\textup{Hal}}),\infty)$ we have that $0\in \sigma(\Hi_{\textup{Hal},\la,\omega})$ almost surely, and $\Hi_{\textup{Hal},\la,\omega}$ exhibits DL in $E=0$, where $\la_{\rho_a}(\Hi_{\textup{Hal}}) $ denotes the strong disorder  threshold;
    \item there exists $\la^*\in \left[\frac{|\G_1|}{2a_0},\la_{\rho_a}(\Hi_{\textup{Hal}})\right]$ such that $\Hi_{\textup{Hal},\la^*,\omega}$ exhibits DD in $E=0$.
\end{enumerate}
    
This result shows that in Theorem \ref{DD disorder} the interval for $\la^*$, the value of the disorder parameter at which DD holds at $E$, can be at a positive distance from $\la_0(E)$. In this example, for all $a>a_0$,  the disorder parameter for which DD holds belongs to an interval with a fixed lower edge $\frac{|\G_1|}{2a_0}$. Choosing, for example, $a=n a_0$, with $n$ an integer, the minimal distance from $\la^*$ to the value $\la_0(E)$ where the gap closes is $(n-1)\frac{|\mathcal G_1|}{2a}$.

Next, we give an example by making a specific choice of the parameters in the model described above, in order to give numerical estimates on the region of localization and delocalization.

{\bf Example}: Take $M=0$, $\phi=\pm\frac{\pi}{2}$ and $t_2=\frac{t_1}{3\sqrt{3}}$. The explicit expression for the Haldane Bloch bands is given in \cite[Section 3]{Marcelli Monaco Moscolari Panati 2018}, where for this choice of parameters one can see that the size of the internal gap is given by $|\G_1|=2t_1$. In order to compute $a_0$ in \eqref{eq:a_0}, we can take: \\
$t=1,\ q=2$ in \eqref{explicit constant}, with $s=\frac{1}{4}<\bigg[1+\frac{2}{t}+\frac{1}{q}\bigg]^{-1}=\frac{2}{7}$, so that $C_q=\frac{\sqrt{2\pi \eu^3}}{8\sqrt{3}}$, $p=\frac{s}{1-2s}=\frac{1}{2}$ and $C_{p,q}=1+\frac{p(2^qC)^{\frac{1}{1+q}}}{\frac{q}{1+q}-p}=1+3\sqrt[3]{4C_q}$, and therefore
\begin{align*}
     K=
    \max\set{5(2B)^s,2^{2s+1}B^{s}(1+B^{s}C_{p,q})}= \sqrt{8}\sqrt[8]{\eu}\Bigg(1+\sqrt[8]{\eu}+3\sqrt[8]{\eu}\sqrt[3]{\frac{\sqrt{2\pi \eu^3}}{2\sqrt{3}}}\Bigg)\cong 22.96.
\end{align*}
The constant $C_{s,\alpha}$ above can be explicitly computed from  $S_\alpha$ with $\alpha>0$ in Proposition \ref{combes thomas}. Consider the dimerization $\Ci\cong\Ga\times\set{0,d_3}$, where $\Ga=\Span_\Z\set{a_1,a_2}$.
\begin{align*}
    S_\alpha&=\sup_{\ga\in\Ga}\sum_{\xi\in\Ga}\norm{\Hi_{\textup{Hal}}(\ga,\xi)}(\eu^{\alpha|\ga-\xi|}-1)\\ &=\Big(\norm{\Hi_{\textup{Hal}}(0,a_3)}+\norm{\Hi_{\textup{Hal}}(0,-a_3)}\Big)(\eu^{\sqrt{2}\alpha}-1) \\
    &+\Big( \norm{\Hi_{\textup{Hal}}(0,a_1)}+\norm{\Hi_{\textup{Hal}}(0,-a_1)}+\norm{\Hi_{\textup{Hal}}(0,a_2)}+\norm{\Hi_{\textup{Hal}}(0,-a_2)}\Big)(\eu^\alpha-1)
\end{align*}
where the hopping matrices are given by
\begin{gather}
    \Hi_{\textup{Hal}}(0,a_1)=
    \begin{bmatrix}
    t_2 \eu^{-\iu\phi} & t_1 \\
    0  & t_2 \eu^{\iu\phi}
    \end{bmatrix}
    = \Hi_{\textup{Hal}}(0,-a_1)^\dagger \label{matrix 1}\\
    \Hi_{\textup{Hal}}(0,a_2)=
    \begin{bmatrix}
    t_2 \eu^{-\iu\phi} & 0 \\
    t_1  & t_2 \eu^{\iu\phi}
    \end{bmatrix}
    =\Hi_{\textup{Hal}}(0,-a_2)^\dagger \label{matrix 2}\\
    \Hi_{\textup{Hal}}(0,a_3)=
    \begin{bmatrix}
    t_2 \eu^{-\iu\phi} & 0 \\
    0 & t_2 \eu^{\iu\phi}
    \end{bmatrix}
    =\Hi_{\textup{Hal}}(0,-a_3)^\dagger \label{matrix 3}
\end{gather}
where $A^\dagger$ is the Hermitian transpose of the matrix $A\in M_n(\C)$.

Using the fact that $\norm{A}\leq n\norm{A}_\infty$, which holds for all $A\in M_n(\C)$, and $\norm{A}=\norm{A}_\infty$ if $A$ is a diagonal matrix, we get
\begin{equation*}
    S_\alpha\leq 10\max\set{t_1,t_2}(\eu^{\sqrt{2}\alpha}-1)=10\,t_1(\eu^{\sqrt{2}\alpha}-1).
\end{equation*}
We choose 
\begin{equation*}
    \alpha=\frac{1}{\sqrt{2}}\ln\bigg(1+\frac{|\G_1|}{40\, t_1} \bigg)=\frac{1}{\sqrt{2}}\ln\bigg(1+\frac{1}{20} \bigg) ,
\end{equation*}
so that $20\,t_1(\eu^{\sqrt{2}\alpha}-1)= \frac{|\G_1|}{2}$, and $2S_\alpha\leq \frac{|\G_1|}{2}$.
Since $s=\frac{1}{4}$ we have that
\begin{equation*}
    C_{s,\alpha}=|\G_1|^2\bigg(1+\frac{32}{s^2\alpha^2}\bigg)=|\G_1|^2\Bigg[1+\Bigg(\frac{32}{\ln\big(1+\frac{1}{20}\big)}\bigg)^2\Bigg].
\end{equation*}
Therefore 
\begin{equation*}
    a_0=\bigg(4K\bigg[1+\bigg(\frac{32}{\ln(1+\frac{1}{20})}\bigg)^2\bigg]\bigg)^{4}\cong 2.4\cdot 10^{30}
\end{equation*}   
and so
\begin{equation*}
    \frac{|\G_1|}{2 a_0}=t_1\bigg(4K\bigg[1+\bigg(\frac{32}{\ln(1+\frac{1}{20})}\bigg)^2\bigg]\bigg)^{-4}\cong \ 4.1\cdot 10^{-31}\, t_1. 
\end{equation*}
Then, taking $a>a_0$, for all $\la \in \left[\frac{|\G_1|}{2 a},\frac{|\G_1|}{2 a_0} \right)$ the operator $\Hi_{\textup{Hal},\la,\omega}$ contains $E=0$ in its almost-sure spectrum and exhibits DL there.

Next, we compute the strong disorder threshold $\la_{\rho_a}(\Hi_{\textup{Hal}})$. Recall that
\begin{gather*}
    \la_{\rho_a}(\Hi_{\textup{Hal}})=\inf_{s\in (0,1)}\inf_{\mu>0}\bigg[C_{s,1}(\rho_a)\sup_{\ga\in\Ga}\sup_{i=1,2}\sum_{\substack{\xi\in \Ga,j=1,2 \\ (\xi,j)\neq (\ga,i)}} |\langle\delta_{\ga,i},\Hi_{\textup{Hal}}\,\delta_{\xi,j}\rangle|^s \eu^{\mu|\ga-\xi|}\bigg]^{\frac{1}{s}} \\
    \textup{where} \quad C_{s,1}(\rho_a)=\frac{C_1(\rho_a)^s}{1-s}=\frac{2^s\norm{\rho_a}^s_\infty}{1-s}=\frac{2^s}{(1-s)(2\pi)^{\frac{s}{2}}\norm{f}_{L^1[-a,a]}^s}.
\end{gather*}
Recalling the expression of the matrices \eqref{matrix 1}, \eqref{matrix 2}, \eqref{matrix 3}, noticing that $\Hi_{\textup{Hal}}(0,0)=
\begin{bmatrix}
        0 & t_1 \\
        t_1 & 0
\end{bmatrix}$ and using the periodicity property \ref{item: P1}, we get
\begin{align*}
    &\la_{\rho_a}(\Hi_{\textup{Hal}})=\inf_{s\in (0,1)}\inf_{\mu>0}\bigg[C_{s,1}(\rho_a)\sup_{i=1,2}\sum_{(\ga,j)\neq (0,i)} |\langle\delta_{0,i},\Hi_{\textup{Hal}}\,\delta_{\ga,j}\rangle|^s \eu^{\mu|\ga-\xi|}\bigg]^{\frac{1}{s}} \\
    &=\inf_{s\in (0,1)}\inf_{\mu>0}\bigg[C_{s,1}(\rho_a)\Big(|t_1|^s+(4|t_2|^s+2|t_1|^s)\eu^\mu+2|t_2|^s\eu^{\sqrt{2}\mu}\Big)\bigg]^{\frac{1}{s}} \\
    &=\inf_{s\in (0,1)}\bigg[C_{s,1}(\rho_a)\Big(3|t_1|^s+6|t_2|^s\Big)\bigg]^{\frac{1}{s}}=t_1\inf_{s\in (0,1)}\bigg[C_{s,1}(\rho_a)\Big(3+6(3\sqrt{3})^{-s}\Big)\bigg]^{\frac{1}{s}}\\
    &=\frac{2t_1}{\sqrt{2\pi}\norm{f}_{L^1[-a,a]}}\inf_{s\in (0,1)}\bigg[\frac{3+6(3\sqrt{3})^{-s}}{1-s}\bigg]^{\frac{1}{s}}<\frac{39.98 \, t_1}{\norm{f}_{L^1[-a,a]}}.
\end{align*}

In particular, if $a>a_0$ we have $\norm{f}_{L^1[-a,a]}>39.98/40$ and so we obtain
\begin{equation*}
    \la_{\rho_a}(\Hi_{\textup{Hal}})<  40\, t_1.
\end{equation*}
Therefore, for all $\la>\la_{\rho_a}(\Hi_{\textup{Hal}})$,  the operator $\Hi_{\textup{Hal},\la,\omega}$ contains $E=0$ in its almost-sure spectrum and exhibits DL there.

Combining the results on localization, we have that there exists $\la^*>0$ satisfying 
\begin{equation*}
    4\cdot 10^{-31}< \frac{\la^*}{t_1}<40   
\end{equation*}
such that $\Hi_{\textup{Hal},\la^*,\omega}$ exhibits DD in $E=0$.

\goodbreak

\appendix

\renewcommand{\thesubsection}{\Alph{subsection}}
\counterwithin{equation}{subsection}
\counterwithin{thm}{subsection}

\section*{Appendix}

\subsection{Deterministic Chern character}\label{appendix deterministic chern}

\begin{proof}[Proof of Proposition \ref{deterministic chern}]
    In order to avoid heavy notation, we will omit the dependence of the projection on $E$ and $\la$ .

    Recall that $H_{\omega}$ is an ergodic operator, since there exists $\set{\tau_\ga}_{\ga\in\Ga}$ an ergodic family of measure preserving automorphisms of $\Omega$, given by $(\tau_\ga\omega)_{\xi,i}=\omega_{\xi+\ga,i}$ for $(\xi,i)\in \Ga\times \set{1,\ldots,n}$, such that
    \begin{equation*}
        T_\ga^*H_{\omega} T_\ga=H_0+\la T_\ga^*V_\omega T_\ga=H_{\tau_\ga\omega} \quad \textup{for all} \ \ga\in\Ga.
    \end{equation*}

    By functional calculus, we have that $P_{\omega}=\chi_{(-\infty,E]}(H_{\omega})$ is an ergodic operator too:
    \begin{equation}\label{ergodicity projection}
        T_\ga^*P_\omega T_\ga=P_{\tau_\ga\omega} \quad \textup{for all} \ \ga\in\Ga.
    \end{equation}
    As a consequence, the operator $[X_j,P_\omega]$ is also ergodic for $j=1,2$, since
    \begin{equation}\label{ergodicity commutator}
        T_\ga^*[X_j,P_\omega]T_\ga=[T_\ga^*X_jT_\ga,T_\ga^*P_\omega T_\ga]=[X_j+\ga_j\Id,P_{\tau_\ga\omega}]=[X_j,P_{\tau_\ga\omega}].
    \end{equation}

    1. Since the box $\Lambda_L$ contains finitely many points, it is enough to show that there exists
    $\Omega_0\subset \Omega$ with $\mathbb{P}(\Omega_0)=1$ such that
    \begin{equation}\label{bound-C}
        |\Tr_{\C^n}\mathfrak{C}_{P_{\omega}}(\ga,\ga)|<\infty \quad \textup{for all} \ \omega\in \Omega_0,\ga\in\Ga.
    \end{equation}
    Recalling that $\mathfrak{C}_{P_\omega}=P_\omega[[X_1,P_\omega],[X_2,P_\omega]]P_\omega$, we can write
    \begin{align}\label{chern cancellation}
        \mathfrak{C}_{P_\omega}(\ga,\ga)
        &=(P_\omega X_1 P_\omega X_2 P_\omega^2-P_\omega^2X_1X_2P_\omega^2-P_\omega X_1 P_\omega^2X_2P_\omega+P_\omega^2 X_1 P_\omega X_2 P_\omega)(\ga,\ga)\nonumber \\
        &-(P_\omega X_2 P_\omega X_1 P_\omega^2-P_\omega^2X_2X_1P_\omega^2-P_\omega X_2 P_\omega^2 X_1 P_\omega+P_\omega^2 X_2 P_\omega X_1 P_\omega)(\ga,\ga).
    \end{align}
    We will show that
    \begin{equation}\label{finite expecation cancellation}
        \E\Big[\Tr_{\C^n}|(P_\omega X_j^mP_\omega^{m'}X_k^{m''}P_\omega)(\ga,\ga)|\Big]<\infty
    \end{equation}
    for all $\ga\in\Ga, m,m',m''\in \set{0,1}, j,k\in\set{1,2}$.

    Indeed, given $\ga\in\Ga$, $m,m',m''\in \set{0,1}$ and $j,k\in\set{1,2}$, we have
    \begin{align*}
        \E\Big[\Tr_{\C^n}|(P_\omega X_j^mP_\omega^{m'}X_k^{m''}&P_\omega)(\ga,\ga)|\Big]=\E\bigg[\Tr_{\C^n}\bigg|\sum_{\xi,\zeta\in\Ga}P_\omega(\ga,\xi)\xi_j^m P^{m'}_\omega(\xi,\zeta) \zeta_k^{m''}P_\omega(\zeta,\ga)\bigg| \bigg]\\
        &\leq n\,\E\bigg[\sum_{\xi,\zeta\in\Ga}|\xi_j|^m|\zeta_k|^{m''}\Big(\Tr_{\C^n}|P_\omega(\ga,\xi)|^2 \Tr_{\C^n}|P_\omega(\zeta,\ga)|^2\Big)^{\frac{1}{2}}\bigg]
    \end{align*}
    where we used $\Tr(|ABC|)\leq \Tr(|A|)\norm{B}\norm{C}\leq n (\Tr|A|^2)^{\frac{1}{2}}(\Tr|C|^2)^{\frac{1}{2}}$ which holds for all matrices $A,B,C\in M_n(\C)$ with $\norm{B}\leq 1$. Using H\"{o}lder inequality for the expectation, the ergodicity property of the projection \eqref{ergodicity projection} and a change of variables, we get
    \begin{align*}
        &\E\Big[\Tr_{\C^n}|(P_\omega X_j^mP_\omega^{m'}X_k^{m''}P_\omega)(\ga,\ga)|\Big]\leq n\bigg(\sum_{\xi\in\Ga}(1+|\xi|)\E\Big[\Tr_{\C^n}|P_\omega(\ga,\xi)|^2\Big]^{\frac{1}{2}}\bigg)^2\\
        &\leq n \Bigg(\sum_{\xi\in \Ga}|\xi|\E\Big[\Tr_{\C^n}|P_\omega(0,\xi)|^2\Big]^{\frac{1}{2}}+(1+|\ga|)\sum_{\xi\in \Ga}\max\set{1,|\xi|}\E\Big[\Tr_{\C^n}|P_\omega(0,\xi)|^2\Big]^{\frac{1}{2}}\Bigg)^2
    \end{align*}
    which is finite by assumption. As a consequence, for every $\ga\in\Ga$,
    \begin{equation}\label{finite expectation}
        \E\Big[\Tr_{\C^n}\mathfrak{C}_{P_\omega}(\ga,\ga)\Big]<\infty .  
    \end{equation}
    Therefore, there exists $\Omega(\ga)\subset \Omega$ with $\mathbb{P}(\Omega(\ga))=1$ such that   $ |\Tr_{\C^n}\mathfrak{C}_{P_\omega}(\ga,\ga)|<\infty$   for all $\omega \in \Omega(\ga)$. Taking  $\Omega_0=\bigcap_{\ga\in\Ga}\Omega(\ga)$ we have that $\mathbb{P}(\Omega_0)=1$ and \eqref{bound-C} holds for this choice of $\Omega_0$.

    2. Consider $\tilde{\Omega}_0=\Omega_0\cap \bigcap_{\ga\in\Ga}\tau_\ga(\Omega(\ga=0))$ with $\mathbb{P}(\tilde{\Omega}_0)=1$. Let $\ga\in\Ga$ and $\omega\in \tilde{\Omega}_0$ then, using the ergodicity property of the projection \eqref{ergodicity projection} and of the commutator \eqref{ergodicity commutator}, we have that
    \begin{align*}
        \Tr_{\C^n}\mathfrak{C}_{P_\omega}(\ga,\ga)&=\sum_{i=1}^n\langle\delta_{\ga,i},\mathfrak{C}_{P_\omega}\delta_{\ga,i}\rangle \\
        &=\sum_{i=1}^n\langle\delta_{0,i},T_\ga^*P_\omega[[X_1,P_\omega],[X_2,P_\omega]]P_\omega T_\ga\delta_{0,i}\rangle \nonumber \\
        &=\sum_{i=1}^n\langle\delta_{0,i},P_{\tau_\ga\omega}[[X_1,P_{\tau_\ga\omega}],[X_2,P_{\tau_\ga\omega}]]P_{\tau_\ga\omega}\delta_{0,i}\rangle =\Tr_{\C^n}\mathfrak{C}_{P_{\tau_\ga\omega}}(0,0).
    \end{align*}
    Therefore, for all $\omega\in\tilde{\Omega}_0$, we get
    \begin{align*}
        C(P_\omega)&=\lim_{\substack{L\to\infty \\ L\in\N}}\frac{2\pi\iu}{|\La_L|}\sum_{\ga\in\La_L}\Tr_{\C^n}\mathfrak{C}_{P_\omega}(\ga,\ga)=\lim_{\substack{L\to\infty \\ L\in\N}}\frac{2\pi\iu}{|\La_L|}\sum_{\ga\in\La_L}\Tr_{\C^n}\mathfrak{C}_{P_{\tau_\ga\omega}}(0,0). 
    \end{align*}
    Since $\E\Big[\Tr_{\C^n}\mathfrak{C}_{P_\omega}(0,0)\Big]<\infty$ by \eqref{finite expectation}, we are in position to apply Birkhoff's ergodic theorem, obtaining
    \begin{equation*}
        C(P_\omega)=2\pi\iu\,\E\Big[\Tr_{\C^n}\mathfrak{C}_{P_\omega}(0,0)\Big] \quad \textup{for all} \ \omega\in\tilde{\Omega}_0.
    \end{equation*} 
    Moreover, by \eqref{chern cancellation} and \eqref{finite expecation cancellation}, we have that
    \begin{align*}
        \E\Big[\Tr_{\C^n}\mathfrak{C}_{P_\omega}(0,0)\Big]& =\E\Big[\Tr_{\C^n}(P_\omega X_1P_\omega X_2 P_\omega-P_\omega X_2P_\omega X_1P_\omega)(0,0)\Big]\\
        &=\sum_{\ga,\xi\in\Ga}\E\Big[\Tr_{\C^n}\big(P_{\omega}(0,\ga)P_{\omega}(\ga,\xi)P_{\omega}(\xi,0)\big)\Big](\ga\wedge \xi).
    \end{align*}
\end{proof}

\subsection{Initial length scale estimate for band edge localization}\label{Appendix band edge}

In this appendix we prove Theorem~\ref{inital length estimate}, namely DL at the (internal) band edges of the spectrum, by showing that the initial length scale estimate holds, following the argument in \cite[Thm. 3']{Figotin Klein 1994}. 
   
For $L\in\N^*$, consider the box $\La_L\subset \Ga$ given in \eqref{box} and define, for $i\in\set{1,\ldots,n}$, the periodic extension of the random variables $\omega_{\ga,i}$ restricted to the box $\La_L$, which is given by 
\begin{gather*}
    \omega^L_{\ga,i}=\omega_{\ga,i} \ \ \ga\in\La_L, \quad \omega^L_{\ga+L\xi,i}=\omega^L_{\ga,i} \ \ \ga,\xi\in\Ga \nonumber.
\end{gather*}
Correspondingly, we define the associated $L\Ga$-periodic random potential
\begin{gather*}
    V^L_{\omega}=\sum_{\ga\in\Ga}\sum_{i=1}^{n}\omega^L_{\ga,i}\ket{\delta_{\ga,i}}\bra{\delta_{\ga,i}}
\end{gather*}
and the corresponding $L\Gamma$-periodization of $H_{\la,\omega}$, as the Hamiltonian 
\begin{equation}\label{periodic extension}
    H_{\la,\omega}^L=H_0+\la V_\omega^L.
\end{equation}

Assuming that $H_0$ satisfies \ref{item: P2}, we define the periodic restriction of $H_{\la,\omega}^L$ to $\La_L$ as
\begin{gather*}
    \mathring{H}^L_{\la,\omega,\La_L}(\ga,\xi)=\sum_{\zeta\in\Ga}H_{\la,\omega}^L(\ga,\xi+L\zeta) \quad \ga, \xi\in\Ga \\
    \mathring{H}^L_{\la,\omega,\La_L}=\Big(\mathring{H}^L_{\la,\omega.\La_L}(\ga,\xi)\Big)_{\ga,\xi\in\La_L}\in M_{\La_L}(M_n(\C)).
\end{gather*}

The next lemma summarizes the properties of $\mathring{H}^L_{\la,\omega,\La_L}$ and its relation to the operator $H_{\la,\omega,\La_{L}}$, the restriction of $H_{\la,\omega}$ to $\La_{L}$ with simple boundary conditions, see \eqref{defn:simple-bc}. In the following, we will say that $H_0$ is $L\Ga$-periodic if $T_\ga H_0T_\ga^*=H_0$ for all $\ga\in L\Ga$, where $T_\ga\in \mathcal{B}(\ell^2(\Ga;\C^n))$ is the ordinary translation operator defined in \ref{item: P1} for $B=0$.
\begin{lemma}\label{property-per}
    Assume that $H_0$ is $L\Gamma$-periodic and \ref{item: P2} holds. Then $H_{\la,\omega}^{L}$ is $L\Ga$-periodic, and the periodic restriction to $\La_{L}$ satisfies the following properties:
    \begin{enumerate}
        \item\label{point 1 per} \textup{\cite[Lemma 2.2]{Figotin Klein 1994}} $\mathring{H}^{L}_{\la,\omega,\La_{L}}$ is a self-adjoint matrix;
        \item\label{point 2 per} \textup{\cite[Lemma 2.11]{Figotin Klein 1994}} Let $L>r-1$ so that $\La_{L}\setminus\partial\La_{L}^{(r)}\neq \emptyset$. If either one of the points $\ga$ or $\xi$ is well inside the box $\La_L$, that is, if $\ga\in \La_{L}\setminus \partial\La_{L}^{(r)}$ or $\xi\in\La_{L}\setminus\partial\La_{L}^{(r)}$, with $\partial\La_{L}^{(r)}$ given in \eqref{boundary-box}, then $\mathring{H}^{L}_{\la,\omega,\La_{L}}(\ga,\xi)=H^{L}_{\la,\omega,\La_{L}}(\ga,\xi)=H_{\la,\omega,\La_{L}}(\ga,\xi)$;
        \item\label{point 3 per} \textup{\cite[Thm. 4]{Figotin Klein 1994}} $\sigma\Big(\mathring{H}^{L}_{\la,\omega,\La_{L}}\Big)\subset \sigma(H_{\la,\omega}^{L})$.
    \end{enumerate}
\end{lemma}
We will use the Combes-Thomas estimate, as stated \eg in \cite[Thm. 10.5]{Aizenman Warzel 2016}.
    
\begin{prop}[\protect{Combes-Thomas estimate}]\label{combes thomas}
    Let $H\in\mathcal{B}(\ell^2(\Ga;\C^n))$ be a self-adjoint operator such that for some $\alpha>0$ it holds that
    \begin{equation*}
        S_{\alpha}=\sup_{\ga\in \Ga}\sum_{\xi\in\Ga}\norm{H(\ga,\xi)}(\eu^{\alpha |\ga-\xi|}-1)<\infty
    \end{equation*}
    where $\norm{\cdot}$ is the operator norm on $M_n(\C)$. Let $z\notin\sigma(H_0)$ and denote $\Delta=\textup{dist}(\sigma(H),z)>0$. Define $G(\ga,\xi;z)=(H-z)^{-1}(\ga,\xi)$. Then
    \begin{align}
        &\textup{if} \ \Delta\geq  2 S_\alpha\implies\norm{G(\ga,\xi;z)}\leq \frac{2}{\Delta}\eu^{-\alpha|\ga-\xi|}  \label{ct1}\\
        &  \textup{if} \ \Delta\leq 2 S_\alpha\implies \norm{G(\ga,\xi;z)}\leq\frac{2}{\Delta}\eu^{-\frac{\alpha\Delta}{2 S_\alpha} |\ga-\xi|}.
    \end{align}
\end{prop}

We will also use a version of the Combes-Thomas estimate suitable for the periodic restrictions introduced above, proved in \cite[Lemma 2.15]{Figotin Klein 1994}. Note that, by Lemma \ref{property-per}-\eqref{point 3 per}, if $z\notin\sigma(H_{\la,\omega}^{L})$ then $z\notin\sigma(\mathring{H}^{L}_{\la,\omega,\La_{L}})$, and thus the following resolvent is well defined
\begin{equation*}
    \mathring{G}^{L}_{\la,\omega,\La_{L}}(\ga,\xi;z)=(\mathring{H}^{L}_{\la,\omega,\La_{L}}-z)^{-1}(\ga,\xi).
\end{equation*}
\begin{thm}[\protect{Combes-Thomas estimate for periodic restrictions}]\label{combes thomas periodic}
    Assume that $H_0$ is $L\Ga$-periodic and \ref{item: P2} holds. Let $z\notin\sigma(H_{\la,\omega}^{L})$ and denote $\Delta=\textup{dist}(\sigma(H_{\la,\omega}^{L}),z)>0$. Then, given $\alpha>0$, for all $\ga,\xi\in\La_{L}$ we have that
    \begin{equation*}
        \norm{\mathring{G}^{L}_{\la,\omega,\La_{L}}(\ga,\xi;z)}\leq \frac{2}{\Delta}\bigg(1+\frac{2}{1-\eu^{-\frac{\alpha }{\sqrt{2}}\min\set{1,\frac{\Delta}{2 S_\alpha}}L}}\bigg)^2\eu^{-\frac{\alpha }{\sqrt{2}}\min\set{1,\frac{\Delta}{2 S_\alpha}} \textup{dist}(\ga-\xi,L\Ga)}
    \end{equation*}
    where $\textup{dist}(\ga,L\Ga)=\min_{\xi\in\Ga}|\ga-L\xi|$ for $\ga\in\Ga$.
\end{thm}

Lastly, we recall a Wegner estimate that is a consequence of the continuity of the IDS shown in \cite[Thm. 1.3]{Krishna 2007}.

\begin{prop}[\protect{Wegner estimate}]
    Assume that $H_{\la,\omega}$ satisfies \ref{item: P1}, \ref{item: P4} and \ref{item: R1}. Then we have that for any $E\in \mathbb R$,
    \begin{equation*}\label{wegner estimate}
        \mathbb{P}(\dist (E,\sigma(H_{\la,\omega,\La_L}))<\eps)\leq 4\pi nC_\tau(\rho)|\La_L|\frac{\eps^\tau}{\la^\tau}
    \end{equation*}
    for all $L\in \N^*$, $\eps>0$ and $\la>0.$
\end{prop}

In order to prove Theorem \ref{inital length estimate}, we show the following (slightly) stronger result, in which we only deal with internal band edges. For external bands edges we can repeat the same proof, replacing periodic restrictions with simple boundary restrictions, without any commensurability assumption on the magnetic flux.

\begin{thm}\label{appendix band edge thm}
    Assume that $H_{\la,\omega}$ satisfies hypothesis \ref{item: P1}, with $B=\frac{2\pi p}{q}$ for some $p\in\N,q\in\mathbb N^*$, \ref{item: P2}-\ref {item: P4}, \ref{item: R1} and \ref{item: R2}. Then, for all $\la<\frac{|\G_i|}{a+b}$ with $i\in\set{1,\ldots, N-1}$, for all $\theta>\frac{2}{\tau}$ it holds that
    \begin{align}
        &\lim_{\substack{L\to\infty\\ L\in 3\N^*+4r}}\mathbb{P}(\La_{qL} \ \textup{is} \ (\omega,\theta,\beta_i+b\la)\textup{-suitable with range} \ qr)=1 \label{input 1} \\
        &\lim_{\substack{L\to\infty\\L\in 3\N^*+4r}}\mathbb{P}(\La_{qL} \ \textup{is} \ (\omega,\theta,\alpha_{i+1}-a\la)\textup{-suitable with range} \ qr)=1. \label{input 2}
    \end{align}
\end{thm}

\begin{proof}
    We only prove \eqref{input 1} for the right spectral band edge $E_0=\beta_i+b\la$, since the same proof applies to the case of the left spectral band edge. Note that, since \ref{item: P1} hold with $B=\frac{2\pi p}{q}$, we have that $H_0$ is $q\Ga$-periodic and so we can consider  $H_{\la,\omega}^{qL}$, the $qL\Gamma$-periodization of $H_{\la,\omega}$. Note also that $H_0$ satisfies \ref{item: P2} with radius $qr$. This motivates the use of boxes $\La_{qL}$. We will  consider the restriction of $H_{\la,\omega}$ to $\La_{qL}$ with simple boundary conditions, and show that for all $\theta>\frac{2}{\tau}$, the event 
    \begin{equation*}
    \left\{ \omega \in \Omega:\,\, \norm{G_{\la,\omega,\La_{qL}}(\ga,\xi;E_0)}\leq \frac{1}{(qL)^{\theta}}   \mbox{ for all }\,\ga\in\La_{\frac{qL+2qr}{3}},\xi\in\partial\La_{qL}^{(qr)}  \right\}
    \end{equation*}
    has a probability that tends to 1 as $L$ tends to infinity. Let $L\in 3\N^*+4r$, and consider 
    \begin{equation}\label{decomp}
            H_{\la,\omega,\La_{qL}}=H^{qL}_{\la,\omega,\La_{qL}}=\mathring{H}^{qL}_{\la,\omega,\La_{qL}}+R_{\La_{qL}}
    \end{equation}
    where $R_{\La_{qL}}$ is a matrix representing the difference between restrictions with periodic and simple boundary conditions, that satisfies:
    \begin{enumerate}
        \item if $\ga\in \La_{qL}\setminus \partial\La_{qL}^{(qr)}$ or $\xi\in\La_{qL}\setminus\partial\La_{qL}^{(qr)}$ then $R_{\La_{qL}}(\ga,\xi)=0$ by Lemma \ref{property-per}-\eqref{point 2 per};
        \item $\norm{R_{\La_{qL}}(\ga,\xi)}\leq 8\norm{H_0}$ for all $\ga,\xi\in\partial\La^{(qr)}_{qL}$.
    \end{enumerate}
    
    Given $\alpha>0$ and $\theta>\frac{2}{\tau}$ define
    \begin{equation}\label{delta}
        \Delta_{qL}=\frac{2\sqrt 2 S_\alpha(3\theta+5)}{\alpha\la}\frac{8\log (qL)}{qL},
    \end{equation}
    with $S_\alpha$ given in Proposition \ref{combes thomas}. Let $qL$ be large enough, depending on $\alpha,H_0,a,b,\theta,\la$, such that
    \begin{equation}\label{restriction 2}
        \Delta_{qL}<b+a
    \end{equation}
    and define the event
    \begin{equation}\label{event R2}
        \mathcal{E}_{qL}=\{\omega\in\Omega: \omega_{\ga,j}\in [-a,b-\Delta_{qL}), \ \textup{for all} \ \ga\in\La_{qL},j\in\{1,\dots,n\}\}.
    \end{equation}
    If $\omega\in \mathcal{E}_L$ we have, by \cite[Thm. V.4.10]{Kato 1995}, that $\sigma(H^{qL}_{\la,\omega})\subset\sigma(H_0)+\la[-a,b-\Delta_{qL}]$. So, taking $qL$ large enough, depending on $\alpha,H_0,a,b,\theta,\la$, so that 
    \begin{equation}\label{restriction 3}
        \la\Delta_{qL}\leq |\G_i(\la)|=|\G_i|-\la(a+b)
    \end{equation}
    we obtain
    \begin{equation}\label{dist-delta}
        \la\Delta_{qL}\leq \textup{dist}(\sigma(H_{\la,\omega}^{qL}),E_0).
    \end{equation}
    Therefore, for $\omega\in \mathcal{E}_L$, $E_0\notin \sigma(H_{\la,\omega}^{qL})$, and, by Lemma \ref{property-per}-(3), $E_0\notin \sigma(\mathring{H}^{qL}_{\la,\omega,\La_{qL}})$. 
         
    Using the resolvent identity we get, from \eqref{decomp}, for all $\gamma,\xi\in\La_{qL}$
    \begin{align*}
        &G_{\la,\omega,\La_{qL}}(E_0)=\mathring{G}^{qL}_{\la,\omega,\La_{qL}}(E_0)+\mathring{G}^{qL}_{\la,\omega,\La_{qL}}(E_0)R_{\La_{qL}}G_{\la,\omega,\La_{qL}}(E_0) \\
        &G_{\la,\omega,\La_{qL}}(\ga,\xi;E_0)=\mathring{G}^{qL}_{\la,\omega,\La_{qL}}(\ga,\xi;E_0)\\
        &\hspace{2.6cm}+\sum_{\zeta,\zeta'\in\partial\La_{qL}^{(qr)}}\mathring{G}^{qL}_{\la,\omega,\La_{qL}}(\ga,\zeta;E_0)R_{\La_{qL}}(\zeta,\zeta')G_{\la,\omega,\La_{qL}}(\zeta',\xi;E_0). 
    \end{align*}
    To estimate the matrix elements of $\mathring{G}^{qL}_{\la,\omega,\La_{qL}}$ in the right-hand side, we can use Theorem \ref{combes thomas periodic} with $\Delta_{qL}$ replacing $\Delta$ there, since we have \eqref{dist-delta}. Then, for all  $\ga,\xi\in\La_{qL}$, and $qL$ large enough, depending on $\alpha,H_0,\theta$, such that 
    \begin{equation}\label{restriction 4}
        \lambda \Delta_{qL}\leq 2S_\alpha
    \end{equation} 
    we obtain by the Combes-Thomas estimate for the periodic restriction
    \begin{equation*}
        \norm{\mathring{G}^{qL}_{\la,\omega,\La_{qL}}(\ga,\xi;E_0)}\leq \frac{2}{\la\Delta_{qL}}\bigg(1+\frac{2}{1-\eu^{-\frac{\alpha\la\Delta_{qL}}{ 2\sqrt 2S_\alpha}qL}}\bigg)^2\eu^{-\frac{\alpha\la\Delta_{qL}}{ 2\sqrt 2S_\alpha} \textup{dist}(\ga-\xi,qL\Ga)}.
    \end{equation*}
    Taking $qL$ large enough, depending on $\theta$, such that
    \begin{equation}\label{restriction 5}
        \eu^{-\frac{\alpha \la\Delta_{qL} qL}{ 2\sqrt 2S_\alpha}}\leq \frac{1}{3}
    \end{equation}
    we obtain for all $\ga,\xi\in\La_L$ that
    \begin{equation*}
        \norm{\mathring{G}^{qL}_{\la,\omega,\La_{qL}}(\ga,\xi;E_0)}\leq \frac{2^5}{\la\Delta_{qL}}\eu^{-\frac{\alpha\la\Delta_{qL}}{ 2\sqrt 2S_\alpha}\textup{dist}(\ga-\xi,qL\Ga)}.
    \end{equation*}
       
    Since we are interested in the suitable box condition given in Definition \ref{definition good box}, we are only considering $\ga\in\La_{\frac{qL+2qr}{3}},\xi\in\partial\La_{qL}^{(qr)}$, for which
    \begin{align*}
        \textup{dist}(\ga-\xi,qL\Ga)&=\min_{\eta\in\Ga}|\ga-\xi-qL\eta|\geq \min_{\eta\in\Ga}|\ga-\xi-qL\eta|_\infty \\
        &\geq \min_{\eta\in\Ga}\Big||\ga-\xi|_\infty-qL|\eta|_\infty\Big|=\min\set{|\ga-\xi|_\infty,qL-|\ga-\xi|_\infty}\geq  \frac{qL}{3}-\frac{4qr}{3}.
    \end{align*}
    Thus for all $\ga\in \La_{\frac{qL+2qr}{3}},\xi\in\partial\La_{qL}^{(qr)}$ we obtain
    \begin{equation*}
        \norm{G_{\la,\omega,\La_{qL}}(\ga,\xi;E_0)}\leq \frac{2^5}{\la\Delta_{qL}}\eu^{-\frac{\alpha\la\Delta_{qL}}{ 2\sqrt 2S_\alpha}\left(\frac{qL}{3}-\frac{4qr}{3}\right)}\Big(1+8\norm{H_0}(qL)^{4}\norm{G_{\la,\omega,\La_{qL}}(E_0)}\Big).
    \end{equation*}
    Define the event
    \begin{equation*}
        \mathcal{W}_{qL}(E_0)=\set{\omega\in\Omega:\norm{G_{\la,\omega.\La_{qL}}(E_0)}\leq (qL)^\theta}.
    \end{equation*}
    Let $qL\geq \eu$ and $qL\geq \frac{32}{5}qr$, so that $\frac{qL}{3}-\frac{4qr}{3}\geq \frac{qL}{8}$. If $\omega\in\mathcal{E}_{qL}\cap\mathcal{W}_{qL}(E_0)$, recalling the definition of $\Delta_{qL}$ in \eqref{delta}, we get
    \begin{align}
        &\norm{G_{\la,\omega,\La_{qL}}(\ga,\xi;E_0)} \nonumber\\
        &\leq \frac{8\sqrt 2\alpha}{S_\alpha(3\theta+5)}\frac{qL}{8\log(q L)}\eu^{-(3\theta+5)\frac{8\log (qL)}{qL}\left(\frac{qL}{3}-\frac{4qr}{3}\right)}\Big(1+8\norm{H_0}(qL)^{4+\theta}\Big) \nonumber\\
        &\leq\frac{\sqrt 2\alpha}{S_\alpha(3\theta+5)} \frac{qL}{(qL)^{3\theta+5}}\Big(1+8\norm{H_0}(qL)^{4+\theta}\Big) \nonumber\\
        &=\frac{\sqrt 2\alpha(1+8\norm{H_0})}{S_\alpha(3\theta+5)}\frac{(qL)^{5+\theta}}{(qL)^{3\theta+5}}=\frac{\sqrt 2\alpha(1+8\norm{H_0})}{S_\alpha(3\theta+5)}\frac{1}{(qL)^{2\theta}}\leq \frac{1}{(qL)^{\theta}} \label{restriction 6}
    \end{align}
    which holds for $qL$ large enough depending on $\alpha,H_0,r,\theta$. Therefore the probability of the box $\La_{qL}$ not being $(\theta,E_0)$-suitable can be estimated as
    \begin{align}\label{bound-G}
        \mathbb{P}\bigg(\norm{G_{\la,\omega,\La_{qL}}(\ga,\xi;E_0)}>\frac{1}{(qL)^\theta} \ \text{for some} \ \ga\in\La_{\frac{qL+2qr}{3}},\xi\in\partial\La_{qL}^{(qr)}\bigg)\leq\mathbb{P}(\mathcal{E}_{qL}^{\complement})+\mathbb{P}(\mathcal{W}_{qL}(E_0)^\complement).
    \end{align}
    We estimate the first term in the right-hand side using hypothesis \ref{item: R2}, to get
    \begin{align}\label{restriction 7}
        \mathbb{P}(\mathcal{E}_{qL}^\complement)&=1-\mathbb{P}(\mathcal{E}_{qL})=1-n(qL)^2\rho([-a,b-\Delta_{qL})) \nonumber \\
        &\leq n(qL)^2(1-\rho([-a,b-\Delta_{qL}))=n(qL)^2\rho([b-\Delta_{qL},b]) \nonumber \\
        &\leq Cn\bigg(\frac{2\sqrt 2 S_\alpha(3\theta+5)}{\alpha\la}\bigg)^\beta(qL)^2\bigg(\frac{8\log (qL)}{qL}\bigg)^\beta
    \end{align}
    which tends to zero as $L$ tends to $\infty$,  since $\beta>2$. For the second term instead we use the Wegner estimate \eqref{wegner estimate}
    \begin{equation}\label{restriction 8}
        \mathbb{P}(\mathcal{W}_{qL}(E_0)^\complement)=\mathbb{P}\bigg(\text{dist}(\sigma(H_{\la,\omega,\La_{qL}}),E_0)<\frac{1}{(qL)^\theta}\bigg)\leq \frac{4\pi nC_\tau(\rho)}{\la^\tau} \frac{(qL)^{2}}{(qL)^{\tau\theta}}
    \end{equation}
    which tends to zero as $L$ tends to $\infty$,  since $\theta>\frac{2}{\tau}$. Inserting \eqref{restriction 7} and \eqref{restriction 8} in \eqref{bound-G} yields \eqref{input 1}.
\end{proof}

\subsubsection{Region of band edge localization}\label{band edge region}

It was shown in \cite[Thm. 4.2]{Germinet Klein 2004} that, given $\theta>\frac{2}{\tau}$, the verification of the initial length scale estimate, appearing in \eqref{input 0 or}-\eqref{input 2 or}, at an energy level $E$ is equivalent to the existence of a finite scale $L^{(0)}=L_\theta^{(0)}(E)>0$ such that
\begin{equation}\label{841}
    \mathbb{P}(\La_{qL} \ \textup{is} \ (\omega,\theta,E)\textup{-suitable with range} \ qr)>1-p_0
\end{equation}
for some $L\geq L^{(0)}$ with $L\in 3\N^*+4r$, and for some $p_0\in (0,1)$ independent of the length scale $L^{(0)}$, see \cite[Thm. 5.1]{Germinet Klein 2001}, \cite[Lemma 36]{Figotin Klein 1996}. 

Given $\la\Delta_{qL}$ as defined in \eqref{delta} consider the interval
\begin{equation*}
    I_{\la,qL}=\left[\beta_i+b\la-\frac{\la\Delta_{qL}}{2}, \beta_i+b\la\right].
\end{equation*}
Let $\Delta_{qL}<b+a$ and consider $\omega\in\mathcal{E}_{qL} $, as defined in \eqref{event R2}. If $\frac{\la\Delta_{qL}}{2}\leq |\G_i|-\la(a+b)$ we get
\begin{equation*}
    \sup_{E\in I_{\la,qL}}\dist (E,\sigma(H_{\la,\omega}^{qL}))\geq \frac{\la \Delta_{qL}}{2}.
\end{equation*}
Since the last estimate is uniform in $E\in I_{\la,L}$, we can repeat the same proof of Proposition \ref{appendix band edge thm}, replacing $\la\Delta_{qL}$ with $\frac{\la\Delta_{qL}}{2}$ and imposing \eqref{restriction 7}, \eqref{restriction 8} to be smaller than $\frac{p_0}{2}$, obtaining, for all $\theta>\frac{\tau}{2}$, the existence of a finite scale $L^{(0)}=L^{(0)}_{\alpha,H_0,a,b,\theta,\beta,\tau,\rho,n}(\la)>0$ such that
\begin{equation}\label{841 uniform}
    \sup_{E\in I_{\la,qL}}\mathbb{P}(\La_{qL} \ \textup{is} \ (\omega,\theta,E)\textup{-suitable with range} \ qr)> 1-p_0
\end{equation}
for all $L\geq L^{(0)}$ with $L\in 3\N^*+4r$. As a consequence, given $\la<\frac{|\G_i|}{a+b}$, we have that $H_{\la,\omega}$ exhibits DL in $I_{\la,qL^{(0)}}=\left[\beta_i+b\la-\frac{\la\Delta_{qL^{(0)}}}{2},\beta_i+b\la\right]$. Hence, in order to give more information on how the region of band edge localization depends on  the disorder parameter $\la$, we need to understand the explicit dependence of $qL^{(0)}$ on $\la$. To do so,  we have to keep track of all the restrictions on the value of $L$ that we imposed throughout the proof of Proposition \ref{appendix band edge thm}. 

These restrictions are given in \eqref{restriction 2}, \eqref{restriction 3}, \eqref{restriction 4}, \eqref{restriction 5}, \eqref{restriction 6}, \eqref{restriction 7}, \eqref{restriction 8} (notice that in \eqref{restriction 3}, \eqref{restriction 4}, \eqref{restriction 5}, \eqref{restriction 6} we have to replace $\la\Delta_{qL}$ with $\frac{\la\Delta_{qL}}{2}$). From these conditions, we can derive a lower bound for the value of $L$ in order for \eqref{841 uniform} to hold. To do so, let $\eps\in (0,1)$, $c_\eps>0$, $c_{\eps,\beta}>0$ be such that
\begin{equation*}
    \log x\leq c_\eps\,x^\eps, \quad \log x\leq c_{\eps,\beta}\,x^{\frac{\eps(\beta-2)}{\beta}}, \quad \textup{for all} \ x>0.
\end{equation*}
Then a lower bound for $L$ is given by the maximum of the following thresholds, coming respectively from \eqref{restriction 2}, \eqref{restriction 3}, \eqref{restriction 4}, \eqref{restriction 5}, \eqref{restriction 6}, \eqref{restriction 7}, \eqref{restriction 8}:
\begin{gather*}
    qL^{(1)}=qL^{(1)}_{\alpha,H_0,a,b,\theta,\eps}(\la)=\bigg(\frac{32\sqrt2S_\alpha (3\theta+5)c_\eps}{\alpha(a+b)}\bigg)^{\frac{1}{1-\eps}}\bigg(\frac{1}{\la}\bigg)^{\frac{1}{1-\eps}} \\
    qL^{(2)}=qL^{(2)}_{\alpha,H_0,a,b,\theta,\eps}(\la)=\bigg(\frac{8\sqrt2S_\alpha (3\theta+5)c_\eps}{\alpha(a+b)}\bigg)^{\frac{1}{1-\eps}}\left(\frac{1}{\frac{|\G_i|}{a+b}-\la}\right)^{\frac{1}{1-\eps}} \\
    qL^{(3)}=qL^{(3)}_{\alpha,H_0,\theta,\eps}=\bigg(\frac{4\sqrt2 (3\theta+5)c_\eps}{\alpha}\bigg)^{\frac{1}{1-\eps}},\quad qL^{(4)}=qL^{(4)}_\theta=3^{\frac{1}{4(3\theta+5)}} \\
    qL^{(5)}=qL^{(5)}_{\alpha,H_0,\theta}=\max\set{\eu,16qr,\bigg(\frac{2\sqrt{2} \alpha (1+8\norm{H_0}}{S_\alpha (3\theta+5)}\bigg)^{\frac{1}{\theta}}} \\
    qL^{(6)}=qL^{(6)}_{\alpha,H_0,\beta,\theta,n,\eps}(\la)=\left(2 C n p_0\bigg(\frac{16\sqrt{2}S_\alpha(3\theta+5)c_{\eps,\beta}}{\alpha}\bigg)^\beta \right)^{\frac{1}{(\beta-2)(1-\eps)}}\bigg(\frac{1}{\la}\bigg)^{\frac{\beta}{(\beta-2)(1-\eps)}} \\
    qL^{(7)}=qL^{(7)}_{\theta,\tau,\rho,n}(\la)=(8\pi n C_\tau(\rho)p_0)^{\frac{1}{\tau\theta-2}}\bigg(\frac{1}{\la}\bigg)^{\frac{\tau}{\tau\theta-2}}.
\end{gather*}

Let $K_\eps=K_{\alpha,H_0,a,b,\theta,\beta,\tau,\rho,n,\eps}$ be the maximum of all the numerical constants appearing in each $qL^{(i)}$, for $i=1,\ldots,7$, which do not depend on $\la$. Define 
\begin{align*}
    qL^{(0)}&=qL^{(0)}_{\alpha,H_0,a,b,\theta,\beta,\tau,\rho,n,\eps}(\la)\\
    &=K_\eps\max\set{1, \bigg(\frac{1}{\la}\bigg)^{\frac{1}{1-\eps}},\bigg(\frac{1}{\la}\bigg)^{\frac{\beta}{(\beta-2)(1-\eps)}},\bigg(\frac{1}{\la}\bigg)^{\frac{\tau}{\tau\theta-2}},\left(\frac{1}{\frac{|\G_i|}{a+b}-\la}\right)^{\frac{1}{1-\eps}}}.
\end{align*}
By definition we have $L^{(0)}\geq \max\set{L^{(1)},L^{(2)},L^{(3)},L^{(4)},L^{(5)},L^{(6)},L^{(7)}}$. We choose $\theta>\frac{2}{\tau}$ such that
\begin{equation*}
    \frac{\beta}{(\beta-2)(1-\eps)}=\frac{\tau}{\tau\theta-2}\iff \theta =\frac{2}{\tau}+\bigg(1-\frac{2}{\beta}\bigg)(1-\eps).
\end{equation*}
Since $\frac{\beta}{\beta-2}>1$ for all $\beta>2$, we have
    $\la^{-\frac{1}{1-\eps}}\leq \la^{-\frac{\beta}{(\beta-2)(1-\eps)}}$ for all $\la\leq 1$.
Hence it holds that
\begin{equation*}
    qL^{(0)}=K_\eps\max\set{1, \bigg(\frac{1}{\la}\bigg)^{\frac{\beta}{(\beta-2)(1-\eps)}},\left(\frac{1}{\frac{|\G_i|}{a+b}-\la}\right)^{\frac{1}{1-\eps}}}.
\end{equation*}
Recalling the expression \eqref{delta} of $\Delta_{qL}$ we obtain
\begin{align*}
    \la\Delta_{qL^{(0)}}&=\frac{16\sqrt2S_\alpha(3\theta+5)}{\alpha}\frac{\log(qL^{(0)})}{qL^{(0)}} \\
    &=\frac{16\sqrt2S_\alpha(3\theta+5)}{\alpha}\frac{\log\left(K_\eps\max\set{1,\la^{-\frac{\beta}{(\beta-2)(1-\eps)}},\left(\frac{|\G_i|}{a+b}-\la\right)^{-\frac{1}{1-\eps}}}\right)}{K_\eps\max\set{1,\la^{-\frac{\beta}{(\beta-2)(1-\eps)}},\left(\frac{|\G_i|}{a+b}-\la\right)^{-\frac{1}{1-\eps}}}} \\
    &\geq \frac{16\sqrt2S_\alpha(3\theta+5)}{\alpha}\frac{\log K_\eps}{K_\eps}\min\set{1,\la^{\frac{\beta}{(\beta-2)(1-\eps)}},\left(\frac{|\G_i|}{a+b}-\la\right)^{\frac{1}{1-\eps}}}=f(\la).
\end{align*}
Notice that the last inequality is meaningful because $\log K_\eps\geq1$, since $K_\eps\geq qL^{(5)}\geq \eu$.
So, if $\la<\frac{|\G_i|}{a+b}$, it holds that $H_{\la,\omega}$ exhibits DL in 
\begin{equation*}
    \left[\beta_i+b\la-\frac{f(\la)}{2},\beta_i+b\la\right]\subset \left[\beta_i+b\la-\frac{\la\Delta_{qL^{(0)}}}{2},\beta_i+b\la\right].
\end{equation*}
Therefore we obtain the following region of band edge localization
\begin{equation*}
    \set{(E,\la):\beta_i+b\la-\frac{D\log K_\eps}{K_\eps}\min\set{1,\la^{\frac{\beta}{(\beta-2)(1-\eps)}},\left(\frac{|\G_i|}{a+b}-\la\right)^{\frac{1}{1-\eps}}}\leq E\leq \beta_i+b\la}
\end{equation*}
where $D=\frac{16\sqrt{2}S_\alpha(3\theta+5)}{2\alpha}$. From the previous expression we see that, for sufficiently small value of $\la$, the dominating term in $f(\la)$ is $\la^{\frac{\beta}{(\beta-2)(1-\eps)}}$, while if $\la$ is close enough to $\frac{|\G_i|}{a+b}$ the dominating term is given by $\left(\frac{|\G_i|}{a+b}-\la\right)^{\frac{1}{1-\eps}}$. For intermediate values of $\la$ the function $f(\la)$ remains constant.

\bigskip
\flushleft
{\footnotesize
\begin{tabular}{ll}

(G. Panati)
&  \textsc{Dipartimento di Matematica, \virg{La Sapienza} Universit\`{a} di Roma} \\
&  Piazzale Aldo Moro 2, 00185 Rome, Italy \\
&  {E-mail address}: \href{mailto:panati@mat.uniroma1.it}{\texttt{panati@mat.uniroma1.it}} \\
\\
(C. Rojas-Molina)
&  \textsc{Laboratoire de Mathématiques AGM, UMR CNRS 8088} \\
& \textsc{CY Cergy Paris Université} \\
&  2 Avenue Adolphe Chauvin, 95300 Cergy, France \\
&  {E-mail address}: \href{mailto:crojasmo@cyu.fr}{\texttt{crojasmo@cyu.fr}} \\
\\
(V. Rossi)
&  \textsc{GSSI - Gran Sasso Science Institute} \\
&  Viale Francesco Crispi 7, 67100 L'Aquila, Italy \\
&  {E-mail address}: \href{mailto:vincenzo.rossi@gssi.it}{\texttt{vincenzo.rossi@gssi.it}} \\
\\

\end{tabular}
}

\end{document}